\documentclass[12pt,onecolumn,draftclsnofoot]{IEEEtran}

\usepackage{amsfonts}
\usepackage{cite}
\usepackage[top=1 in, bottom=1 in, left=1 in, right= 1 in]{geometry}

\usepackage[cmex10]{amsmath}
\usepackage{amssymb}
\usepackage{array}
\usepackage{chemarrow}
\usepackage{times,epsfig}
\usepackage{graphicx}
\usepackage{subcaption}
\usepackage{setspace}
\usepackage{times,epsfig}
\usepackage{bbm}
\usepackage{verbatim}
\usepackage{amsmath}
\usepackage{epstopdf}
\usepackage[font=small]{caption}
\newtheorem{remark}{\textbf{Remark}}
\newtheorem{corollary}{\textbf{Corollary}}
\newtheorem{definition}{\textbf{Definition}}

\newtheorem{theorem}{\textbf{Theorem}}
\newtheorem{lemma}{\textbf{Lemma}}

\makeatletter

\newcommand{\Rmnum}[1]{\expandafter\@slowromancap\romannumeral #1@}

\makeatother

\makeatletter
\g@addto@macro{\normalsize}{%
   \setlength{\abovedisplayskip}{3pt plus 2pt minus 2pt}
   \setlength{\abovedisplayshortskip}{3pt plus 2pt minus 2pt}
   \setlength{\belowdisplayskip}{3pt plus 2pt minus 2pt}
   \setlength{\belowdisplayshortskip}{3pt plus 2pt minus 2pt}
   \setlength{\textfloatsep}{10pt plus 2pt minus 2pt}
   }
\makeatother

\begin{document}
%
% --- Author Metadata here ---
%\conferenceinfo{WOODSTOCK}{'97 El Paso, Texas USA}
%\CopyrightYear{2007} % Allows default copyright year (20XX) to be over-ridden - IF NEED BE.
%\crdata{0-12345-67-8/90/01}  % Allows default copyright data (0-89791-88-6/97/05) to be over-ridden - IF NEED BE.
% --- End of Author Metadata ---

\title{Statistical Modeling and Probabilistic Analysis of Cellular Networks with Determinantal Point Processes}

\author{ Yingzhe Li, Fran\c{c}ois Baccelli, Harpreet S. Dhillon, Jeffrey G. Andrews \thanks{Y. Li, F. Baccelli and J. G. Andrews are with the Wireless Networking and Communications Group (WNCG), The University of Texas at Austin (email: yzli@utexas.edu, francois.baccelli@austin.utexas.edu, jandrews@ece.utexas.edu). H. S. Dhillon is with the Wireless@VT, Department of Electrical and Computer Engineering, Virginia Tech, Blacksburg, VA (email: hdhillon@vt.edu). Part of this paper will be presented in IEEE GLOBECOM 2014 in Austin, TX~\cite{li2014fitting}. Date revised: \today.}}

\maketitle

%-------------------------------------------------------------------------
\begin{abstract}
Although the Poisson point process (PPP) has been widely used to model base station (BS) locations in cellular networks, it is an idealized model that neglects the spatial correlation among BSs. The present paper proposes the use of determinantal point process (DPP) to take into account these correlations; in particular the repulsiveness among macro base station locations. DPPs are demonstrated to be analytically tractable by leveraging several unique computational properties. Specifically, we show that the empty space function, the nearest neighbor function, the mean interference and the signal-to-interference ratio (SIR) distribution have explicit analytical representations and can be numerically evaluated for cellular networks with DPP configured BSs. In addition, the modeling accuracy of DPPs is investigated by fitting three DPP models to real BS location data sets from two major U.S. cities. Using hypothesis testing for various performance metrics of interest, we show that these fitted DPPs are significantly more accurate than popular choices such as the PPP and the perturbed hexagonal grid model.
%Based on the K-function and coverage probability, these DPP examples are validated as accurate for modeling macro cellular BS deployments, which outperform both the PPP and the perturbed hexagonal grid model. 
%Different DPP models are compared in terms of modeling accuracy to real deployments and repulsiveness level. 
   
\end{abstract}


\section{Introduction}
%Wireless industry has experienced unprecedented developments during recent years, which greatly changes the quality of daily communications. The fast development of wireless industry also exerts more challenges to explore new models and analytical tools for modern wireless networks, in particular the cellular networks. In this paper, we will use statistical modeling approaches to evaluate the accuracy of using determinantal point processes to model macro base station distribution. 

Historically, cellular base stations have been modeled by the deterministic grid-based model, especially the hexagonal grid. However, the increasingly dense capacity-driven deployment of BSs, along with other topological and demographic factors, have made cellular BS deployments more organic and irregular. Therefore, random spatial models, in particular the PPP, have been widely adopted to analyze cellular networks using stochastic geometry~\cite{trac,Harpreet,dhillon2013load, mukherjee2012distribution,madhusudhanan2011multi,elsawy2013stochastic, heath2013hetnets,madhusudhanan2013modeling,dhillon2013downlink,kountouris2013antennas}. However, since no two macro base stations are deployed arbitrarily close to each other, the PPP assumption for the BS locations fails to model the underlying repulsion among macro BSs and generally gives a pessimistic signal-to-interference-plus-noise ratio (SINR) distribution~\cite{trac}. In this paper, we propose to use DPPs~\cite{zeros} to model the macro BS locations. We demonstrate the analytical tractability of the proposed model and present statistical evidence to validate the accuracy of DPPs in modeling BS deployments. %This is based on the analytical tractability of DPPs, as well as statistical evidence that DPPs provide an accurate model for many BS deployments.
%deployments and validate that cellular networks with DPP configured BSs are accurate and analytically tractable.  

\subsection{Related Works}
Cellular network performance metrics, such as the coverage probability and achievable rate, strongly depend on the spatial configuration of BSs. PPPs have become increasingly popular to model cellular BSs not only because they can describe highly irregular placements, but also because they allow the use of powerful tools from stochastic geometry and are amenable to tractable analysis~\cite{trac}. While cellular networks with PPP distributed BSs have been studied in early works such as~\cite{baccelli1997stochastic,brown2000shotgun,baccelli2001coverage}, the coverage probability and average Shannon rate were derived only recently in~\cite{trac}. The analysis of cellular networks with PPP distributed BSs has been widely extended to other network scenarios, including heterogeneous cellular networks (HetNets)~\cite{Harpreet,dhillon2013load,mukherjee2012distribution,madhusudhanan2011multi,elsawy2013stochastic}, MIMO cellular networks~\cite{heath2013hetnets,madhusudhanan2013modeling}, and MIMO HetNets~\cite{heath2013hetnets,dhillon2013downlink,kountouris2013antennas}. 

Real (macro) BS deployments exhibit ``repulsion'' between the BSs, which means that macro BSs are typically distributed more regularly than the realization of a PPP. Although the statistics of the propagation losses between a typical user and the BSs converge to that of a Poisson network model under i.i.d. shadowing with large variance~\cite{blaszczyszyn2014wireless}, these assumptions are quite restrictive and may not always hold in practice.
%In addition, several spatial descriptive statistics cannot be accurately estimated under PPP assumption for BSs~\cite{pairmodel}~\cite{anjin}. 
Therefore, several recent research efforts have been devoted to investigating more accurate point process models for representing BS deployments. One class of such point processes is the Gibbs point process~\cite{pairmodel,anjin,riihijarvi2010modeling}. Gibbs models were validated to be statistically similar to real BS deployments using SIR distribution and Voronoi cell area distribution~\cite{pairmodel}. The Strauss process, which is an important class of Gibbs processes, can also provide accurate statistical fit to real BS deployments~\cite{anjin,riihijarvi2010modeling}. By contrast, the PPP and the grid models were demonstrated to be less accurate models for real BS deployments~\cite{pairmodel,anjin}. A significant limitation of Gibbs processes is their lack of tractability, since their probability generating functional is generally unknown~\cite{anjin}. Therefore, point processes that are both tractable and accurate in modeling real BS deployments are desirable.

For several reasons, determinantal point processes (DPPs) are a promising class of point processes to model cellular BS deployments. First, DPPs have soft and adaptable repulsiveness~\cite{FME}. Second, there are quite effective statistical inference tools for DPPs~\cite{zeros,rubak}. Third, many stationary DPPs can be easily simulated~\cite{rubak,decreusefond2013perfect,decreusefond14ginibre}. Fourth, DPPs have many attractive mathematical properties, which can be used for the analysis of cellular network performance~\cite{shirai2003random,miyoshi2012cellular}.

The Ginibre point process, which is a type of DPP, has been recently proposed as a possible model for cellular BSs. Closed-form expressions of the coverage probability and the mean data rate were derived for Ginibre single-tier cellular networks in~\cite{miyoshi2012cellular}, and heterogeneous cellular networks in~\cite{nakata2013spatial}. In~\cite{deng14ginibre}, several spatial descriptive statistics and the coverage probability were derived for Ginibre single-tier networks. These results were empirically validated by comparing to real BS deployments. 
%However, the analysis for Ginibre configured cellular networks has been largely facilitated by a specific property of the GPP model~\cite{goldman2010palm}. 
That being said, the modeling accuracy and analytical tractability of using general DPPs to model cellular BS deployments are still largely unexplored. 

\subsection{Contributions}
In this work, we derive several key performance metrics in cellular networks with DPP configured BSs for the first time. Then we use statistical methods to show that DPPs indeed accurately model cellular BSs. Finally, we describe the gains provided by the use of DPPs for the performance evaluation of cellular networks. The main contributions of this paper are now summarized.

%and they are g even for Although DPPs have many potential advantages, the accuracy of using DPP models for macro BS distribution is still unexplored. In this work, we will fit three DPP models (Gauss, Cauchy, Generalized Gamma) proposed in~\cite{rubak} to real base station locations of a major U.S. city. Based on metrics including K-function, L-function, sum interference and coverage probability, the goodness-of-fit of these DPP models is evaluated through hypothesis testing procedures. The main observations of this paper include: (1) DPPs are accurate models for real macro BS deployment, especially for landlocked cities where repulsiveness among BSs are expected; (2) Generally, Cauchy DPP can be used to model BS deployment with small repulsiveness, and Generalized Gamma DPP is useful for BS deployment with stronger repulsiveness. By contrast, Gauss DPP model has the best balance between modeling accuracy and mathematical tractability, which makes it a promising DPP model for real BS deployment.  

\textbf{DPPs are tractable models to analyze cellular networks:} 
%Three unique computational properties of the DPPs have been identified, including closed form expressions of the product density, Laplace functional and the reduced Palm distribution. 
We summarize three key computational properties of the DPPs, and derive the Laplace functional of the DPPs and independently marked DPPs for functions satisfying certain conditions. 
Based on these computational properties, we analytically derive and numerically evaluate several performance metrics, including the empty space function, nearest neighbor function, mean interference\footnote{By interference, we mean the sum interference power, which is a random shot-noise field~\cite{stochgeom}.} and SIR distribution. The Quasi-Monte Carlo integration method is used for efficient evaluation of the derived empty space function, nearest neighbor function, and mean interference. Finally, the SIR distribution under the nearest BS association scheme is derived, and a close approximation is proposed for efficient numerical evaluation in the high SIR regime.%the Laplace transform of the interference is used to evaluate the coverage probability under the nearest BS association scheme.

\textbf{DPPs are accurate models for macro BS deployments:} We fit three stationary DPP models---the Gauss, Cauchy and Generalized Gamma DPP---to real macro BS deployments from two major U.S. cities, and show that these DPP models are generally accurate in terms of spatial descriptive statistics and coverage probability. We find that the Generalized Gamma DPP provides the best fit to real BS deployments in terms of coverage probability, but is generally less tractable. In contrast, the Gauss DPP model also provides a reasonable fit while offering better mathematical tractability. Compared to other DPP models, the fitted Cauchy DPP provides the least precise results in terms of coverage probability. We also show that the fitted Generalized Gamma DPP is the most repulsive while the fitted Cauchy DPP is the least repulsive.

\textbf{DPPs outperform the PPPs to predict key performance metrics in cellular networks:} By combining the analytical, numerical and statistical results, we show that DPPs are more accurate than PPPs to model BS deployments in terms of the empty space function, the nearest neighbor function, the mean interference and most importantly, the coverage probability.

\section{Mathematical Preliminaries on Determinantal Point Processes}~\label{SECPrelimi}
%The soft repulsive nature of determinantal point processes has made them suitable models for the spatial distribution of cellular base stations. 
%In this section, we provide the definition, important computational properties and several examples of DPPs.
\vspace{-2.5em}
\subsection{Definition of DPPs} 

DPPs are defined based on their $n$-th order product density. Consider a spatial point process $\Phi$ defined on a locally compact space $\Lambda$; then $\Phi$ has $n$-th order product density function $\rho^{(n)} : \Lambda^{n} \rightarrow [0,\infty)$ if for any Borel function $h: \Lambda^{n} \rightarrow [0,\infty)$:
\begin{equation}\label{productdensity}
\begin{split}
E \sum_{X_{1},...,X_{n} \in \Phi}^{\neq} h(X_{1},...,X_{n}) =& \int_{\Lambda} \cdotp\cdotp\cdotp \int_{\Lambda} \rho^{(n)}(x_{1},...,x_{n}) \times h(x_{1},...,x_{n}){\rm d}x_{1}\cdotp \cdotp \cdotp {\rm d}x_{n},
\end{split}
\end{equation}
where $\neq$ means $X_{1},...,X_{n}$ are pair-wise different. 

Let $\mathbb{C}$ denote the complex plane; then for any function $K: \Lambda \times \Lambda \rightarrow \mathbb{C}$, we use $\left(K(x_{i},x_{j})\right)_{1\leq i,j\leq n}$ to denote the square matrix with  $K(x_{i},x_{j})$ as its $(i,j)$-th entry. In addition, denote by $\det \textit{A}$ the determinant of the square matrix \textit{A}. 

\begin{definition}\label{DPPdefnition}
The point process $\Phi$ defined on a locally compact space $\Lambda$ is called a determinantal point process with kernel $K: \Lambda \times \Lambda \rightarrow \mathbb{C}$, if its $n$-th order product density has the following form:
\begin{equation}\label{DPPdefn}
\rho^{(n)}(x_{1},...,x_{n})=\text{det}\left(K(x_{i},x_{j})\right)_{1\leq i,j\leq n}, \qquad (x_1,...,x_n) \in \Lambda^n.
\end{equation}
\end{definition}

Throughout this paper, we will focus on DPPs defined on the Euclidean plane $\mathbb{R}^2$, and we denote the DPP $\Phi$ with kernel $K$ by $\Phi \sim \text{DPP}(K)$. 
%In order to guarantee the existence of the DPP, 
The kernel function $K(x,y)$ is assumed to be a continuous, Hermitian, locally square integrable and non-negative definite function\footnote{This is not a sufficient condition to guarantee the existence of the DPP. Readers are referred to~\cite{zeros,rubak} for more details.}. 

\begin{remark}
The soft-core repulsive nature of DPPs can be explained by the fact that when two points $x_{i} \approx x_{j}$ for $i\neq j$, we have $\rho^{(n)}(x_{1},...,x_{n}) \approx 0$. 
\end{remark}

A DPP $\Phi$ is stationary if its $n$-th order product density is invariant under translations. A natural way to guarantee the stationarity of a DPP is that its kernel $K$ has the form: 
\begin{equation*}
K(x,y)=K_{0}(x-y),  \qquad x,y \in \mathbb{R}^2.
\end{equation*}
In this case, $K_{0}$ is also referred to as the covariance function of the DPP. For stationary DPPs, the intensity measure (i.e., first order product density) is constant over $\mathbb{R}^2$. Further if the stationary DPP is isotropic, i.e., invariant under rotations, its kernel only depends on the distance between the node pair. 
%kernel will have the form of: 
%\begin{equation*}
%K(x,y)=K_{0}(\|x-y\|)  \qquad x,y \in \mathbb{R}^2
%\end{equation*}
%Some good properties can be immediately derived for the stationary and isotropic DPPs. For example, intensity function of $\Phi$ will be homogeneous throughout $\mathbb{R}^2$, and pair correlation function will only depend the distance between node pairs. 
Another important property of stationary DPPs is their spectral density.
 
\begin{definition}
(Spectral Density\cite{rubak}) The spectral density $\varphi$ of a stationary DPP $\Phi$ with covariance function $K_{0}(t)$ is defined as the Fourier transform of $K_{0}(t)$, i.e., $\varphi(x)=\int_{\mathbb{R}^2} K_{0}(t) e^{-2 \pi ix \cdot t}{\rm d} t$ for $x \in \mathbb{R}^2$.
%\begin{equation}
%\varphi(x)=\int_{\mathbb{R}^2} K_{0}(t) e^{-2 \pi ix \cdot t}{\rm d} t, \qquad x \in \mathbb{R}^2.
%\end{equation}
\end{definition}

The spectral density is useful for simulating stationary DPPs.
%based on which the covariance function $K_{0}$ can be approximated by its Fourier series expansion on $\mathbb{R}^2$. 
In addition, the spectral density can also be used to assess the existence of the DPP associated with a certain kernel. Specifically, from Proposition 5.1 in~\cite{rubak}, the existence of a DPP is equivalent to its spectral density $\varphi$ belonging to $[0,1]$.
%\end{enumerate}

\subsection{Computational Properties of DPPs}\label{SubSecCompProp}
We now list the computational properties which make DPPs mathematically tractable for analyzing cellular networks. %Specifically, we will detail three properties which  will be used in the following sections: 
%In particular, the three properties discussed below will be useful in the following sections:

1. DPPs have closed-form product densities of any order. Specifically, for any $n \in \mathbb{N}$, the $n$-th order product density of $\Phi \sim \text{DPP}(K)$ is given by~(\ref{DPPdefn}). Therefore, higher order moment measures of shot noise fields such as the mean/variance of interference in cellular networks can be derived. In addition, the factorial moment expansion approach of~\cite{FME} can also be applied to derive the success probability in wireless networks, which only depends on the product density~\cite[Theorem~3]{FME}. 

%Poisson point process has been widely adopted to analyze cellular networks due to its simple closed-form expression for Laplace functional, as well as Slyvniak theorem which states for any Poisson process $\Phi$ and an event $A$, $\mathbb{P}[\Phi \in A ] = \mathbb{P}_{o}^{!}[\Phi \in A]$. In fact, DPPs have similar properties as Poisson process, which are shown in the following two properties. 

2. DPPs have a closed-form Laplace functional for any nonnegative measurable function $f$ on $\mathbb{R}^2$ with compact support~\cite[Theorem 1.2]{shirai2003random}.

\begin{lemma}[Shirai \textit{et al.}~\cite{shirai2003random}]\label{LFDPP}
Consider $\Phi \sim \text{DPP}(K)$ defined on $\mathbb{R}^2$, where the kernel $K$ guarantees the existence of $\Phi$. Then $\Phi$ has the Laplace functional: %for any nonnegative function $f$ defined on $B$, which is given as:
\begin{equation}\label{LPFunEq}
\small
\mathbb{E}\left[\exp\left(-\int_{\mathbb{R}^2} f(x) \Phi({\rm d} x)\right)\right] = \sum_{n=0}^{+\infty}  \frac{(-1)^n}{n!}  \int_{(\mathbb{R}^2)^n}\det \left(K(x_i,x_j)\right)_{1 \leq i,j \leq n} \prod_{i=1}^{n} \left(1-\exp(-f(x_i))\right) {\rm d}x_1 ... {\rm d}x_n,
\end{equation}
for any nonnegative measurable function $f$ on $\mathbb{R}^2$ with compact support.
\end{lemma}

In the next lemma, we relax the strong requirement for $f$ to have compact support, and show~(\ref{LPFunEq}) holds for more general functions. 

\begin{lemma}\label{LFDPP2}
Consider $\Phi \sim \text{DPP}(K)$ defined on $\mathbb{R}^2$, where the kernel $K$ guarantees the existence of $\Phi$. Then for any nonnegative measurable function $f$ which satisfies the following conditions\footnote{For $x \in \mathbb{R}^2$ and $r \geq 0$, $B(x,r)$ ($B^{o}(x,r)$) denotes the closed (open) ball with center $x$ and radius $r$. In addition, $B^{c}(x,r)$ denotes the complement of $B(x,r)$.}: (a) $\lim\limits_{|x| \rightarrow \infty} f(x) = 0$; (b) $\lim\limits_{r \rightarrow \infty} \int_{\mathbb{R}^2 \setminus B(0,r)} K(x,x) f(x) {\rm d}x = 0$; and (c) $\int_{\mathbb{R}^2} K(x,x) (1-\exp(-f(x))) {\rm d}x < +\infty$, the Laplace functional of $\Phi$ is given by~(\ref{LPFunEq}).
\end{lemma}
\begin{proof}
The proof is provided in Appendix~\ref{LFDPP2Appdx}.
\end{proof}

Based on Lemma~\ref{LFDPP2}, we can easily derive the probability generating functional (pgfl)~\cite{stochgeom} of $\Phi \sim \text{DPP}(K)$, which is given in the following corollary.
\begin{corollary}\label{pgfl}
If $K$ guarantees the existence of $\Phi \sim \text{DPP}(K)$, then the pgfl of $\Phi$ is:
\allowdisplaybreaks
\begin{align}
\small
G[v] \triangleq \mathbb{E}\left(\prod_{x \in \Phi} v(x) \right) = \sum_{n=0}^{+\infty} \frac{(-1)^n}{n!}  \int_{(\mathbb{R}^2)^n}\det \left(K(x_i,x_j)\right)_{1 \leq i,j \leq n} \prod_{i=1}^{n} \left(1-v(x_i)\right) {\rm d}x_1 ... {\rm d}x_n,
\end{align}
for all measurable functions $v : \mathbb{R}^2 \rightarrow [0,1]$, such that $-\log v$ satisfies the conditions in Lemma~\ref{LFDPP2}.
\end{corollary}
%\begin{proof}
%The proof is provided in Appendix~\ref{PGFLCoroApp}.
%\end{proof}

This corollary can be derived using Lemma~\ref{LFDPP2}, thus we omit the detailed proof.

In the next lemma, we extend the Laplace functional of DPPs to independently marked DPPs, where the marks are independent and identically distributed (i.i.d.) and also independent of the ground point process. 
%This scenario is common in wireless networks, where the marks can represent fading coefficients, which are i.i.d. and independent of the transmitter point process. 
\begin{lemma}~\label{LFMarkDPPCoro}
Consider a DPP $\Phi = \sum_{i} \delta_{x_{i}}$, where $\Phi$ is defined on $\mathbb{R}^2$ with kernel $K$. Each node $x_i \in \Phi$ is associated with an i.i.d. mark $p_i$, which is also independent of $x_i$. Denote the probability law of the marks as $F(\cdot)$. Then the Laplace functional of the independently marked point process $\tilde{\Phi} = \sum_i \delta_{(x_i,p_i)}$ is given by:
\begin{small}
\begin{align}
\allowdisplaybreaks
 L_{\tilde{\Phi}}(f) &\triangleq \mathbb{E}\left[\exp\left(-\sum_{i} f(x_i,p_i) \right)\right] \nonumber\\
&=\sum_{n = 0}^{+\infty} \frac{(-1)^n}{n!} \int_{(\mathbb{R}^2)^n} \det \left(K(x_i,x_j)\right)_{1 \leq i,j \leq n} \prod_{i=1}^{n} \left(1-\int_{\mathbb{R}^{+}}\exp(-f(x_i,p_i)) F({\rm d} p_i) \right) {\rm d}x_1 ... {\rm d}x_n,
\end{align}
\end{small}
for any nonnegative measurable function $f$ on $\mathbb{R}^2$, such that $-\log {\int_{\mathbb{R}^{+}} \exp(-f(x,p))F({\rm d}p) }$ satisfies the conditions in Lemma~\ref{LFDPP2}.
\end{lemma}
\begin{proof}
The proof is provided in Appendix~\ref{LFMarkCoroApp}.
\end{proof}

The Laplace functional provides a strong tool to analyze the shot noise field of a DPP. In particular, it facilitates the
analysis of interference and coverage probability in cellular networks.

3.  Under the reduced Palm distribution\footnote{For a spatial point process $\Phi$, denote $\mathbb{P}_{x_0}^{!}(\cdot)$ as the reduced Palm distribution given $x_0 \in \Phi$. For any event $A$, a heuristic definition of $\mathbb{P}_{x_0}^{!}(\cdot)$ is: $\mathbb{P}_{x_0}^{!}(A) = \mathbb{P}(\Phi \backslash \{x_0\} \in A | x_0 \in \Phi)$. The readers are referred to~\cite[p. 131]{chiu2013stochastic} for formal definitions.}, the DPP has the law of another DPP whose kernel is given in closed-form~\cite[Theorem 1.7]{shirai2003random}.

\begin{lemma}[Shirai \textit{et al.}~\cite{shirai2003random}]\label{PalmDPP}
%\cite[Theorem~1.7]{shirai2003random} %Define the reduced Palm distribution of point process $\Phi$ given $x_0 \in \Phi$ as $\mathbb{P}_{x_0}^{!}(A) = \mathbb{P}\left(\Phi-\delta_{x_0} \in A | x_0 \in \Phi \right)$, where $A$ is an event. 
Consider $\Phi \sim \text{DPP}(K)$, where the kernel $K$ guarantees the existence of $\Phi$. Then under the reduced Palm distribution at $x_0 \in \mathbb{R}^2$, $\Phi$ coincides with another DPP associated with kernel $K_{x_0}^{!}$ for Lebesgue almost all $x_0$ with $K(x_0,x_0) >0$, where:
\begin{align}\label{PalmEqn}
\allowdisplaybreaks
K_{x_0}^{!} (x,y) = \frac{1}{K(x_0,x_0)} \det\left(\begin{array}{cc}
K(x,y) & K(x,x_0) \\ 
K(x_0,y) & K(x_0,x_0)
\end{array} \right).
\end{align}
\end{lemma}

This property shows that DPPs are closed under the reduced Palm distribution, which provides a tool similar to Slyvniak's theorem for Poisson processes~\cite{chiu2013stochastic}. In cellular networks, when $x_0$ is chosen as the serving base station to the typical user, this property shows that all other interferers will form another DPP with the modified kernel provided in~(\ref{PalmEqn}). 

In addition, it has been proved in~\cite[Theorem 6.5]{shirai2003random} that if $K(x_0,x_0) > 0$, we have:
\begin{equation}~\label{PalmRelation}
\det(K_{x_0}^{!} (x_i,x_j))_{1 \leq i,j \leq n} = \frac{1}{K(x_0,x_0)} \det(K(x_i,x_j))_{0 \leq i,j \leq n}.
\end{equation}
Therefore, under the reduced Palm distribution at $x_0$ with $ \rho^{(1)}(x_0) >0$, a DPP $\Phi$ with $n$-th order product density function $\rho^{(n)}(x_1,...,x_n)$ will coincide with another DPP with $n$-th order product density: $\rho_{x_0}^{(n)}(x_1,...,x_n) = \rho^{(n+1)}(x_0,x_1,...,x_n) / \rho^{(1)}(x_0)$.

\subsection{Examples of Stationary DPP Models}~\label{SubSectDPPEg}
%In this part, we will give three specific stationary DPP models, which are proposed in~\cite{rubak}.
%Since stationary DPP models are easy to simulate and analyze, 
We will study three DPP models which were proposed in~\cite{rubak}. %In the next sections, we will demonstrate that these three examples are quite accurate in modeling cellular BS deployments.

1. (Gauss DPP Model): A stationary point process $\Phi$ is a Gauss DPP if it has covariance function:
\begin{equation}\label{GaussKernel}
K_{0}(x)= \lambda \exp(-\|x\|^2/\alpha^2), \qquad x \in \mathbb{R}^2.
\end{equation}
%Equivalently, the spectral density of a Gauss DPP is defined as:
%\begin{equation*}
%\varphi(x)=\lambda (\sqrt{\pi} \alpha)^{2} \exp(-\|\pi \alpha x \|^{2}), \qquad x \in \mathbb{R}^{2}.
%\end{equation*}
In the above definition, $\lambda$ denotes the spatial intensity of the Gauss DPP, while $\alpha$ is a measure of its repulsiveness. In order to guarantee the existence of the Gauss DPP model, the parameter pair $(\lambda, \alpha)$ needs to satisfy: $\lambda \leq (\sqrt{\pi} \alpha)^{-2}$.

%\begin{equation*}
%\lambda \leq (\sqrt{\pi} \alpha)^{-2}.
%\end{equation*}

2. (Cauchy DPP Model): The Cauchy DPP model has a covariance function:
\begin{equation}
K_{0}(x)=\frac{\lambda}{\left( 1+\|x\|^{2}/ \alpha^{2}\right) ^{\nu+1}}, \qquad x \in \mathbb{R}^{2}.
\end{equation}
In this model, $\lambda$ describes the intensity, while $\alpha$ is the scale parameter and $\nu$ is the shape parameter. Both $\alpha$ and $\nu$ affect the repulsiveness of the Cauchy DPP. To guarantee the existence of a Cauchy DPP, the parameters need to satisfy: $\lambda \leq \frac{\nu}{(\sqrt{\pi}\alpha)^{2}}$.
%\begin{equation*}
%\lambda \leq \frac{\nu}{(\sqrt{\pi}\alpha)^{2}}.
%\end{equation*}

3. (Generalized Gamma DPP Model): The Generalized Gamma DPP model is defined based on its spectral density: 
\begin{equation}
\varphi(x)=\lambda \frac{\nu \alpha^{2}}{2 \pi \Gamma(2/\nu)}\exp(-\| \alpha x \|^{\nu}),
\end{equation}
where $\Gamma(\cdot)$ denotes the Euler Gamma function. 
The existence of a Generalized Gamma DPP can be guaranteed when $\lambda \leq \frac{2 \pi \Gamma(2/\nu)}{\nu \alpha^2}$.

%\begin{equation*}
%\lambda \leq \frac{2 \pi \Gamma(2/\nu)}{\nu \alpha^2}.
%\end{equation*}

%Generalized Gamma DPP model generally exhibits stronger repulsiveness, which will be studied in more detail in next section.

\subsection{Two Base Station Deployment Examples}

BS deployments in two major U.S. cities are investigated in this paper\footnote{BS location data was provided by Crown Castle.}. Fig.~\ref{RealBS} shows the BS deployment of 115 BSs in a 16 km $\times$ 16 km area of Houston, as well as the deployment of 184 BSs in a 28 km $\times$ 28 km area of Los Angeles (LA). Both deployments are for sprawling and relatively flat areas, where repulsion among BSs is expected. 

\begin{comment}
\begin{figure}
        \centering
        \includegraphics[height=1.8in, width=2in]{BSHouston.eps}
        \includegraphics[height=1.8in, width=2in]{BSLA.eps}
        \caption{BS deployments in Houston (left) and LA (right).}\label{RealBS}
\end{figure}

\begin{figure}
        \centering                
                \includegraphics[height=1.8in, width=2in]{GaussDPPHou.eps}
                 \includegraphics[height=1.8in, width=2in]{CauchyDPPHou.eps}
                 \includegraphics[height=1.8in, width=2in]{GGDPPHou.eps}                             
        \caption{Gauss DPP (left), Cauchy DPP (middle) and Generalized Gamma DPP (right) fitted to the Houston BS deployment.}\label{BSfit}
\end{figure}
\end{comment}
\begin{figure}
         \centering
        \begin{subfigure}[b]{0.3\textwidth}
                \includegraphics[height=1.8in, width=2in]{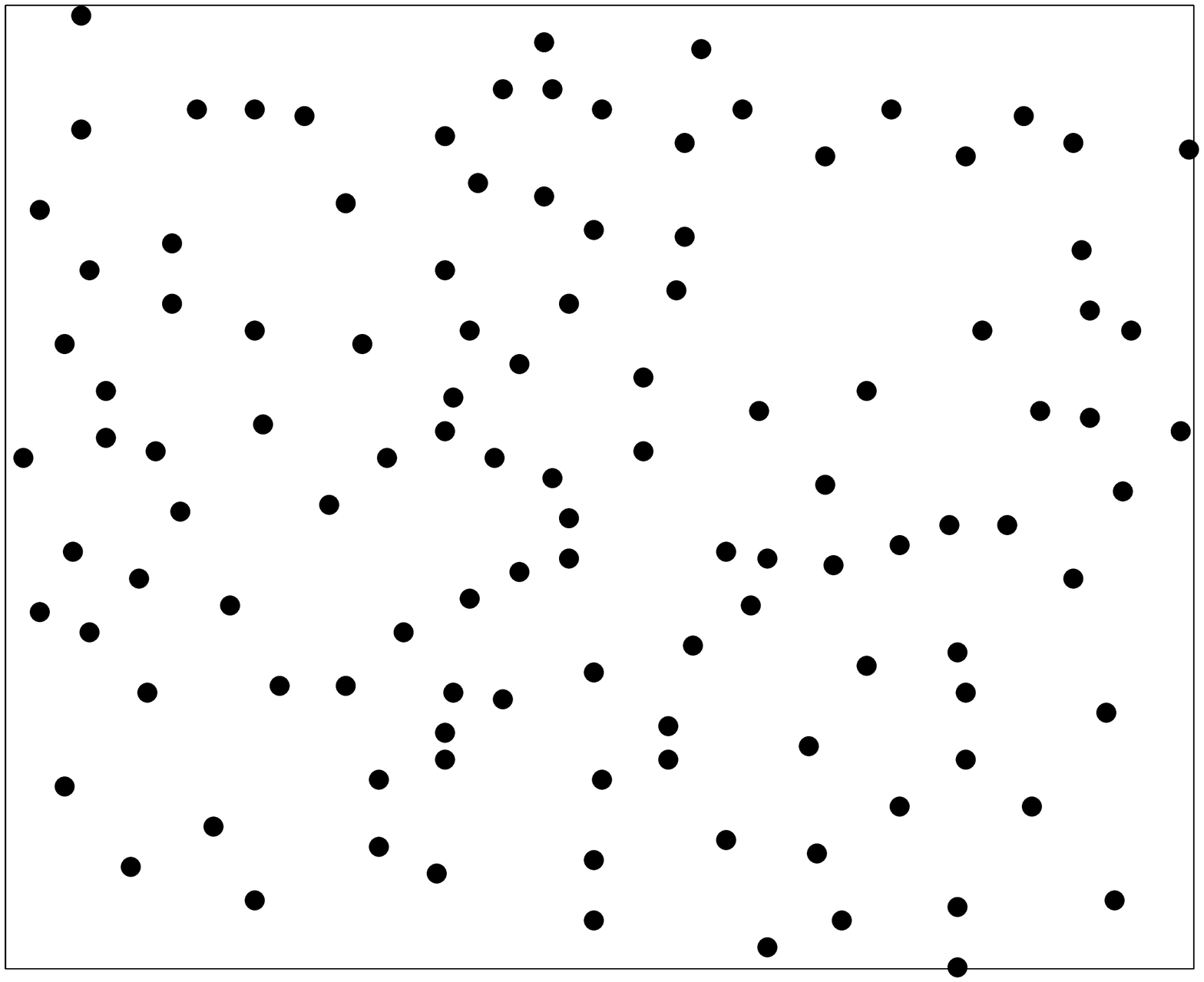}
                \caption{Houston data set}
%                \label{GGfitKest}
        \end{subfigure}
        \begin{subfigure}[b]{0.3\textwidth}
                \includegraphics[height=1.8in, width=2in]{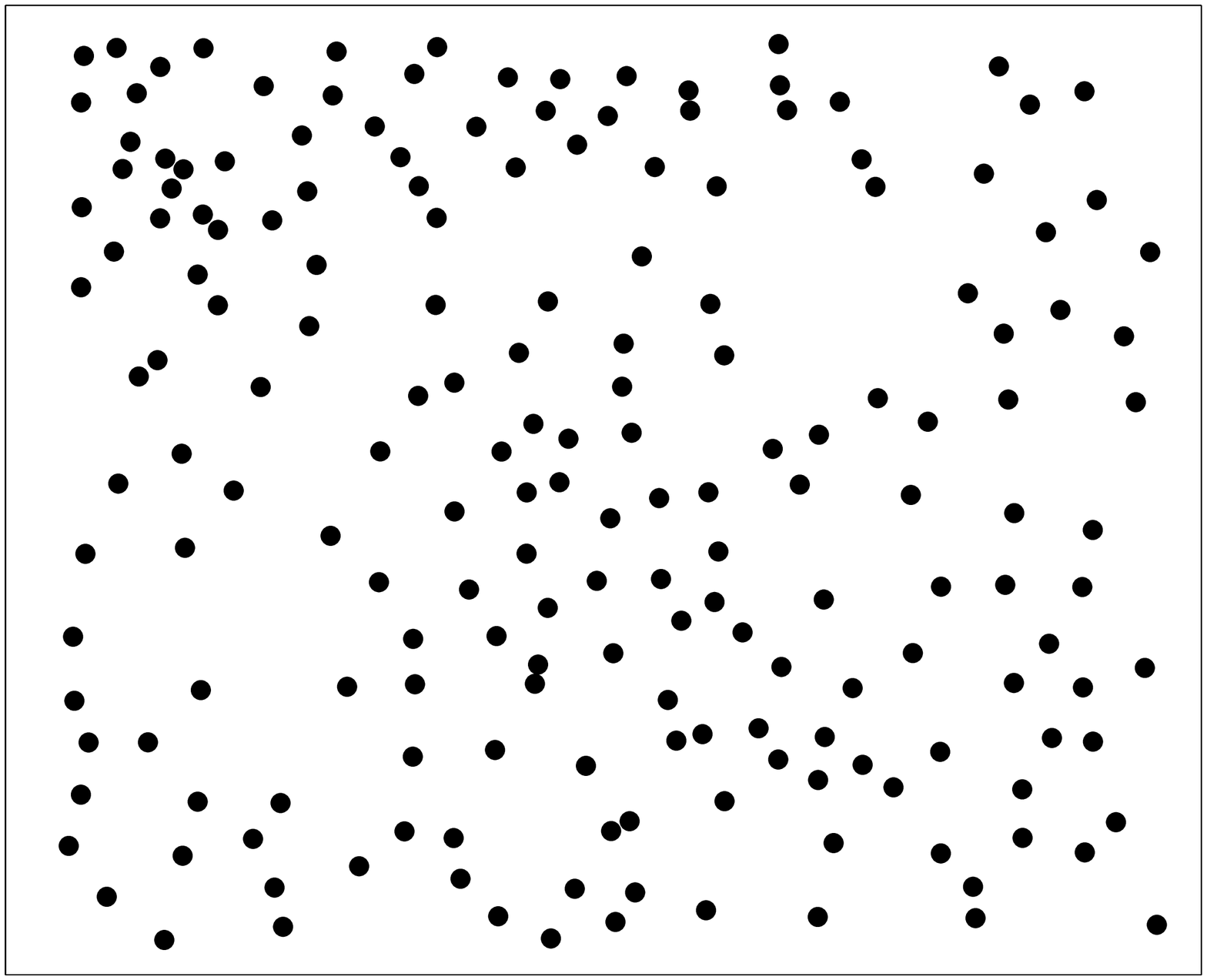}
                \caption{LA data set}
%                \label{coverageGG}
        \end{subfigure}
        \caption{Real macro BS deployments.}\label{RealBS}
\end{figure}

\begin{figure}
        \centering          
        \begin{subfigure}[b]{0.3\textwidth}
                 \includegraphics[height=1.8in, width=2in]{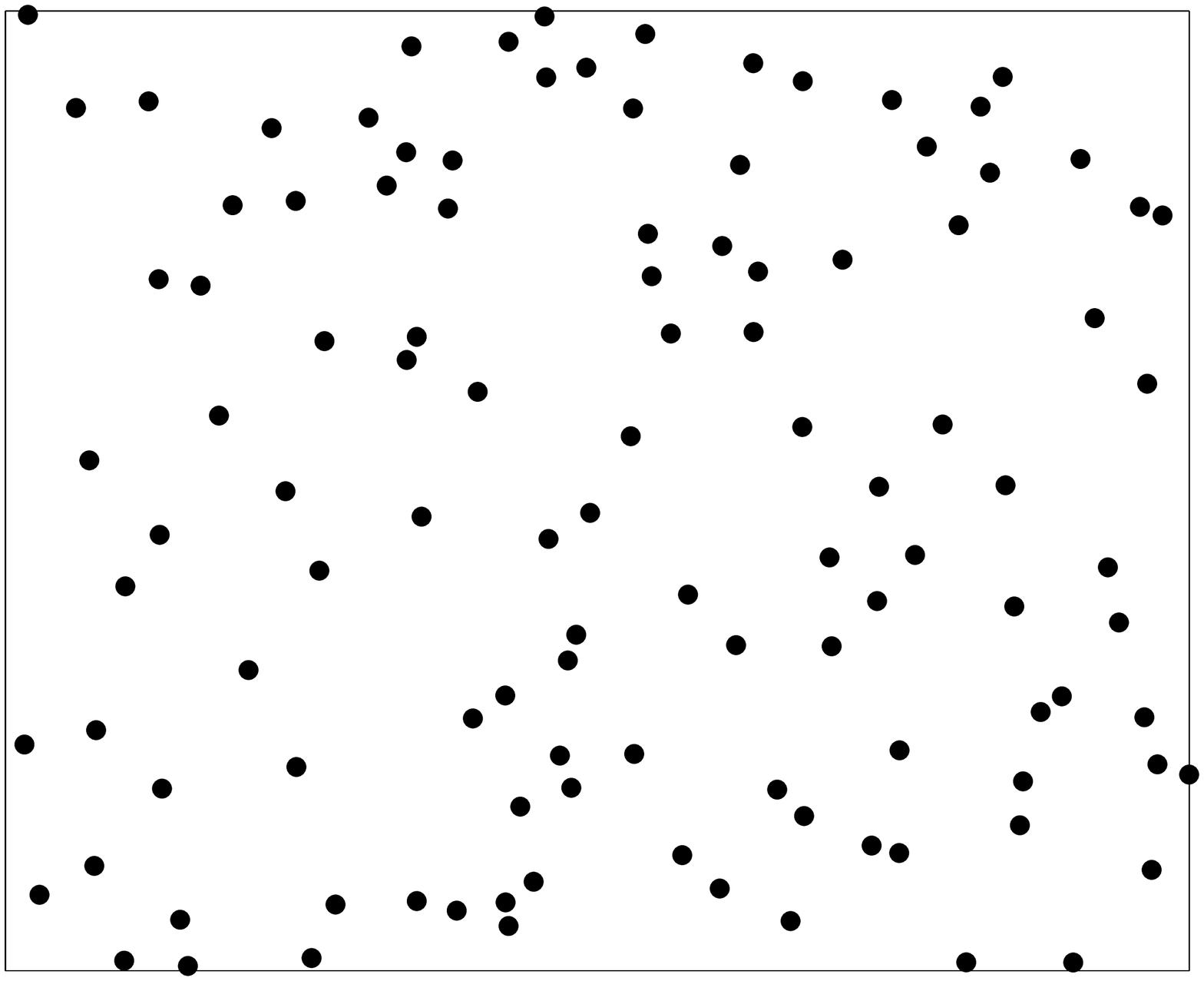}
                 \caption{Gauss DPP}
        \end{subfigure}      
        \begin{subfigure}[b]{0.3\textwidth}
                 \includegraphics[height=1.8in, width=2in]{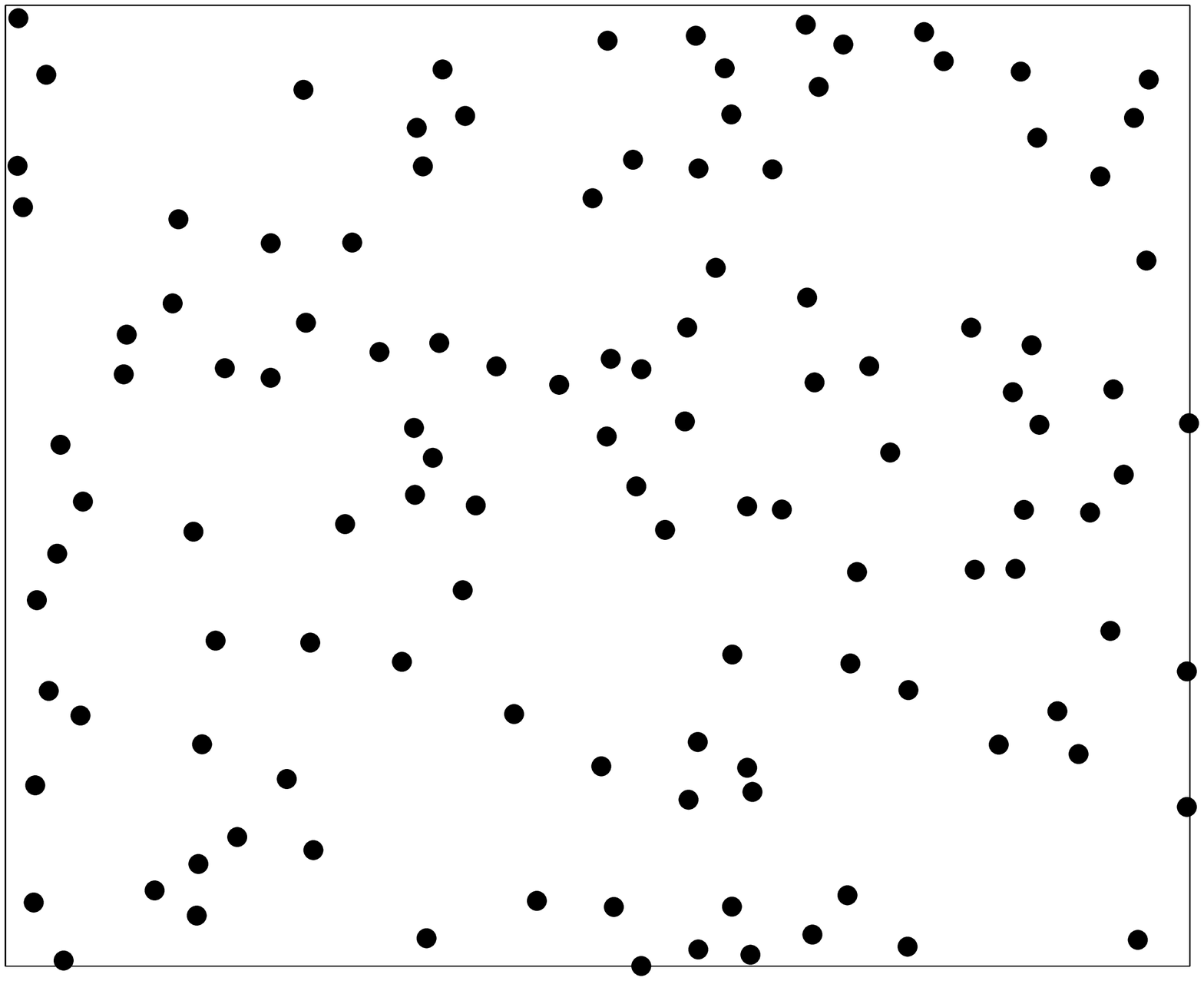}
                 \caption{Cauchy DPP}
        \end{subfigure}    
        \begin{subfigure}[b]{0.3\textwidth}
                 \includegraphics[height=1.8in, width=2in]{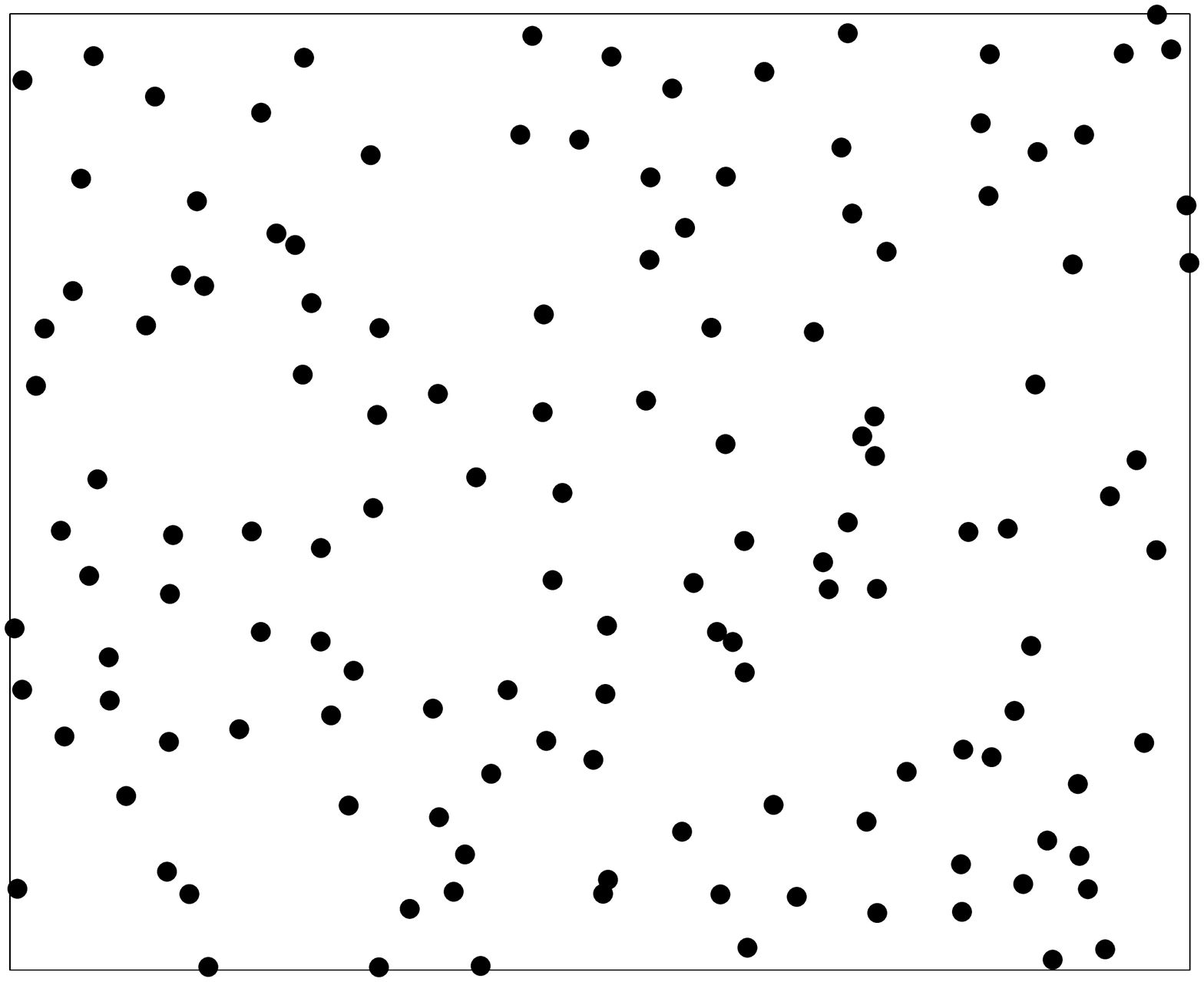}
                 \caption{Generalized Gamma DPP}
        \end{subfigure}                         
        \caption{DPP models fitted to the Houston BS deployment.}\label{BSfit}
\end{figure}

Based on the maximum likelihood (ML) estimate method which is implemented in the software package provided in~\cite{rubak}, we have summarized the estimated parameters for different DPPs fitted to the Houston and LA data set in Table~\ref{par1} and Table~\ref{par2}. Realizations of the Gauss DPP, Cauchy DPP and Generalized Gamma DPP fitted to the Houston urban area deployment are shown in Fig.~\ref{BSfit}. From these figures, it can be qualitatively observed that the fitted DPPs are regularly distributed and close to the real BS deployments. In Section~\ref{SECModelAccu}, we will rigorously validate the accuracy of these DPP models based on different summary statistics.

\begin{table}[h]
\small
\center \caption{DPP Parameters for the Houston Data Set}
\begin{tabular}{|c|c|c|c|p{60mm}|}
\hline  Model & $\lambda$ & $\alpha$ & $\nu$ \\ 
\hline  Gauss DPP  &  0.4492  & 0.8417 & $-$ \\  
\hline  Cauchy DPP & 0.4492 & 1.558 &  3.424 \\ 
\hline Generalized Gamma DPP  & 0.4492 &  2.539 & 2.63\\ 
\hline 
\end{tabular}  
\label{par1}
%\end{table}

%\begin{table}[h]
%\small
\center \caption{DPP Parameters for the LA Data Set}
\begin{tabular}{|c|c|c|c|p{60mm}|}
\hline  Model & $\lambda$ & $\alpha$ & $\nu$ \\ 
\hline  Gauss DPP  &  0.2347  & 1.165 & $-$ \\ 
\hline  Cauchy DPP & 0.2347 & 2.13 &  3.344 \\ 
\hline Generalized Gamma DPP  & 0.2347 &  3.446 & 2.505\\ 
\hline 
\end{tabular}  
\label{par2}
\end{table}

\section{Analyzing Cellular Networks using Determinantal Point Processes}\label{SectAnalyze}
%DPPs are attractive models for cellular networks not only because they can accurately model macro base station deployments, but also because they are mathematically tractable. 
In this section, based on the three important computational properties discussed in Section~\ref{SubSecCompProp}, we analyze several fundamental metrics for the analysis of downlink cellular networks with DPP configured BSs: (1) the empty space function; (2) the nearest neighbor function; (3) the mean interference and (4) the downlink SIR distribution.

\subsection{Empty Space Function}
The empty space function is the cumulative distribution function (CDF) of the distance from the origin to its nearest point in the point process. It is also referred to as the spherical contact distribution. Consider $\Phi \sim \text{DPP}(K)$ and let $d(o, \Phi) = \inf \{\|x\| : x \in \Phi \}$; then the empty space function $F(r)$ is defined as: $F(r) = \mathbb{P} \left(d(o, \Phi) \leq r\right)$ for $r \geq 0$~\cite{chiu2013stochastic}. 

In cellular networks, when each user is associated with its nearest BS, the empty space function provides the distribution of the distance from the typical user to its serving BS, which further dictates the statistics of the received signal power at the typical user. %Based on the Laplace functional of DPP models provided in Lemma~\ref{LFDPP}, we have the following lemma:

\begin{lemma}\label{ESF}
For any $\Phi \sim \text{DPP}(K)$, the empty space function $F(r)$ for $r \geq 0$ is given by:
\begin{equation}\label{ESFEq}
F(r) = \sum_{n=1}^{+ \infty} \frac{(-1)^{n-1}}{n!} \int_{\left(B(0,r)\right)^n} \det\left(K(x_i,x_j)\right)_{1 \leq i,j \leq n } {\rm d}x_1...{\rm d}x_n.
\end{equation}
\end{lemma}

\begin{proof}
Choose $f(x) = -\log\mathbbm{1}_{\{\|x\| > r\}}$ for $x \in \mathbb{R}^2$,
%by checking the conditions in Lemma~\ref{LFDPP2} are satisfied, 
we have:
\allowdisplaybreaks
\begin{align*}
\mathbb{E}\left[\exp\left(-\int f(x) \Phi({\rm d} x) \right)\right] &= \mathbb{E} \left[\exp\left(- \sum_{x_i \in \Phi} -\log \mathbbm{1}_{\|x_i\|>r}\right)\right] \\
%&= \mathbb{E} \left[\prod_{x_i \in \Phi} \mathbbm{1}_{\|x_i \| > r}\right] \\
%&= \mathbb{E} \left[ \mathbbm{1}_{d(o, \Phi) > r} \right] \\
&= \mathbb{P} \left[d(o, \Phi) > r\right].
\end{align*}

Therefore, based on Lemma~\ref{LFDPP2}, the empty space function is given by:
\begin{align*}
\allowdisplaybreaks
F(r) &= 1 - \mathbb{E}\left[\exp\left(-\int f(x) \Phi({\rm d}x) \right)\right]\\
&= 1- \sum_{n=0}^{+\infty} \frac{(-1)^n}{n!} \int_{(\mathbb{R}^2)^n} \det \left(K(x_i,x_j)\right)_{1 \leq i,j \leq n} \prod_{i=1}^{n} \left(1-\exp(\log \mathbbm{1}_{\{\|x_i\| > r\}})\right) {\rm d}x_1...{\rm d}x_n \\
%&= 1- \sum_{n=0}^{+\infty} \frac{(-1)^n}{n!} \int_{(\mathbb{R}^2)^n} \det \left(K(x_i,x_j)\right)_{1 \leq i,j \leq n} \prod_{i=1}^{n} \mathbbm{1}_{\{\|x_i\| \leq r\}} {\rm d}x_1...{\rm d}x_n \\
&= \sum_{n=1}^{+ \infty} \frac{(-1)^{n-1}}{n!} \int_{\left(B(0,r)\right)^n} \det\left(K(x_i,x_j)\right)_{1 \leq i,j \leq n } {\rm d}x_1...{\rm d}x_n.
\end{align*}
\end{proof}

\begin{comment}
\begin{remark}
Lemma~\ref{ESF} provides a series representation for the empty space function of any determinantal point process. Each term in this series is a multi-dimensional integration over a matrix determinant. Therefore, numerical integration methods are needed to evaluate~(\ref{ESFEq}).  Specifically, we adopt the Quasi-Monte Carlo (QMC) integration method~\cite{kuo2005lifting,dick2010digital}, which will be described in the next section.
\end{remark}
\end{comment}

Based on Lemma~\ref{ESF}, we can also characterize the probability density function (PDF) $f (r) $ of the distance from the origin to its nearest point for all stationary and isotropic DPPs $\Phi$.

\begin{corollary}\label{ESFPDFCoro}
Let $F(r)$ denote the empty space function for a stationary and isotropic DPP $\Phi$ with covariance function $K$. Then $f(r) \triangleq \frac{{\rm d} F(r)}{{\rm d}r}$ is given by:
\begin{align}\label{ESFPDFEq}
f(r) = 2\pi r \sum\limits_{n=0}^{+\infty} \frac{(-1)^n}{n!} \int_{(B(0,r))^n} \det(K(x_i,x_j))_{0 \leq i,j \leq n} \bigg|_{x_0 = (r,0)} {\rm d}x_1...{\rm d}x_n.
\end{align}
\end{corollary}
\begin{proof}
The proof is provided in Appendix~\ref{ESPPDFCoroProof}. 
\end{proof}

%Lemma~\ref{ESF} provides a series representation for the empty space function of any determinantal point process. Each term in this series is a multi-dimensional integration over a matrix determinant. Therefore, numerical integration methods are needed to evaluate ~(\ref{ESFEq}). 

\subsection{Nearest Neighbor Function}
The nearest neighbor function gives the distribution of the distance from the typical point of a point process to its nearest neighbor in the same point process. For all stationary DPPs $\Phi$, the nearest neighbor function can be defined based on the reduced Palm distribution of $\Phi$ as: $D(r) = \mathbb{P}_{o}^{!} (d(o,\Phi) \leq r )$~\cite{chiu2013stochastic}. 

In cellular networks, the nearest neighbor function provides the distribution of the distance from a typical BS to its nearest neighboring BS, which can be used as a metric to indicate the clustering/repulsive behavior of the network. Specifically, compared to the PPP, a regularly deployed network corresponds to a larger nearest neighbor function, while a clustered network corresponds to a smaller nearest neighbor function. Therefore, when each user is associated with its nearest BS, the dominant interferers in regularly deployed networks are farther from the serving BS than a completely random network. 

\begin{lemma}\label{NN}
For any $\Phi \sim \text{DPP}(K)$ defined on $\mathbb{R}^2$, its nearest neighbor function $D(r)$ is given by:
\begin{equation}\label{NNEq}
D(r) = \sum_{n=1}^{+ \infty} \frac{(-1)^{n-1}}{n !} \int_{\left(B(0,r)\right)^n}\det\left(K_{o}^{!} (x_i,x_j)\right)_{1 \leq i,j \leq n} {\rm d}x_1...{\rm d}x_n,
\end{equation}
\end{lemma}
where $K_{o}^{!} (x,y)$ is: 
\begin{equation}\label{PalmKernel}
K_{o}^{!} (x,y) = \frac{1}{K(0,0)} \det\left(\begin{array}{cc}
K(x,y) & K(x,0) \\ 
K(0,y) & K(0,0)
\end{array} \right).
\end{equation}

\begin{proof}
Denote $\tilde{\Phi} \sim \text{DPP}(K_{o}^{!} (x,y))$; then it follows from Lemma~\ref{PalmDPP} that:
\begin{equation*}
\mathbb{P}_{o}^{!} (d(o,\Phi) \leq r ) = \mathbb{P} (d(o,\tilde{\Phi}) \leq r).
\end{equation*}
Therefore, the proof can be concluded by applying Lemma~\ref{ESF} to the DPP $\tilde{\Phi}$.
\end{proof}

\subsection{Interference Distribution}\label{SubSecIntDist}
%Interference is an important factor that limits the coverage probability and mean data rate in cellular networks. 
In this section, we analyze properties of shot noise fields associated with a DPP. Our aim is to evaluate interference in cellular networks under two BS association schemes. Firstly, the BS to which the typical user is associated is assumed to be at an arbitrary but fixed location\footnote{This simple conditional interference scenario provides fundamental understanding of interference in wireless networks with DPP configured nodes. The results in this case can be extended to ad-hoc networks as well.}. We show that in this case, the mean interference is easy to characterize with DPP configured BSs. Secondly, each user is assumed to be associated with its nearest BS. In this case, we derive the Laplace transform of interference. 

Throughout this part, the cellular BSs are assumed to be distributed according to a stationary and isotropic DPP $\Phi \sim \text{DPP}(K)$, while the mobile users are uniformly distributed and independent of the BSs. Since $\Phi$ is invariant under translations, we focus on the performance of the typical user which can be assumed to be located at the origin. The location for the serving BS of the typical user is denoted by $x_0$. Each BS $x \in \Phi$ has single transmit antenna with transmit power $P$, and it is associated with an independent mark $h_x$ which represents the small scale fading effects between the BS and the typical user. Independent Rayleigh fading channels with unit mean are assumed, which means $h_x \sim \exp(1)$ for $\forall x \in \Phi$. The shadowing effects are neglected, and the thermal noise power is assumed to be 0, i.e., negligible compared to interference power. In addition, the path loss function is denoted by $\mathit{l}(x): \mathbb{R}^2 \mapsto \mathbb{R}^{+}$, which is a non-increasing function with respect to (w.r.t.) the norm of $x$. 

%Throughout this section, we assume the BSs are distributed according to a stationary and isotropic DPP $\Phi$ with kernel $K$. In particular, the location for the serving BS of the typical user at the origin is denoted by $x_0$ . The path loss model is chosen as $\mathit{l}(x) = \min (1, |x|^{-\beta})$, where $\beta \geq 2$ is the path loss exponent.

\subsubsection{Interference with fixed associated BS scheme} Since $\Phi$ is invariant under translation and rotation, we assume the typical user located at the origin is served by the base station at $x_{0} = (r_0,0)$, where $r_0$ denotes the distance from the origin to $x_0$. 
%Each BS $x \in \Phi$ transmits with unit power $P=1$, and is associated with an independent mark $h_x$ which represents the fading coefficient. Independent Rayleigh fading channel with unit mean is assumed for all BSs, which means $h_x \sim \exp(1)$ for $\forall x \in \Phi$. Path loss function is denoted as $\mathit{l}(x): \mathbb{R}^2 \mapsto \mathbb{R}^{+}$, which is assumed to satisfy $\int_{\mathbb{R}^2} \mathit{l}(x) {\rm d}x < +\infty$. 
% Fig. illustrates the system model. 
Conditionally on $x_{0} \in \Phi$ being the serving BS, the interference at the origin is:
$I = \sum_{x_{i} \in \Phi \backslash x_{0}} P h_{x_{i}} \mathit{l}(x_{i}) $. 
%Based on Lemma~\ref{PalmDPP} and Campbell's theorem~\cite{stochgeom}, the mean interference is given as follows:
\begin{lemma}~\label{MeanInt}
Given $x_{0} = (r_0,0)$ is the serving BS for the typical user located at the origin, the mean interference seen by this typical user is: 
\begin{equation}\label{meanInt1}
\mathbb{E}[I |  x_{0} = (r_0,0)] = P \int_{\mathbb{R}^2} K_{x_{0}}^{!} (x,x) \mathit{l} (x) {\rm d}x,
\end{equation}
where $K_{x_{0}}^{!}(\cdot,\cdot)$ is given in~(\ref{PalmEqn})\footnote{This lemma can be seen as a general property of the shot noise field $I$ created by a DPP, since it holds for all function $\mathit{l}(\cdot)$.}. 
\end{lemma}
\allowdisplaybreaks[1]
\begin{proof}
From Lemma~\ref{PalmDPP}, the mean interference can be expressed as:
\begin{align}
\mathbb{E}[ \sum_{x_{i} \in \Phi \backslash x_{0}}P h_{x_{i}} \mathit{l}(x_{i})|  x_{0} = (r_0,0)] =&\mathbb{E}[\sum_{x_{i} \in \tilde{\Phi}} P h_{x_{i}} \mathit{l}(x_{i})] \nonumber\\
%= &\mathbb{E}[\sum_{x_{i} \in \tilde{\Phi}} P h_{x_{i}} \mathit{l}(x_{i})] \nonumber \\
\overset{(a)}{=}& P \int_{\mathbb{R}^2} \int_{\mathbb{R}^{+}} h \mathit{l}(x) K_{x_0}^{!} (x,x) \exp(-h) {\rm d}h {\rm d}x \nonumber \\
=& P\int_{\mathbb{R}^2} K_{x_0}^{!} (x,x) \mathit{l} (x) {\rm d}x \nonumber,
\end{align}
where $\tilde{\Phi} \sim \text{DPP} (K_{x_0}^{!})$ follows from Lemma~\ref{PalmDPP}, and (a) follows from Campbell's theorem. %Here $ K_{x_0}^{!} (x,x)$ denotes the intensity function of $\tilde{\Phi}$.
\end{proof}

%Next we will consider the mean interference for Gauss DPP, which is given as follows:

In fact, all the higher order moment measures of the interference can be calculated similarly based on 
Definition~\ref{DPPdefnition} and Lemma~\ref{PalmDPP}. 
\subsubsection{Interference with nearest BS association scheme} In this part, we consider the BS association scheme where each user is served by its nearest BS. In single tier cellular networks, the nearest BS association scheme provides the highest average received power for each user. 

For a user located at $y \in \mathbb{R}^2$, its associated BS is denoted by $x^{*}(y) = \underset{x \in \Phi}{\operatorname{argmin}} \|x-y\|$. Consider the typical user located at the origin and its associated BS $x^{*}(0)$. The interference at the typical user is then given by $I = \sum\limits_{x_{i} \in \Phi \backslash x^{*}(0)}P h_{x_i} \mathit{l}(x_{i})$, where $h_{x_i} \sim \exp(1)$ denotes the Rayleigh fading variable from $x_i$ to the origin. 
%From Lemma~\ref{ESF} and Corollary~\ref{ESFPDFCoro}, the CDF and PDF for $|x^{*}(0)|$ are already known. 
In the next theorem, we provide the general result which characterizes the Laplace transform of interference conditional on the position of the BS nearest to the typical user. 

\begin{theorem}\label{LTIntThm}
Conditionally on $x^*(0) = x_0$ being the serving BS of the typical user at the origin, if $f(x,h_x) = sP h_x \mathit{l}(x) \mathbbm{1}_{|x| \geq r_0} - \log \mathbbm{1}_{|x| \geq r_0}$ satisfies the conditions in Lemma~\ref{LFMarkDPPCoro}, then the Laplace transform of the interference at the typical user is:
\begin{equation}
\small
\mathbb{E}[e^{-sI}| x^{*}(0) = x_0] = \frac{\sum\limits_{n=0}^{+\infty} \frac{(-1)^n}{n!}\int_{(\mathbb{R}^2)^n} \det(K_{x_0}^{!}(x_i,x_j))_{1\leq i,j \leq n} \prod\limits_{i=1}^{n}[1- \frac{\mathbbm{1}_{|x_i| \geq r_0}}{1+sP\mathit{l}(x_i)}] {\rm d}x_1...{\rm d}x_n}{\sum\limits_{n=0}^{+\infty} \frac{(-1)^n}{n!} \int_{B(0,r_0)^n} \det(K_{x_0}^{!}(x_i,x_j))_{1 \leq i,j \leq n}{\rm d}x_1...{\rm d}x_n },
\end{equation}
where $r_0 = |x^{*}(0)|$ and $K_{x_0}^{!}(\cdot,\cdot)$ is given in~(\ref{PalmEqn}).
\end{theorem}

\begin{proof}
Denote $\tilde{\Phi} \sim \text{DPP}(K_{x_0}^{!})$, we have:
\allowdisplaybreaks
\begin{align}
\mathbb{E}[\exp(-sI) | x^*(0)=x_0] \nonumber = &\mathbb{E}[\exp(-sI) | x_0 \in \Phi, \Phi(B^o(0,r_0))=0] \nonumber\\
%\overset{(a)}{=} & \mathbb{E}_{x_0} [\exp(-sI) | \Phi(B^{o}(0,r_0))=0] \nonumber\\
\overset{(a)}{=} & \mathbb{E}_{x_0}^{!} [\exp(-s \sum_{x_i \in \Phi \cap B^c(0,r_0)} P h_{x_i} \mathit{l}(x_i) ) | \Phi(B^{o}(0,r_0))=0] \nonumber \\
%= & \mathbb{E}_{x_0}^{!} [\exp(-s \sum_{x_i \in \Phi \cap B^c(0,r_0)} P h_{x_i} \mathit{l}(x_i) ) | \Phi (B^{o}(0,r_0))=0] \nonumber\\
% = & \mathbb{E} [\exp(-s \sum_{x_i \in \tilde{\Phi} \cap B^c(0,r_0)} P h_{x_i} \mathit{l}(x_i) ) | \tilde{\Phi} (B^{o}(0,r_0))=0] \nonumber \\
 \overset{(b)}{=} & \mathbb{E} [\exp(-s \sum_{x_i \in \tilde{\Phi} \cap B^c(0,r_0)} P h_{x_i} \mathit{l}(x_i) ) \mathbbm{1}_{\tilde{\Phi} (B^{o}(0,r_0))=0}] / \mathbb{P}[\tilde{\Phi} (B^{o}(0,r_0))=0]~\label{LTIntEq1},
\end{align}
where (a) follows from the Bayes' rule, 
%which implies that for all events B and C, we have $\mathbb{E}[I| B\cap C]= \int_{\mathbb{R}^{+}} \mathbb{P}(I>T|B \cap C) {\rm d} T = \int_{\mathbb{R}^{+}} \mathbb{P}_C(I>T|B) {\rm d} T = \mathbb{E}_C[I| B] $ \footnote{For all events $A$ and $C$, we denote the conditional probability of event $A$ given $C$ by $\mathbb{P}_{C}(A) = \mathbb{P}(A|C)$.}, 
and the fact that conditionally on $x_0 \in \Phi$, $(\Phi-\delta_{x_0})(B^{o}(0,r_0))=0$ is equivalent to $\Phi(B^{o}(0,r_0))=0$ since $x_0$ lies on the boundary of the open ball $B^{o}(0,r_0)$. In addition, (b) follows from the fact that for all random variables $X$ and events $A$, $\mathbb{E}[X|A] = \frac{\mathbb{E}[X\mathbbm{1}_A]}{\mathbb{P}(A)}$.

Next, it is clear that the denominator in~(\ref{LTIntEq1}) is given by:
\allowdisplaybreaks
\begin{align}~\label{LTIntEq2}
\allowdisplaybreaks
 \mathbb{P}[\tilde{\Phi} (B^{o}(0,r_0))=0] %&= \mathbb{P} [d(o, \tilde{\Phi}) \geq r_0]  \nonumber\\
% &= 1-F(r_0) \nonumber\\
 &=\sum\limits_{n=0}^{+\infty} \frac{(-1)^n}{n!} \int_{B(0,r_0)^n} \det(K_{x_0}^{!}(x_i,x_j))_{1 \leq i,j \leq n}{\rm d}x_1...{\rm d}x_n.
\end{align}

The numerator in~(\ref{LTIntEq1}) is calculated as:
\allowdisplaybreaks
\begin{align}
\allowdisplaybreaks
&\mathbb{E} [\exp(-s \sum_{x_i \in \tilde{\Phi} \cap B^c(0,r_0)}P h_{x_i} \mathit{l}(x_i) ) \mathbbm{1}_{\tilde{\Phi} (B^{o}(0,r_0))=0}] \nonumber\\
% = &\mathbb{E} [\exp(-s \sum_{x_i \in \tilde{\Phi} \cap B^c(0,r_0)}P h_{x_i} \mathit{l}(x_i) ) \prod\limits_{x_i \in \tilde{\Phi}} \mathbbm{1}_{|x_i| \geq r_0}] \nonumber\\
% = &\mathbb{E} [\exp(-s \sum_{x_i \in \tilde{\Phi} \cap B^c(0,r_0)} P h_{x_i}\mathit{l}(x_i) ) \exp(-\sum\limits_{x_i \in \tilde{\Phi}}(-\log \mathbbm{1}_{|x_i| \geq r_0} ) )] \nonumber\\
% = & \mathbb{E} [\exp[-\sum_{x_i \in \tilde{\Phi}} (s P h_{x_i} \mathit{l}(x_i)\mathbbm{1}_{|x_i| \geq r_0} -\log\mathbbm{1}_{|x_i| \geq r_0})]] \nonumber\\
\overset{(a)}{=} &\sum\limits_{n=0}^{+\infty} \frac{(-1)^n}{n!} \int_{(\mathbb{R}^2)^n} \int_{(\mathbb{R}^{+})^n} \det(K_{x_0}^{!}(x_i,x_j))_{1 \leq i,j \leq n} \prod\limits_{i=1}^{n} [(1- \nonumber\\
&\exp(-s P h_{x_i} \mathit{l}(x_i)\mathbbm{1}_{|x_i| \geq r_0} + \log\mathbbm{1}_{|x_i| \geq r_0})) \exp(-h_{x_i}) {\rm d}h_i ] {\rm d}x_1...{\rm d}x_n \nonumber\\
= & \sum\limits_{n=0}^{+\infty} \frac{(-1)^n}{n!}\int_{(\mathbb{R}^2)^n} \det(K_{x_0}^{!}(x_i,x_j))_{1\leq i,j \leq n} \prod\limits_{i=1}^{n}[1- \frac{\mathbbm{1}_{|x_i| \geq r_0}}{1+sP\mathit{l}(x_i)}] {\rm d}x_1...{\rm d}x_n,~\label{LTIntEq3}
\end{align}
where (a) is obtained from Lemma~\ref{LFMarkDPPCoro}.
Finally, substituting~(\ref{LTIntEq2}) and~(\ref{LTIntEq3}) into~(\ref{LTIntEq1}) yields the result. 
\begin{comment}
Sides notes about the case when $r_0$ goes to infinity and 0: 
when $r_0$ goes to infinity, then it's easy to check that the numerator and denominator will correspond to the same value, so (18) will be 1, which means the interference is almost 0. However, since almost surely, we have $r_0$ for DPP is less than $r_0$ for PPP, which is finite a.s., so $r_0$ is finite a.s., so the integration will not account for the case with $r_0$ = infinity;

While when $r_0 = 0$, we have the denominator will just be 1. However, for the numerator, to play with our trick, we need the pass loss exponent to be greater than 2 to make the integration finite (this also conditions on s >0). This is fine in our numerical evaluations. On the other hand, since $x_0$ has 0 probability to choose at the origin, otherwise it's Palm distribution, so we also do not need to consider this point in the integration for calculating coverage probability. 

Finally, if s = 0, then both numerator and denominator will be the same again, so no matter what $r_0$ is, we are in good shape. 

Overall, the integration to calculate the coverage probability is well defined, which integrates from $0^{+}$ to an arbitrary large number, wrt to the Laplace transform. 
\end{comment}
\end{proof}

\begin{remark}
In contrast with what happens in the PPP case, because of the repulsion among DPP points, $\Phi \cap B^c(0,r_0)$ and $\Phi \cap B^o(0,r_0)$ are not independent.
%, which can be observed from the proof of Theorem~\ref{LTIntThm}.
\end{remark}

\begin{remark}
If $\Phi$ is a stationary PPP with intensity $\lambda$, then by substituting $\det(K(x_i,x_j))_{1 \leq i,j \leq n} = \det(K_{x_0}^{!}(x_i,x_j))_{1 \leq i,j \leq n} = \lambda^n$, Theorem~\ref{LTIntThm} gives the Laplace transform of the interference at the typical user to be:
\begin{equation*}
\mathbb{E}[e^{-sI} | x^*(0)= x_0] = \exp\left(-\lambda \int_{B^c(0,r_0)} (1-\frac{1}{1+sP\mathit{l}(x)}){\rm d}s\right),
\end{equation*}
which is consistent with (12) in~\cite{trac}. 
\end{remark}

Since the Laplace transform fully characterizes the probability distribution, many important performance metrics can be derived using Theorem~\ref{LTIntThm}. Specifically, the next lemma gives the mean interference under the nearest BS association scheme. 

\begin{lemma}~\label{MeanIntNearestBSLemma}
The mean interference at the typical user conditional on $x^*(0) = x_0$ is:
\begin{equation}~\label{MeanIntx0}
\mathbb{E}[I|x^*(0) = x_0] = \frac{\sum\limits_{n=0}^{+\infty} \frac{(-1)^{n}}{n!}  \int_{(B(0,r_0))^{n}} \int_{B^{c}(0,r_0)} \det(K_{x_0}^{!} (x_i,x_j))_{1\leq i,j \leq n+1}P \mathit{l}(x_1) {\rm d}x_1...{\rm d}x_{n+1}}{ \sum\limits_{n=0}^{+\infty}\frac{(-1)^n}{n!} \int_{B(0,r_0)^n} \det(K_{x_0}^{!}(x_i,x_j))_{1 \leq i,j \leq n}{\rm d}x_1...{\rm d}x_n},
\end{equation}
where $r_0 = |x_0|$. 
\end{lemma}

\begin{proof}
The proof is provided in Appendix~\ref{MeanIntNearestBSAppdix}. 
\end{proof}

Since the DPPs are assumed to be stationary and isotropic, thus only the distance from the origin to its nearest BS will affect the mean interference result, which can be observed from Lemma~\ref{MeanIntNearestBSLemma}.

%\subsection{Coverage Probability}
\subsection{SIR Distribution}
%In this part, under the same assumptions as in Section~\ref{SubSecIntDist}, we derive the coverage probability for the typical user under the nearest BS association scheme. 
%The BSs are assumed to be distributed according to a stationary and isotropic DPP $\Phi$.
Based on the same assumptions as in Section~\ref{SubSecIntDist}, we derive the SIR distribution as the complementary cumulative distribution function (CCDF) of the SIR at the typical user under the nearest BS association scheme. 
%Since the thermal noise power is assumed to be 0, the CCDF of SIR at the typical user coincides with its coverage probability. 
%The coverage probability is defined as the probability that the received SIR at the typical user is larger than the SIR threshold $T$, i.e., $\mathbb{P}(\text{SIR}(0,\Phi) > T)$. Denote by $x^*(0)$ the BS to which the typical user at the origin associates. Then its received SIR can be expressed as:
Denote by $x^*(0)$ the BS to which the typical user at the origin associates, its received SIR can be expressed as:
\begin{equation}~\label{COPEq}
\text{SIR}(0,\Phi) = \frac{ P h_{x_0} \mathit{l}(x^*(0))}{\sum_{x_{i} \in \Phi \setminus x^*(0)}  P h_{x_i} \mathit{l}(x_i)}.
\end{equation}

%\begin{definition}
%(Coverage probability)
%\end{definition}

%The next lemma gives the coverage probability conditionally on the location of the serving BS for the typical user.
%In the next lemma, we will first give the coverage probability conditionally on the location of serving BS for the typical user.

\begin{lemma}\label{COPThmx0}
%The coverage probability for the typical user at the origin with SIR threshold $T$, given $x^*(0) = x_0$ is:
The SIR distribution for the typical user at the origin, given $x^*(0) = x_0$ is:
\allowdisplaybreaks
\begin{align}
\small
&\mathbb{P}[\text{SIR}(0,\Phi) >T | x^*(0) = x_0] \nonumber\\
= &\frac{\sum\limits_{n=0}^{+\infty} \frac{(-1)^n}{n!}\int_{(\mathbb{R}^2)^n} \det(K_{x_0}^{!}(x_i,x_j))_{1\leq i,j \leq n} \prod\limits_{i=1}^{n}[1- \frac{\mathbbm{1}_{|x_i| \geq r_0}}{1+ T \mathit{l}(x_i)/\mathit{l}(x_0) }] {\rm d}x_1...{\rm d}x_n}{\sum\limits_{n=0}^{+\infty} \frac{(-1)^n}{n!} \int_{B(0,r_0)^n} \det(K_{x_0}^{!}(x_i,x_j))_{1 \leq i,j \leq n}{\rm d}x_1...{\rm d}x_n }~\label{COPx0Eq},
\end{align}
where $r_0 = |x_0|$. 
\end{lemma}

\begin{proof}
Since the channels are subject to Rayleigh fading with unit mean, we have: 
\allowdisplaybreaks
\begin{align*}
\mathbb{P}[\text{SIR}(0,\Phi) >T | x^*(0) = x_0] =& 
\mathbb{P} [\frac{h_0 \mathit{l}(x_0)}{I} >T | x^*(0) = x_0] \\
=& \mathbb{E} [\exp(-\frac{T}{\mathit{l}(x_0)}I) | x^*(0) = x_0],
\end{align*}
and the result follows from Theorem~\ref{LTIntThm}.
\end{proof}

%In Corollary~\ref{ESFPDFCoro}, the probability density function for the distance from the origin to its nearest BS has been characterized. Therefore, by combining Corollary~\ref{ESFPDFCoro} and Lemma~\ref{COPThmx0}, we are able to compute the coverage probability of the typical user under the nearest BS association scheme. 

In Corollary~\ref{ESFPDFCoro}, the probability density function for the distance from the origin to its nearest BS has been characterized. Therefore, by combining Corollary~\ref{ESFPDFCoro} and Lemma~\ref{COPThmx0}, we are able to compute the SIR distribution of the typical user under the nearest BS association scheme. 

\begin{theorem}\label{COPThm}
%The coverage probability of the typical user at the origin for SIR threshold $T$ is given by:
The SIR distribution of the typical user at the origin is given by:
\begin{align}\label{COPThmEq}
\small
\begin{split}
&\mathbb{P}(\text{SIR}(0,\Phi) >T ) \\
&= \int_{0}^{+\infty} \lambda 2\pi \left[\sum\limits_{n=0}^{+\infty} \frac{(-1)^n}{n!} \int_{(\mathbb{R}^2)^n} \det(K_{x_0}^{!}(x_i,x_j))_{1 \leq i,j \leq n}  \prod\limits_{i=1}^{n}[1-\frac{\mathbbm{1}_{|x_i| \geq r_0}}{1+ T\mathit{l}(x_i)/\mathit{l}(x_0) }]\bigg|_{x_0 = (r_0,0)}{\rm d}x_1...{\rm d}x_n\right] r_0  {\rm d}r_0.
\end{split}
\end{align}
\end{theorem}

\begin{proof}
When expressing the location of the closest BS to the typical user in polar form as $x^*(0) = (r_0,\theta)$, we know that $x^*(0)$ admits the probability density $\frac{{\rm d}\theta}{2\pi} f (r_0) {\rm d} r_0$, where $f(r_0)$ is given in Corollary~\ref{ESFPDFCoro}. Therefore, we have:
\begin{align*}
\mathbb{P}(\text{SIR}(0,\Phi) >T ) &= \int_{0}^{+\infty} \int_{0}^{2\pi} \mathbb{P}[\text{SIR}(0,\Phi) >T | x^*(0) = (r_0, \theta)]  \frac{1}{2\pi} f (r_0) {\rm d}\theta {\rm d} r_0 \\
&\overset{(a)}{=} \int_{0}^{+\infty} \mathbb{P}[\text{SIR}(0,\Phi) >T | x^*(0) = (r_0, 0)] f(r_0) {\rm d} r_0, 
\end{align*}
where (a) is because the DPP is stationary and isotropic, so that the angle of $x_0$ will not affect the result of $ \mathbb{P}[\text{SIR}(0,\Phi) >T | x^*(0) = (r_0, \theta)] $. It follows from~(\ref{PalmRelation}) that $\det(K_{x_0}^{!}(x_i,x_j))_{1 \leq i,j \leq n}  = \frac{1}{K(x_0,x_0)} \det(K(x_i,x_j))_{0 \leq i,j \leq n}$, then the proof is completed by applying Corollary~\ref{ESFPDFCoro} and Lemma~\ref{COPThmx0}.
\end{proof}

\begin{comment}
In fact, despite we skip some details in this derivation, but all our results on Laplace functional or pdf of the ESF hold true if $r_0$ is finite. Then we can construct a sequence, where each of its element is the integration from 0 to n. Then this sequence is increasing and upper bounded since its always bounded by 1 as they represent some probability. Therefore, we know this sequence converges to its supreme, which is essentially the integration from 0 to infinity of the overall function. Since the pdf goes to 0 as $r_0$ is infinity, then we're good to express everything as we have derived before. 
In addition, the infinity point has measure 0, while the expression of pdf is also going to 0 when r0 is infinity, so the integration at infinity will not affect the result 
\end{comment}

\begin{remark}
If we choose $\Phi$ as a stationary PPP with intensity $\lambda$, i.e., $\det(K(x_i,x_j))_{1 \leq i,j \leq n} = \det(K_{x_0}^{!} (x_i,x_j))_{1 \leq i,j \leq n}= \lambda^n$, then Theorem~\ref{COPThm} leads to the same result as~\cite[Theorem 2]{trac}. 
\end{remark}

\section{Numerical Evaluation using Quasi-Monte Carlo Integration Method}\label{SectNumerical}
In this section, we provide the numerical method used to evaluate the analytical results derived in Section~\ref{SectAnalyze}. The Laplace functional of DPPs involves a series representation, where each term is a multi-dimensional integration. Therefore, we adopt the Quasi-Monte Carlo (QMC) integration method~\cite{kuo2005lifting} for efficient numerical integration.

The QMC integration method approximates the multi-dimensional integration of function $f: [0,1]^n \rightarrow \mathbb{R}$ as:
\begin{equation*}
\int_{[0,1]^n} f(\textbf{x}) d\textbf{x} \approx \frac{1}{N} \sum_{n=0}^{N-1} f(\textbf{x}_n).
\end{equation*}
The sample points $\textbf{x}_0,...,\textbf{x}_{N-1} \in [0,1]^n$ are chosen deterministically in the QMC method, and we use the Sobol points generated in MATLAB as the choice for sample points~\cite{blaszczyszyn2014studying}. Compared to the regular Monte Carlo integration method which uses a pseudo-random sequence as the sample points, the QMC integration method converges much faster. 

%BS deployments in two major US cities are investigated in this section. Fig.~\ref{RealBS} shows the BS deployment\footnote{BS location data was provided by a major tower owner in the USA.} of 115 BSs in a 16 km $\times$ 16 km area of Houston, as well as the deployment of 184 BSs in a 28 km $\times$ 28 km area of Los Angeles (LA). Both deployments are for sprawling and relatively flat areas, where repulsion among BSs is expected. Based on the maximum likelihood (ML) estimate method which is implemented in the software package provided in~\cite{rubak}, we have summarized the estimated parameters for different DPPs fitted to the Houston and LA data set in Table~\ref{par1} and Table~\ref{par2}. 

%\begin{table}[h]
%\center \caption{DPP Parameters for the Houston Data Set}
%\begin{tabular}{|c|c|c|c|p{60mm}|}
%\hline  Model & $\lambda$ & $\alpha$ & $\nu$ \\ 
%\hline  Gauss DPP  &  0.4492  & 0.8417 & $-$ \\  
%\hline  Cauchy DPP & 0.4492 & 1.558 &  3.424 \\ 
%\hline Generalized Gamma DPP  & 0.4492 &  2.539 & 2.63\\ 
%\hline 
%\end{tabular}  
%\label{par1}
%\end{table}

%\begin{table}[h]
%\center \caption{DPP Parameters for the LA Data Set}
%\begin{tabular}{|c|c|c|c|p{60mm}|}
%\hline  Model & $\lambda$ & $\alpha$ & $\nu$ \\ 
%\hline  Gauss DPP  &  0.2347  & 1.165 & $-$ \\ 
%\hline  Cauchy DPP & 0.2347 & 2.13 &  3.344 \\ 
%\hline Generalized Gamma DPP  & 0.2347 &  3.446 & 2.505\\ 
%\hline 
%\end{tabular}  
%\label{par2}
%\end{table}

In the following, we will focus on the numerical results using the Gauss DPP fitted to the Houston and LA data set. 
The modeling accuracy of the fitted Gauss DPPs compared to the real data sets will be validated in Section~\ref{SECModelAccu}. In addition, our simulation results for each metric are based on the average of 1000 realizations of the fitted Gauss DPP. %Realizations of the Gauss DPP, Cauchy DPP and Generalized Gamma DPP fitted to the Houston urban area deployment are shown in Fig.~\ref{BSfit}. From these figures, it can be observed that the fitted DPPs are regularly distributed and close to the real BS deployments. In Section~\ref{SECModelAccu}, we will rigorously validate the accuracy of these DPPs based on different summary statistics.

%\begin{figure}
%        \centering
%        \includegraphics[height=2in, width=2.2in]{BSHouston.eps}
%        \includegraphics[height=2in, width=2.2in]{BSLA.eps}
%        \caption{BS deployments of Houston (left) and LA (right).}\label{RealBS}
%\end{figure}

%\begin{figure}
%        \centering                
%                \includegraphics[height=1.8in, width=2in]{GaussDPPHou.eps}
%                 \includegraphics[height=1.8in, width=2in]{CauchyDPPHou.eps}
%                 \includegraphics[height=1.8in, width=2in]{GGDPPHou.eps}                             
%        \caption{Gauss DPP (left), Cauchy DPP (middle) and Generalized Gamma DPP (right) fitted to the Houston BS %deployment.}\label{BSfit}
%\end{figure}

\subsection{Empty Space Function}
Since the QMC integration method requires integration over the unit square,~(\ref{ESFEq}) can be rewritten as:
\begin{equation}~\label{ESFEqUnit}
\small
F(r) = \sum_{n=1}^{+ \infty} \frac{(-1)^{n-1} (2r)^{2n}}{n!} \int_{\left([0,1] \times [0,1]\right)^n} \det\left(K_0(2r(x_i - x_j))\right)_{1 \leq i,j \leq n } \prod_{i} \mathbbm{1}_{\{\|x_i-(\frac{1}{2},\frac{1}{2})\| \leq \frac{1}{2}\}}{\rm d}x_1...{\rm d}x_n,
\end{equation}
where $K_0(x)$ is the covariance function for the DPP $\Phi$. 

The accuracy of~(\ref{ESFEqUnit}) is verified by computing the empty space function of the Gauss DPP fitted to the Houston and LA data set respectively. Specifically, for the Gauss DPP model, $K_0(x) = \lambda \exp(-\|x/ \alpha\|^2)$, where $\lambda$ and $\alpha$ are chosen according to Table~\ref{par1} and Table~\ref{par2}. Fig.~\ref{ESFGauss} shows the QMC integration results of~(\ref{ESFEqUnit}) with different numbers of Sobol points, as well as the simulation result for the fitted Gauss DPP. We have observed that when the number of Sobol points is $2^{11}$, (\ref{ESFEqUnit}) can be computed very efficiently (in a few seconds) and the QMC integration results are accurate except for the part where $F(r)$ is over 95\%. In contrast, if the number of Sobol points is increased to $2^{15}$, the QMC integration method is almost 10 times slower while the results are accurate for a much larger range of $r$. %This observation is consistent with the discussion before, where small number of Sobol points is sufficient to achieve accurate QMC integration results for small $r$, while more Sobol points are needed to achieve accurate QMC integration result when $r$ grows larger. 

%In fact, since the DPP kernel is a positive definite Hermitian matrix, the Hadamard determinant theorem~\cite[p. 300]{marshall2010inequalities} implies that $\det\left(K(x_i,x_j)\right)_{1 \leq i,j \leq n } \leq \prod_{i=1}^{n} K(x_i,x_i) = \lambda^n$ for a stationary DPP with intensity $\lambda$. Therefore, the absolute value of the $n$-th term in~(\ref{ESFEq}) can be upper bounded as:
%\begin{equation*}
%\biggl|\frac{(-1)^{n-1}}{n!} \int_{\left(B(0,r)\right)^n} \det\left(K(x_i,x_j)\right)_{1 \leq i,j \leq n }\biggl| \leq \frac{1}{n!} \left(\pi \lambda r^2 \right)^n .
%\end{equation*}
%When $r$ is small enough such that $\lambda \pi r^2 <1$, each term in~(\ref{ESFEq}) vanishes very fast and the QMC integration method converges quickly, which means that a small number of Sobol points is sufficient to obtain accurate QMC integration results. 

%\begin{figure}
%        \centering
%                \includegraphics[height=2in, %width=3.2in]{ESFGaussDPPHou2.eps}
%                \includegraphics[height=2in, width=3.2in]{ESFGaussDPPLA2.eps}
%        \caption{Empty space function of the Gauss DPP fitted to the Houston (left) and LA (right) data set. }\label{ESFGauss}
%\end{figure}

 \begin{figure}
         \begin{subfigure}[b]{0.47\textwidth}
                 \includegraphics[height=2in, width=3.2in]{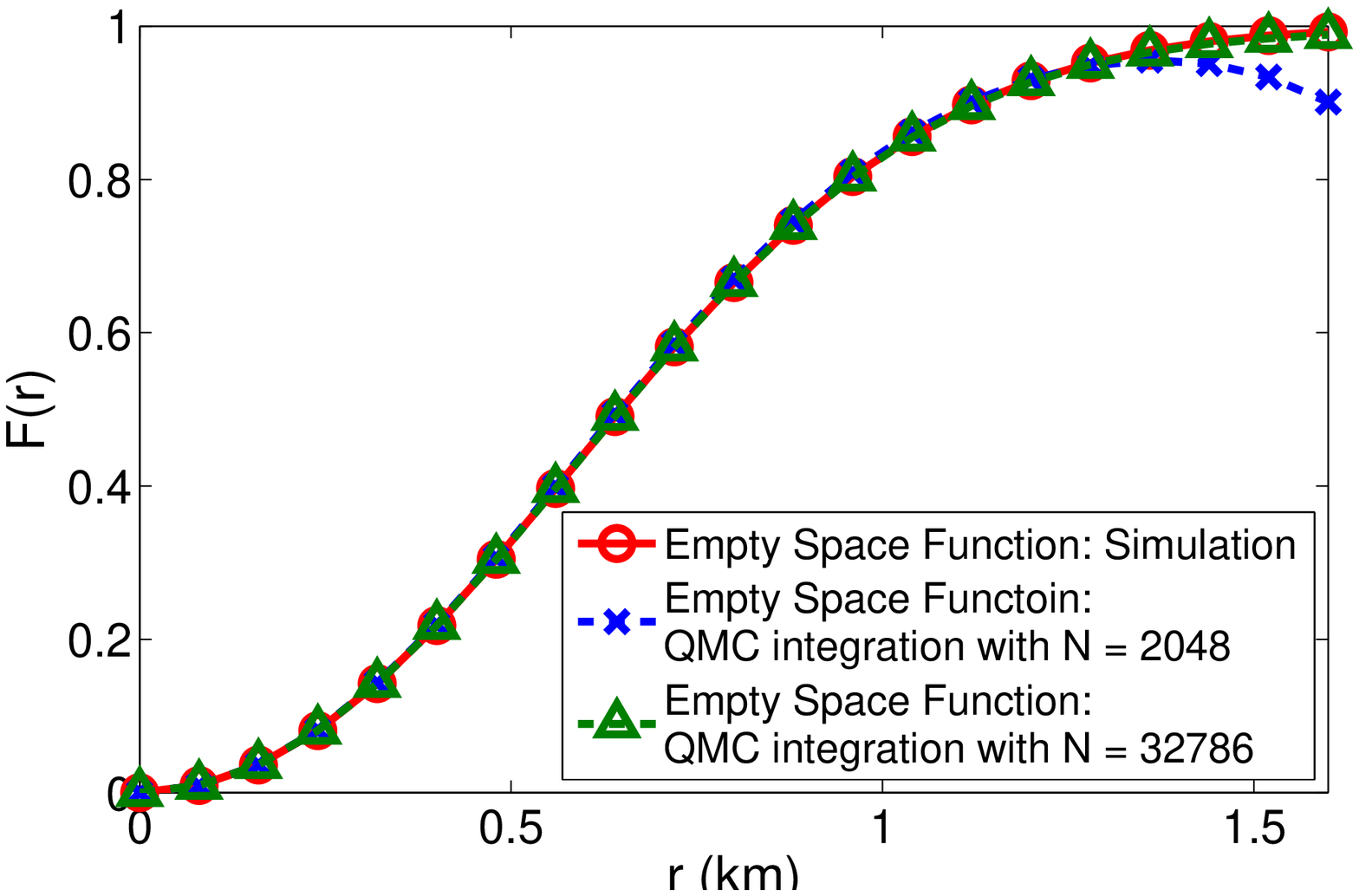}
                 \caption{Houston data set}
 %                \label{GGfitKest}
         \end{subfigure}
         \hfill
         \begin{subfigure}[b]{0.47\textwidth}
                 \includegraphics[height=2in, width=3.2in]{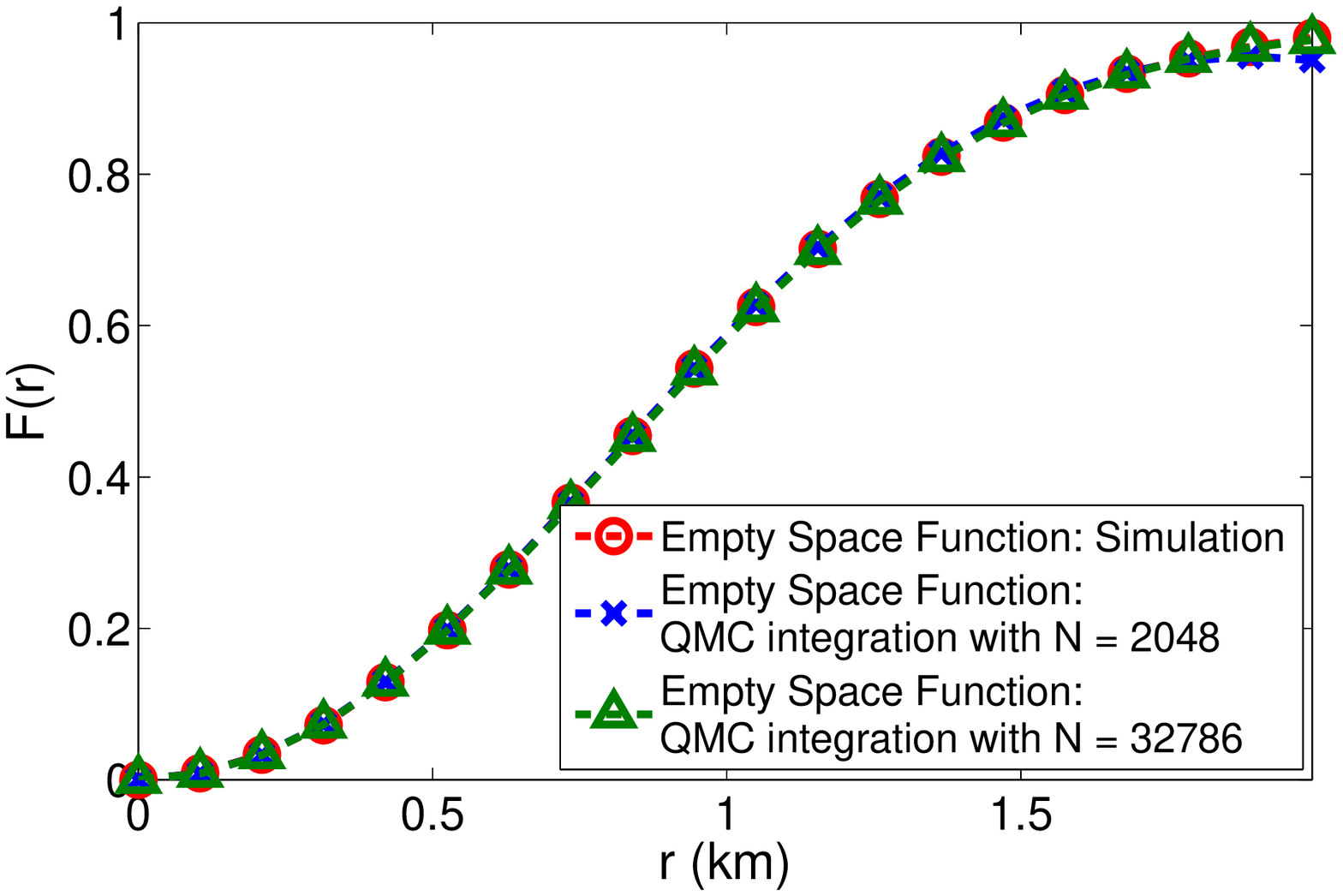}
                 \caption{LA data set}
 %                \label{coverageGG}
         \end{subfigure}
         \caption{Empty space function of the fitted Gauss DPP.}\label{ESFGauss}
 \end{figure}

\subsection{Nearest Neighbor Function}
%The nearest neighbor function for the fitted Gauss DPP is also evaluated using the QMC integration method, which is shown in Fig.~\ref{NNGauss}. The nearest neighbor function for a stationary PPP with intensity $\lambda$ is also provided, which is given by $D_{\text{PPP}} (r) = 1 - \exp(-\lambda \pi r^2)$ from Slivnyak's theorem. 

%\begin{figure}[h]
%        \centering
%                \includegraphics[height=2in, width=3.2in]{NNGaussDPPHou2.eps}
%                \includegraphics[height=2in, width=3.2in]{NNGaussDPPLA2.eps}
%        \caption{Nearest neighbor function of the Gauss DPP fitted to the Houston (left) and LA (right) data set.}\label{NNGauss}
%\end{figure}

\begin{figure}
        \begin{subfigure}[b]{0.47\textwidth}
                \includegraphics[height=2in, width=3.2in]{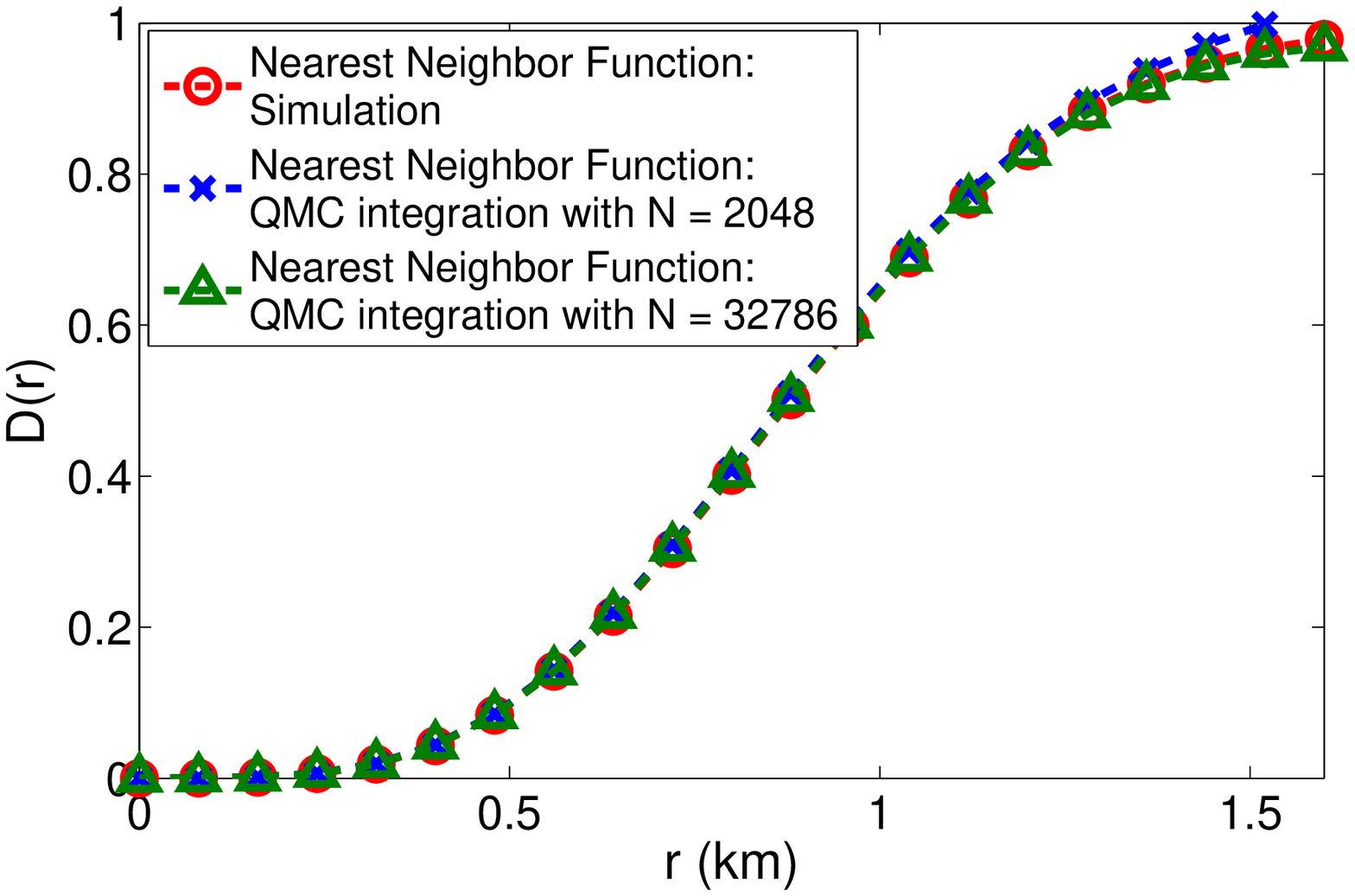}
                \caption{Houston data set}
%                \label{GGfitKest}
        \end{subfigure}
        \hfill
        \begin{subfigure}[b]{0.47\textwidth}
                \includegraphics[height=2in, width=3.2in]{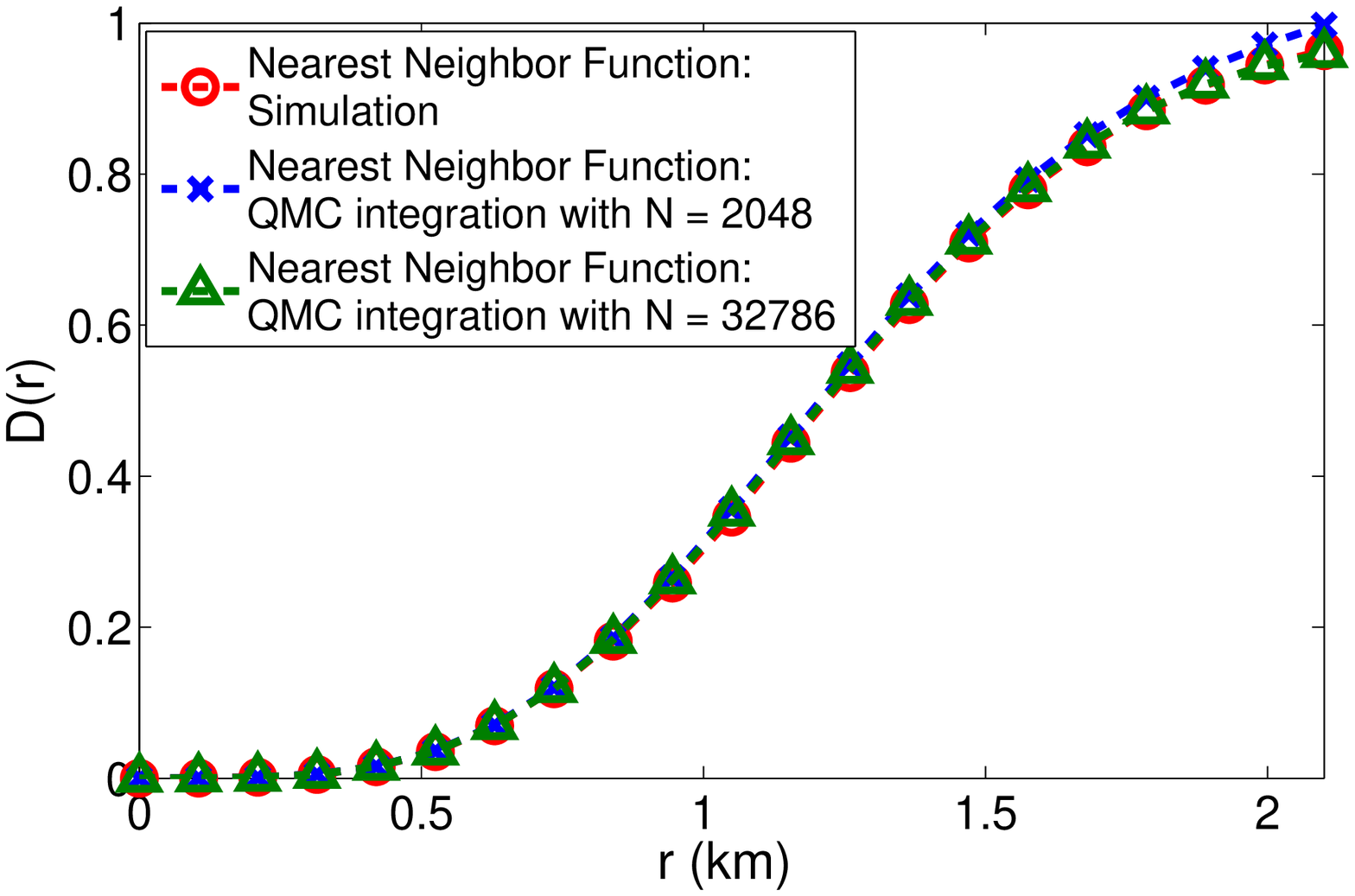}
                \caption{LA data set}
%                \label{coverageGG}
        \end{subfigure}
        \caption{Nearest neighbor function of the fitted Gauss DPP.}\label{NNGauss}
\end{figure}

The QMC integration method is also efficient in the numerical evaluation of the nearest neighbor function. Similar to the empty space function, the QMC integration method with $N=2^{11}$ takes a few seconds to return $D(r)$ in Fig.~\ref{NNGauss}, which is accurate up to 95\%. By contrast, the QMC integration method is more accurate but almost ten times slower when the number of sample points is increased to $N = 2^{15}$.

%For a Gauss DPP $\Phi$ with parameter $\lambda$ and $\alpha$, its corresponding DPP $\tilde{\Phi}$ under reduced Palm distribution will have the kernel:
%\begin{equation*}
%K_{o}^{!} (x,y) = \lambda \left(\exp(-\frac{\|x-y\|^2}{\alpha^2}) - \exp(-\frac{\|x\|^2+\|y\|^2}{\alpha^2})\right).
%\end{equation*}

%Through change of variables to make the integration in~(\ref{NNEq}) is over unit square area, QMC method is used to evaluate the accuracy of Lemma~\ref{NN}. In order to explain the repulsive behavior for Gauss DPP, the nearest neighbor function of stationary PPP with intensity $\lambda$ is also considered, which is given as $D_{\text{PPP}} (r) = 1 - \exp(-\lambda \pi r^2)$ by Slyvniak theorem. 

\subsection{Mean Interference}
In this part, the mean interference of the Gauss DPP is numerically evaluated for the two BS association schemes discussed in Section~\ref{SubSecIntDist}. The path loss model is chosen as $\mathit{l}(x) = \min (1, |x|^{-\beta})$, where $\beta > 2$ is the path loss exponent.
\subsubsection{Mean interference with fixed associated BS scheme}
\begin{corollary}\label{MeanIntCoro}
Conditionally on $x_{0} = (r_0,0)$ as the serving BS for the typical user, the mean interference at the typical user when BSs are distributed according to the Gauss DPP with parameters $(\lambda, \alpha)$ is given by:
\begin{equation*}
%\mathbb{E}_{x_0}^{!}[I] = \frac{\pi \lambda \beta d_{0}^{2-\beta}}{\beta-2} - 2\pi \lambda \exp(-\frac{2r_{0}^{2}}{\alpha^2}) (A_1(r_0)+A_2(r_0)),
\mathbb{E}[I |  x_{0} = (r_0,0)] = \frac{P \pi \lambda \beta}{\beta-2} - 2 P \pi \lambda \exp(-\frac{2r_{0}^{2}}{\alpha^2}) (A_1(r_0)+A_2(r_0)),
\end{equation*}
where  %$A_1(r_0) = \int_{0}^{d_0} \exp(-2\frac{r^2}{\alpha^2}) I_0 (\frac{4rr_0}{\alpha^2}) rdr$, and $A_2(r_0) = \int_{d_0}^{\infty} \exp(-2\frac{r^2}{\alpha^2}) r^{1-\beta} I_0(\frac{4rr_0}{\alpha^2}){\rm d}r$, 
$A_1(r_0) = \int_{0}^{1} \exp(-\frac{2 r^2}{\alpha^2}) I_0 (\frac{4rr_0}{\alpha^2}) r {\rm d}r$, and $A_2(r_0) = \int_{1}^{\infty} \exp(-\frac{2 r^2}{\alpha^2}) r^{1-\beta} I_0(\frac{4rr_0}{\alpha^2}) {\rm d}r$. Here $I_0(\cdot)$ denotes the modified Bessel function of first kind with parameter $\nu = 0$~\cite{integrate}. 
\end{corollary}

\begin{proof}
Based on the fact that $\int_{0}^{2\pi} \exp(\pm \beta \cos(x)) {\rm d} x = 2\pi I_0(\beta)$~\cite[p. 491]{integrate}, this corollary can be derived by substituting the Gauss DPP kernel into Lemma~\ref{MeanInt}.
\end{proof}

In Fig.~\ref{MeanIntFig}, the mean interference for the Gauss DPP fitted to the Houston and LA data sets are provided under different path loss exponents with $P=1$. 
%From Fig.~\ref{ESFGauss}, it can be observed that $r$ = 1.4 km (1.8 km) corresponds to the 95\% percentile of the empty space function for Houston (LA) data set, so we choose the maximum of $r_0$ as 1.4 km (1.8 km). 
From Fig.~\ref{MeanIntFig}, it can be observed that the mean interference increases as $r_0$ increases; this is because it increases the probability for the existence of a strong interferer close to the typical user. In addition, given $r_0$, the mean interference is decreasing when the path loss exponent $\beta$ increases; this is because the path loss function is decreasing with respect to $\beta$ for all interferers.
%In addition, a bias value is added to the simulation results, since the simulation can only generate finite window size, while Corollary~\ref{MeanIntCoro} integrates over the whole plane. 

%\begin{figure}
%        \centering
%                \includegraphics[height=2in, width=3.2in]{MeanIntGaussHou2.eps}
%                \includegraphics[height=2in, width=3.2in]{MeanIntGaussLA2.eps}
%        \caption{Mean interference under the fixed associated BS scheme for the Houston (left) and LA (right) data set.}\label{MeanIntFig}
%\end{figure}

\begin{figure}
        \begin{subfigure}[b]{0.47\textwidth}
                \includegraphics[height=2in, width=3.2in]{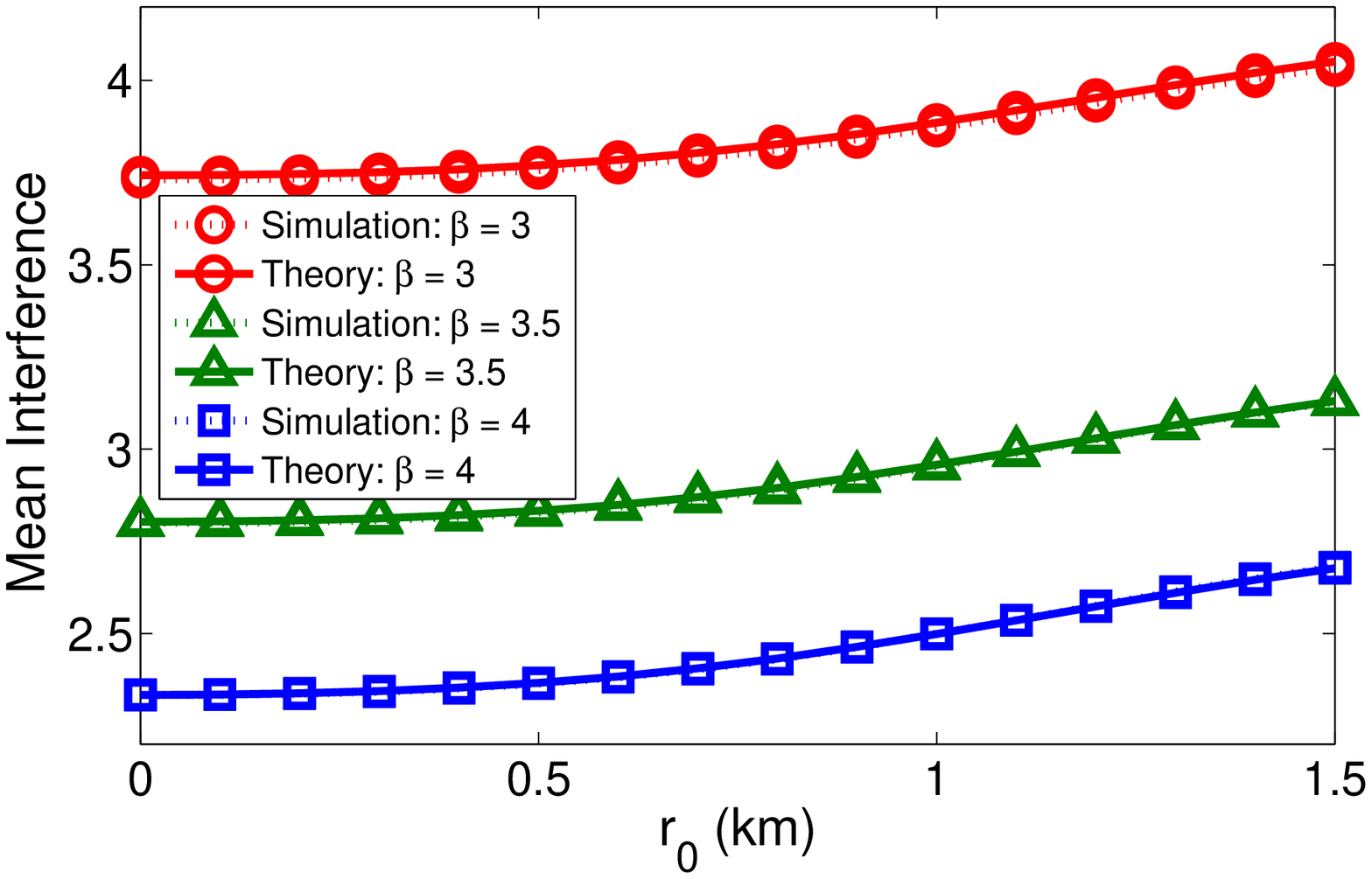}
                \caption{Houston data set}
%                \label{GGfitKest}
        \end{subfigure}
        \hfill
        \begin{subfigure}[b]{0.47\textwidth}
                \includegraphics[height=2in, width=3.2in]{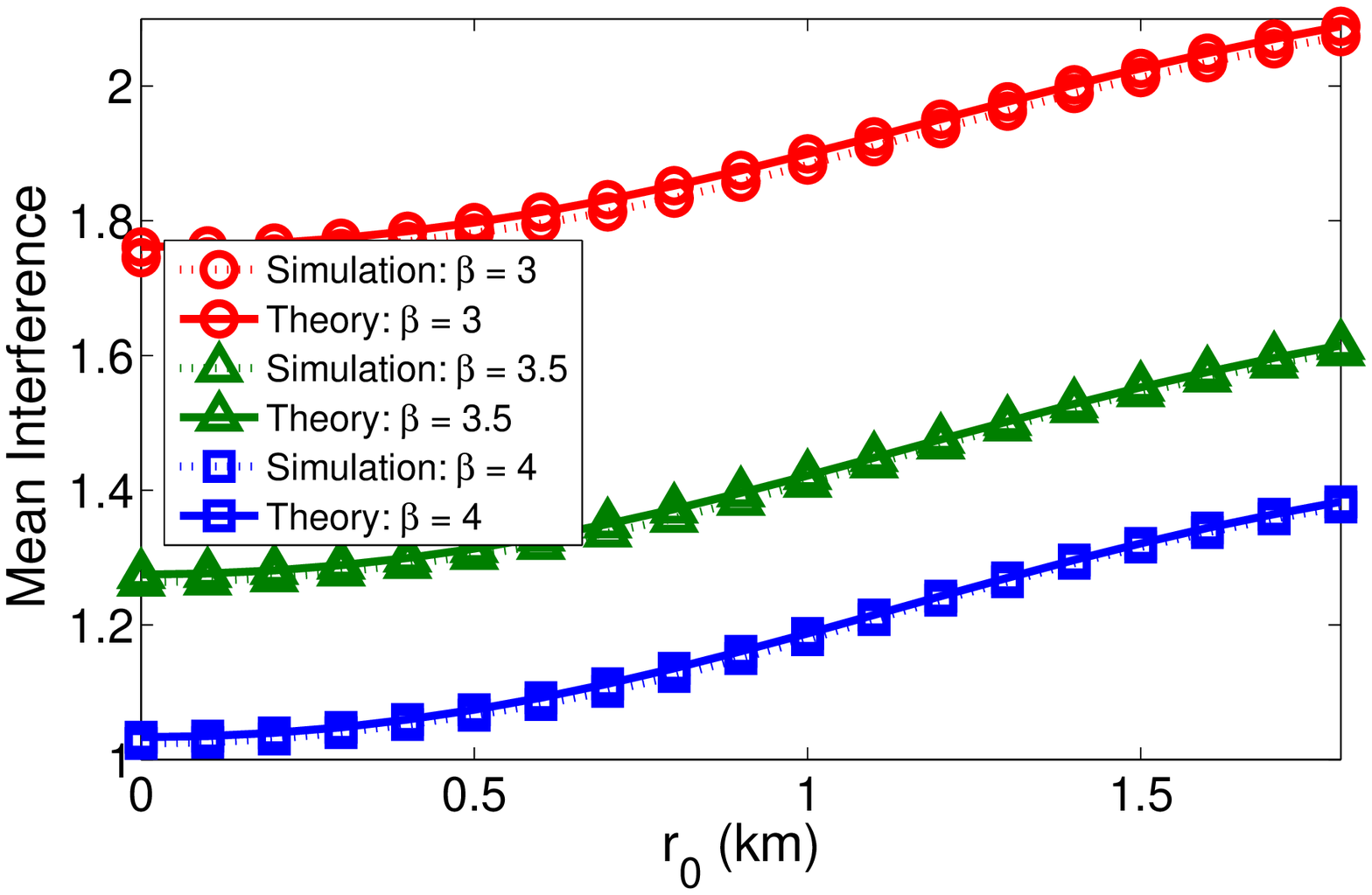}
                \caption{LA data set}
%                \label{coverageGG}
        \end{subfigure}
        \caption{Mean interference under the fixed associated BS scheme.}\label{MeanIntFig}
\end{figure}
\subsubsection{Mean interference with nearest BS association scheme}
The Quasi-Monte Carlo integration method is adopted to evaluate the mean interference under the nearest BS association scheme, which is given in~(\ref{MeanIntx0}). 
%A change of variables is used to make~(\ref{MeanIntx0}) integrate over the unit square area. In particular, the denominator in~(\ref{MeanIntx0}) can be rewritten similarly to~(\ref{ESFEqUnit}), while the numerator can be rewritten as:
%\begin{equation*}
%\small
%\sum\limits_{n=0}^{+\infty} \frac{2\pi(-1)^n}{n!} (2r_0)^{2n}\int_{([0,1] \times [0,1])^{n+1}} \det(K_{x_0}^{!}(\bar{x}_i,\bar{x}_j))_{1 \leq i,j \leq n+1} \frac{r_{0}^{2}}{r^{3}} P\mathit{\tilde{l}}(\frac{r_0}{r})\prod\limits_{i=2}^{n+1} \mathbbm{1}_{\{|x_i - (\frac{1}{2},\frac{1}{2})| \leq \frac{1}{2}\}} {\rm d}r {\rm d}\theta ...{\rm d}x_{n+1},
%\end{equation*}
%where $\bar{x}_1 = (\frac{r_0}{r}\cos(2\pi \theta) , \frac{r_0}{r} \sin(2\pi \theta))$,
% this is derived based on the Jacobian that in the original integration, let $x_1 = (\frac{1}{r} \cos(\theta), \frac{1}{r}\cos(\theta))$
%$\bar{x}_{i} = 2r_0(x_i - (\frac{1}{2},\frac{1}{2}))$, for $i \geq 2$ and $\mathit{\tilde{l}}(r)$ is the radially symmetric path loss function. 
In Fig.~\ref{MeanIntnNBSFig}, the mean interference is evaluated when the path loss exponent $\beta$ is 3, 3.5, 4. It can be observed from Fig.~\ref{MeanIntnNBSFig} that when $r_0$ (i.e., the distance from the typical user to its nearest BS) increases, the mean interference decreases. This is because the strong interferers are farther away from the typical user when $r_0$ increases, which leads to a smaller aggregate interference. This is quite different from the case when the BS associated to the typical user is assumed to be at some fixed location. In addition, since the path loss function $\mathit{l}(x)$ is non-increasing with respect to $\beta$ given the norm of $x$, the mean interference decreases when $\beta$ increases for a given $r_0$. 

%In addition, when $r_0$ is small enough (e.g., $r_0 \ll 0.4 km$), the mean interference is higher when $\beta$ is larger; while when $r_0$ is large enough (e.g., $r_0 \gg 0.4 km$ ), the mean interference is higher when $\beta$ is smaller. This is because when $r_0$ is small enough, the strong interferes which are close to the origin will dominate the interference, thus larger path loss exponent corresponds to stronger interference. On the other hand, when $r_0$ is large enough, there is no strong interferer close to the origin, and thus the larger path loss exponent corresponds to stronger interference power attenuation and thus smaller mean interference.

%\begin{figure}
%        \centering
%                \includegraphics[height=2in, width=3.2in]{MeanIntNearestBSGaussDPPHou2.eps}
%                \includegraphics[height=2in, width=3.2in]{MeanIntNearestBSGaussDPPLA2.eps}
%        \caption{Mean interference under the nearest BS association scheme for the Houston (left) and LA (right) data set.}\label{MeanIntnNBSFig}
%\end{figure}

\begin{figure}
        \begin{subfigure}[b]{0.47\textwidth}
                \includegraphics[height=2in, width=3.2in]{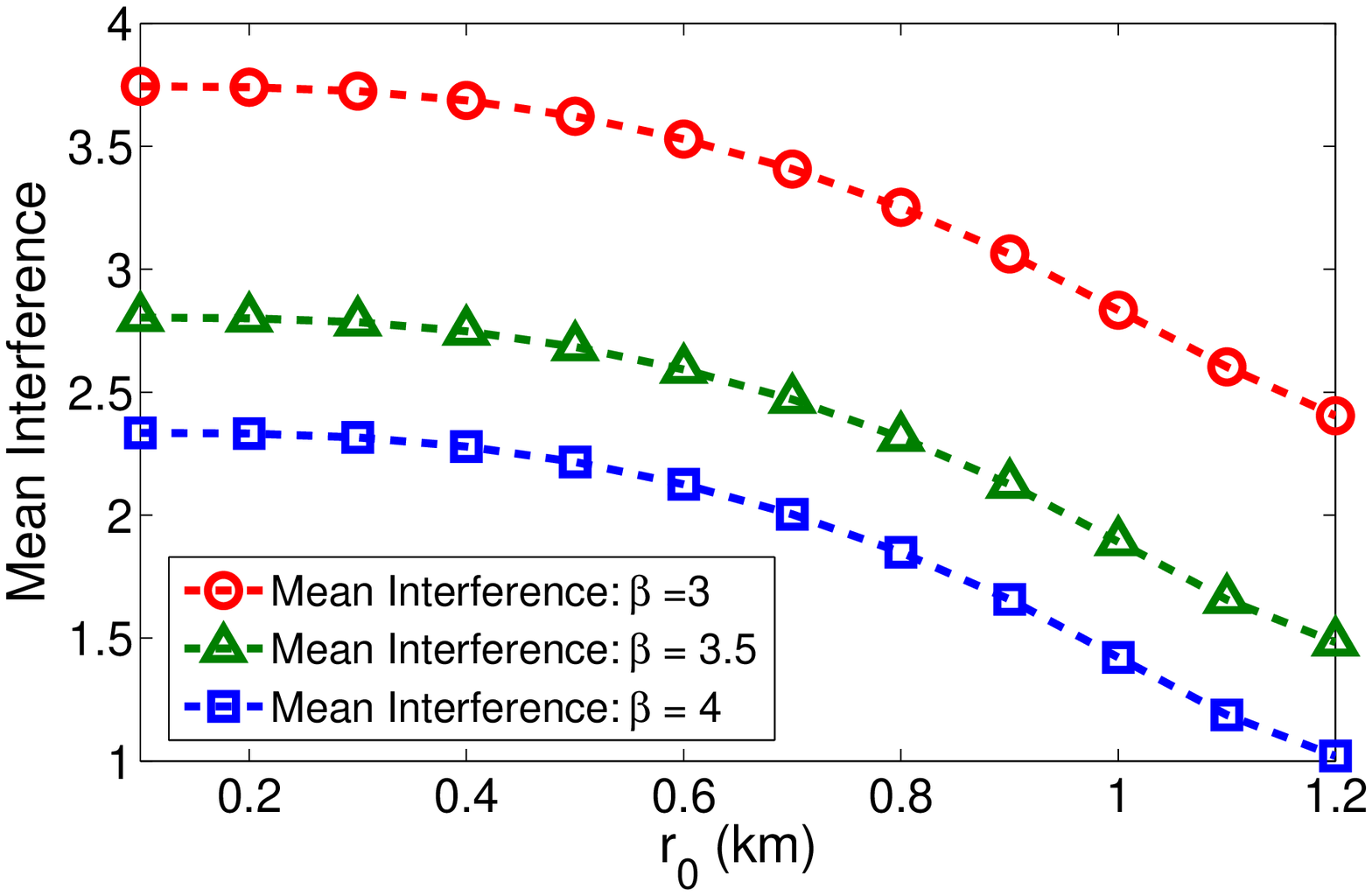}
                \caption{Houston data set}
%                \label{GGfitKest}
        \end{subfigure}
        \hfill
        \begin{subfigure}[b]{0.47\textwidth}
                \includegraphics[height=2in, width=3.2in]{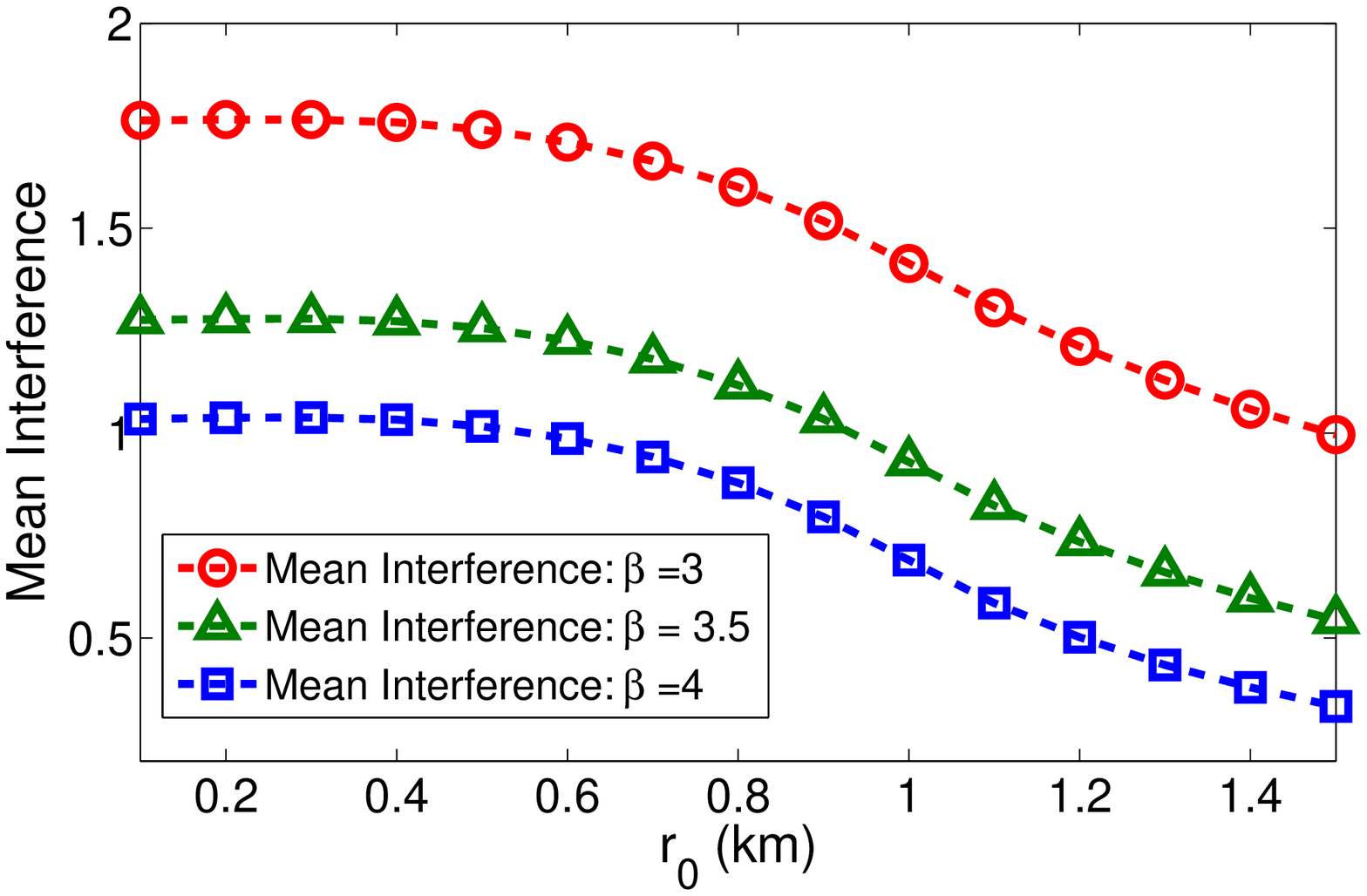}
                \caption{LA data set}
%                \label{coverageGG}
        \end{subfigure}
        \caption{Mean interference under the nearest BS association scheme.}\label{MeanIntnNBSFig}
\end{figure}

\subsection{SIR Distribution}~\label{COPNumSubSec}
%Theorem~\ref{COPThm} can also be evaluated using Quasi-Monte Carlo integration method (note: this part will be finished before this work is submitted).
The QMC integration method can, in principle, be used to numerically evaluate~(\ref{COPThmEq}). However, it is time consuming due to the need to evaluate multiple integrations over $\mathbb{R}^2$. Therefore, we use the diagonal approximation of the matrix determinant~\cite{ipsen2011determinant} to roughly estimate~(\ref{COPThmEq}). Specifically, the determinant of matrix $\left(K(x_{i},x_{j})\right)_{1\leq i,j\leq n}$ is approximated\footnote{The relative error bound for diagonal approximation is provided in~\cite[Theorem 1]{ipsen2011determinant}.} as $\det (\left(K(x_{i},x_{j})\right)_{1\leq i,j\leq n}) \approx \prod_{i=1}^{n} K(x_{i},x_{i})$ under the diagonal approximation.

\begin{lemma}\label{COPThmDiagApproxLemma}
%Under the diagonal approximation, the coverage probability for the typical user with SIR threshold $T$ is approximated as:
Under the diagonal approximation, the SIR distribution for the typical user is approximated as:
\begin{equation}\label{COPThmDiagApproxEq}
\mathbb{P}(\text{SIR}(0,\Phi) >T ) \approx \int_{0}^{+\infty} \lambda 2 \pi r_0 \exp\left(-\int_{\mathbb{R}^2} K_{x_0}^{!}(x,x) (1-\frac{\mathbbm{1}_{|x| \geq r_0}}{1+ T\mathit{l}(x)/\mathit{l}(x_0)}) \bigg|_{x_0 = (r_0,0)}{\rm d} x\right) {\rm d} r_0.
\end{equation}
\end{lemma}

Lemma~\ref{COPThmDiagApproxLemma} can be proved by applying diagonal approximation to Theorem~\ref{COPThm}, thus we omit the proof. 

Next, we evaluate the accuracy of Lemma~\ref{COPThmDiagApproxLemma} by assuming the BSs are distributed according to the Gauss DPP. In addition, the power-law path loss model with path loss exponent $\beta >2$ is used for simplicity, i.e., $\mathit{l}(x)=\|x\|^{-\beta}$ for $x \in \mathbb{R}^2$.

\begin{corollary}
%When BSs are distributed according to the Gauss DPP with parameters $(\lambda,\alpha)$, the coverage probability can be approximated under the diagonal approximation as:
When BSs are distributed according to the Gauss DPP with parameters $(\lambda,\alpha)$, the SIR distribution can be approximated under the diagonal approximation as:
\begin{align}
\allowdisplaybreaks
\mathbb{P}(\text{SIR}(0,\Phi) >T ) \approx & \int_{0}^{+\infty} \lambda 2 \pi r_0 \exp\biggl(-\lambda 2 \pi \biggl[ \int_{0}^{r0} (1-\exp(-\frac{2(r^2+r_{0}^{2})}{\alpha^2})I_0(\frac{4rr_0}{\alpha^2}))r{\rm d}r \nonumber\\
&+\int_{r_0}^{+\infty}(1-\exp(-\frac{2(r^2+r_{0}^2)}{\alpha^2})I_0(\frac{4rr_0}{\alpha^2}))\frac{Tr_{0}^{\beta}r}{Tr_{0}^{\beta}+r^{\beta}} {\rm d}r\biggl]\biggl){\rm d} r_0 ,\label{DiagApproxGaussEq}
\end{align}
\end{corollary}
where $I_0(\cdot)$ denotes the modified Bessel function of first kind with parameter $\nu = 0$.
\begin{proof}
Based on the fact that $\int_{0}^{2\pi} \exp(\pm \beta \cos(x)) {\rm d} x = 2\pi I_0(\beta)$~\cite[p. 491]{integrate}, this corollary can be derived by substituting the Gauss DPP kernel into Lemma~\ref{COPThmDiagApproxLemma}. 
\end{proof}

%\begin{figure}
%        \centering
%                \includegraphics[height=2in, width=3.3in]{COP_Hou_Approx2.eps}
%                \includegraphics[height=2in, width=3.3in]{COP_LA_Approx2.eps}
%        \caption{Diagonal approximation to the coverage probability of the Gauss DPP fitted to the Houston (left) and LA (right) data set.}\label{fig:COP_Approx}
%\end{figure}

\begin{figure}
        \begin{subfigure}[b]{0.47\textwidth}
                \includegraphics[height=2in, width=3.3in]{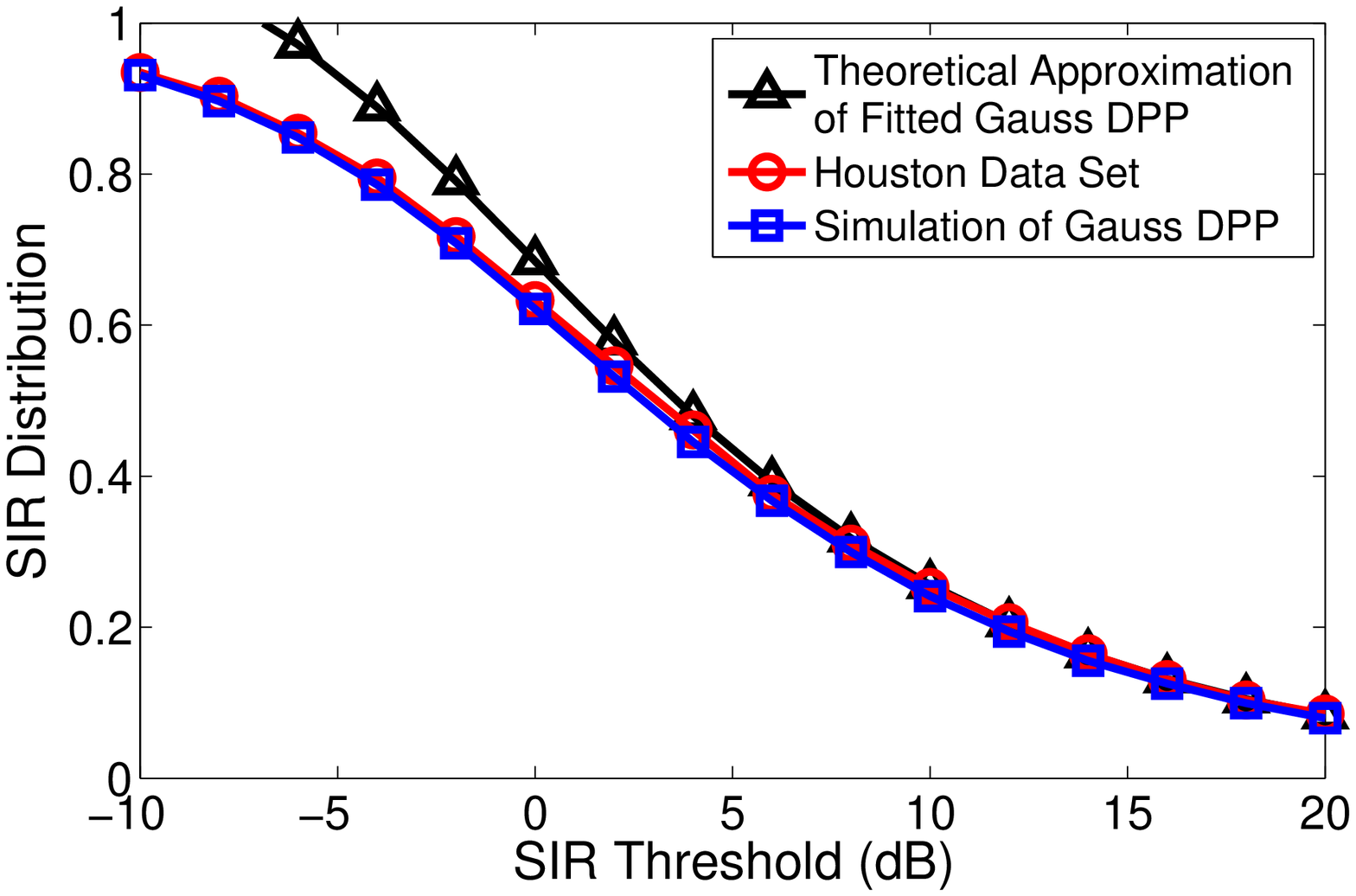}
                \caption{Houston data set}
%                \label{GGfitKest}
        \end{subfigure}
        \hfill
        \begin{subfigure}[b]{0.47\textwidth}
                \includegraphics[height=2in, width=3.3in]{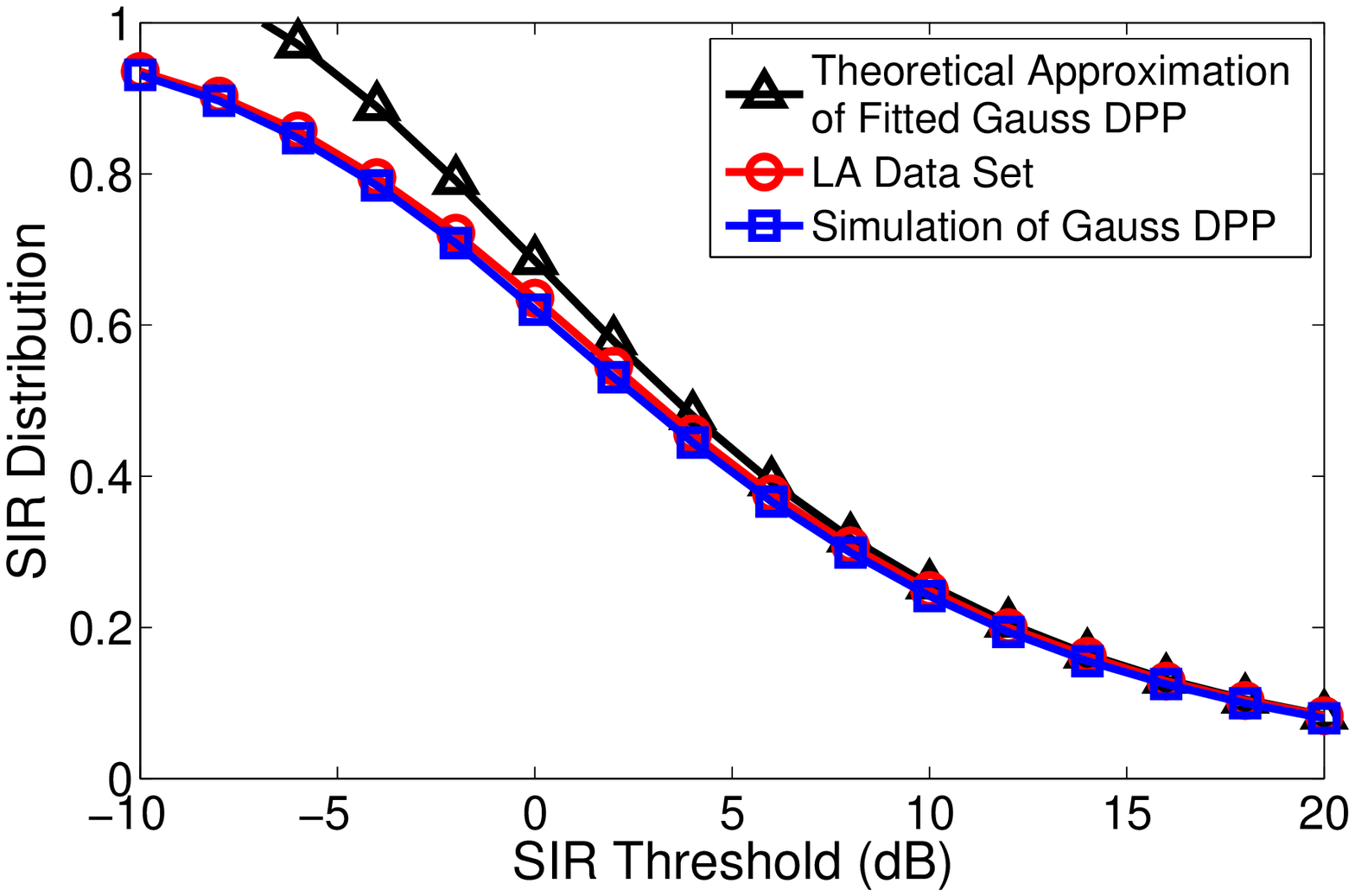}
                \caption{LA data set}
%                \label{coverageGG}
        \end{subfigure}
        \caption{Diagonal approximation to the SIR distribution of the fitted Gauss DPP.}\label{fig:COP_Approx}
\end{figure}

The QMC integration method is used to evaluate (\ref{DiagApproxGaussEq}) with path loss exponent $\beta = 4$, and the result for the Gauss DPP fitted to Houston data set is plotted in Fig.~\ref{fig:COP_Approx}. It can be observed that the diagonal approximation to the coverage probability is accurate compared to the simulation result in the high SIR regime, i.e., when the SIR threshold is larger than 6 dB. The same trend can also be found for the LA data set. Therefore, we can use the diagonal approximation as an accurate estimate for the coverage probability in the high SIR regime. 
%More efficient and accurate numerical evaluation of the theoretical coverage probability is left for our future work. 

%\section{Fitting Stationary DPP Models to Base Station Deployments}~
\section{Goodness-of-fit for Stationary DPPs to Model BS Deployments}~\label{SECModelAccu}
%DPPs are attractive models for cellular networks not only because they are mathematically tractable, but also because they can accurately model macro BS deployments. 
Given that the stationary DPP models are tractable, we provide rigorous investigation of their modeling accuracy to real BS deployments in this section. %In this section,  we fit the stationary DPP models proposed in Section~\ref{SubSectDPPEg} to real BS deployments and investigate their modeling accuracy. 
Our simulations are based on the publicly available package for DPP models~\cite{rubak} implemented in R, which is used as a supplement to the Spatstat library~\cite{R}.

\subsection{Summary Statistics}~\label{SubSecSummStat}
To test the goodness-of-fit of these DPP models, we have used Ripley's K function and the coverage probability as performance metrics, which are described below:

\textbf{Ripley's K function:} Ripley's K function is a second order spatial summary statistic defined for stationary point processes. It counts the mean number of points within distance $r$ of a given point in the point process excluding the point itself. Formally, the K function $\mathbb{K}(r)$ for a stationary and isotropic point process $\Phi$ with intensity $\lambda$ is defined as: 
\begin{align}
\allowdisplaybreaks
\mathbb{K}(r)= \frac{\mathbb{E}_{o}^{!}\left(\Phi(B(0,r))\right)}{\lambda},
\end{align}
where $\mathbb{E}_{o}^{!}(\cdot)$ is the expectation with respect to the reduced Palm distribution of $\Phi$. %, and $B(0,r)$ is the ball centered at origin with radius $r$. 

%The K-function can also be interpreted in terms of pair correlation function $g_{0}(r)$, i.e., $\mathbb{K}(r)= 2 \pi \int_{0}^{r} t g_{0}(t) dt$. As discussed in the previous section, the pair correlation function can be derived immediately from the kernel of a DPP, so we can have closed form expression for several DPP models. For example, the K-function for Gauss DPP can be derived as: 
%\begin{equation}
%\mathbb{K}(r)= \pi r^2 -\frac{\pi \alpha^2}{2} %\left(1-\exp(-\frac{2r^{2}}{\alpha^{2}})\right)
%\end{equation}

The K-function is used as a measure of repulsiveness/clustering of spatial point processes. Specifically, compared to the PPP which is completely random, a repulsive point process model will have a smaller K function, while a clustered point process model will have a larger K function. 
%Closed-form expressions of K function can be calculated for Gauss and Cauchy DPP~\cite{rubak}, which are strictly smaller than PPP. This explains the repulsiveness of these two DPP models. 

%\textbf{L function:} The L-function is also a commonly used summary statistics which is defined as $L(r)= \sqrt{K(r)/ \pi}$. Compared with L-function of PPP which equals $r$, $L(r) > r$ corresponds to clustering point process while $L(r) <r $ corresponds to repulsive point process. 

%K-function, L-function as well as pair correlation functions are all spatial summary statistics of any point process. 

\textbf{Coverage Probability:} The coverage probability is defined as the probability that the received SINR at the typical user is larger than the threshold $T$. When measuring the fitting accuracy of spatial point processes to real BS deployments, metrics related to the wireless system such as the coverage probability are more practical. In particular, the coverage probability also depends on the repulsive/clustering behavior of the underlying point process used to model the BS deployment. Compared to the fitted PPP, due a larger empty space function, the distance from the typical user to its serving BS is stochastically less in a fitted repulsive point process. Similarly, due to a smaller nearest neighbor function, the fitted repulsive point process has stochastically larger distance from the serving BS to its closest interfering BS than the PPP case. Therefore, from~(\ref{COPEq}), a larger coverage probability is expected when the BS deployments are modeled by more repulsive spatial point processes. 
We will use the same parameter assumptions as in Section~\ref{COPNumSubSec} for evaluating the coverage probability. Since the thermal noise power is assumed to be 0, the CCDF of SIR at the typical user, i.e., $\mathbb{P}(\text{SIR}(0,\Phi) > T)$,  coincides with its coverage probability with threshold $T$. 

\subsection{Hypothesis Testing using Summary Statistics}\label{SubSecStat}

In this part, we evaluate the goodness-of-fit of stationary DPP models using the summary statistics discussed above. Particularly, we fit the real BS deployments in Fig.~\ref{RealBS} to the Gauss, Cauchy and Generalized Gamma DPPs.

%\begin{figure}[h]
%        \centering
%                \includegraphics[height=2in, width=3.2in]{HouKestGaussDPP2.eps}
%                \includegraphics[height=2in, width=3.2in]{LAKestGaussDPP2.eps}
%        \caption{K function for the Gauss DPP fitted to the Houston (left) and LA (right) data set.}\label{GaussfitHou}
%\end{figure}

\begin{figure}
        \begin{subfigure}[b]{0.47\textwidth}
                \includegraphics[height=2in, width=3.2in]{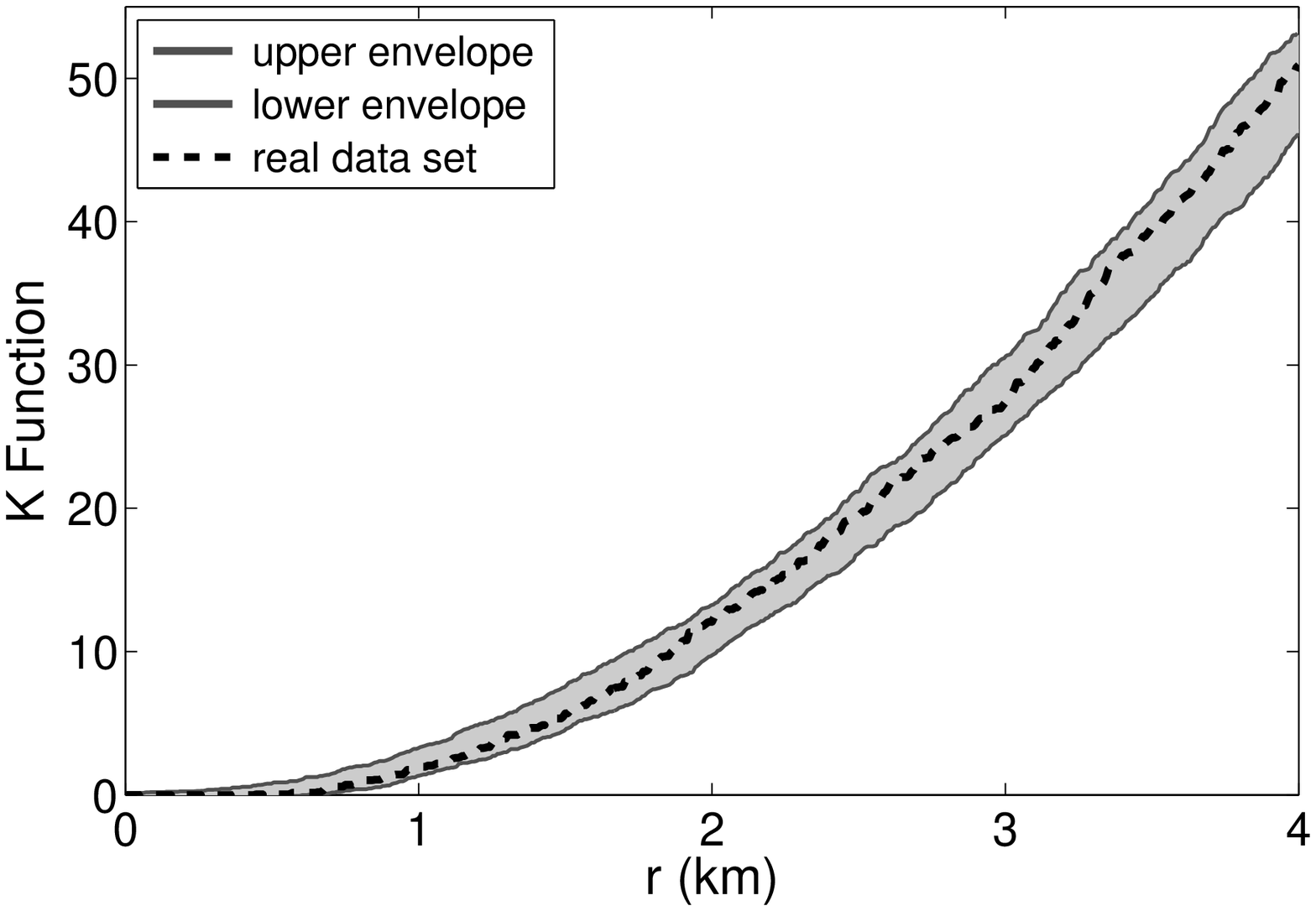}
                \caption{Houston data set}
%                \label{GGfitKest}
        \end{subfigure}
        \hfill
        \begin{subfigure}[b]{0.47\textwidth}
                \includegraphics[height=2in, width=3.2in]{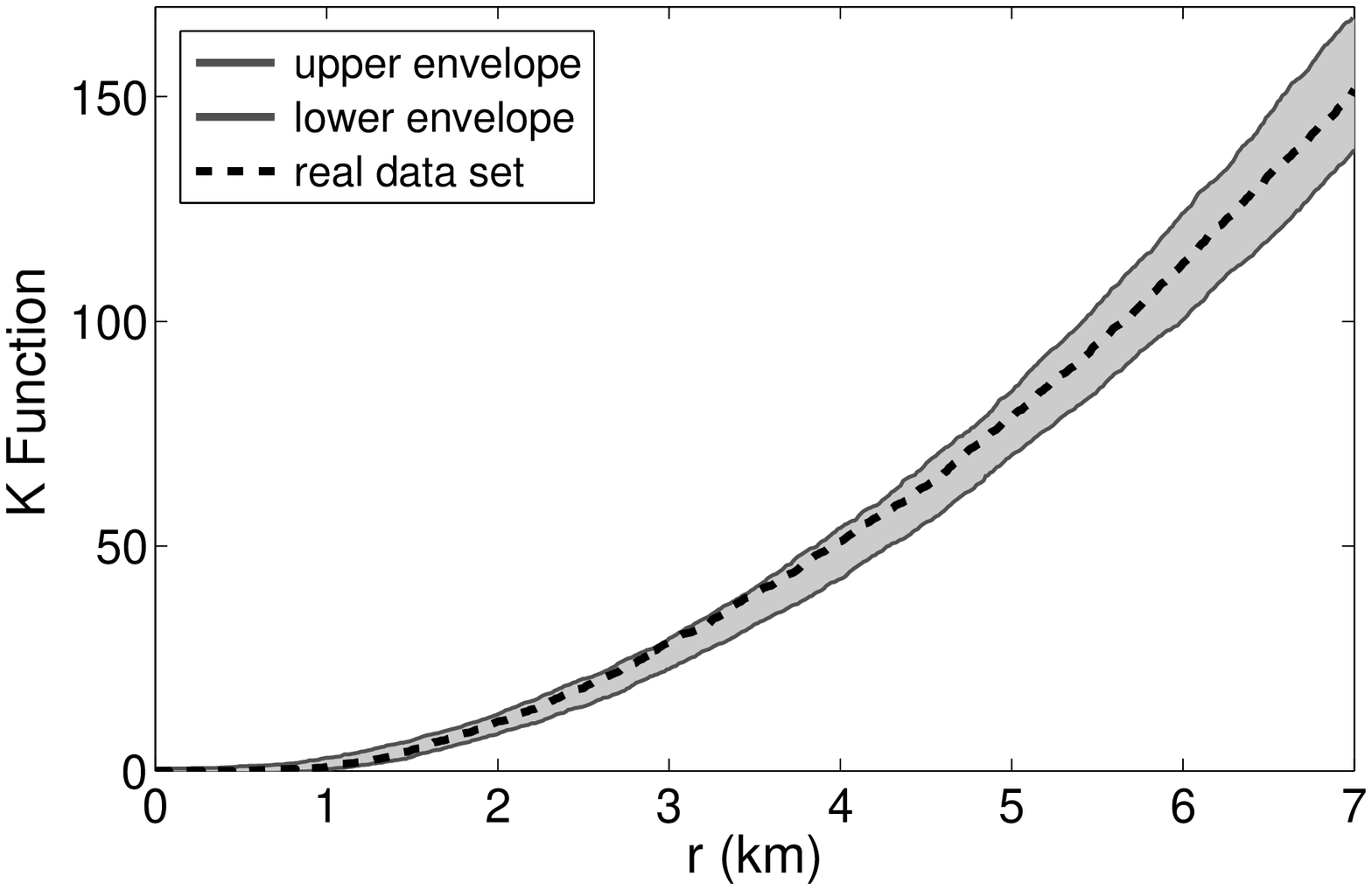}
                \caption{LA data set}
%                \label{coverageGG}
        \end{subfigure}
        \caption{K function of the fitted Gauss DPP.}\label{GaussfitHou}
\end{figure}

%\begin{figure}[h]
%    \centerline{\psfig{figure=intGauss.eps, width=70mm} }
%    \caption{Interference distribution of fitted Gauss DPP}
%     \label{intGauss}
%\end{figure}

%\begin{figure}[h]
%        \centering
%                \includegraphics[height=2in, width=3.2in]{HouGaussDPPCOP2.eps}
%                \includegraphics[height=2in, width=3.2in]{LAGaussDPPCOP2.eps}
%        \caption{Coverage probability of the Gauss DPP fitted to the Houston (left) and LA (right) data set.}\label{coverageGauss}
%\end{figure}

\begin{figure}
        \begin{subfigure}[b]{0.47\textwidth}
                \includegraphics[height=2in, width=3.2in]{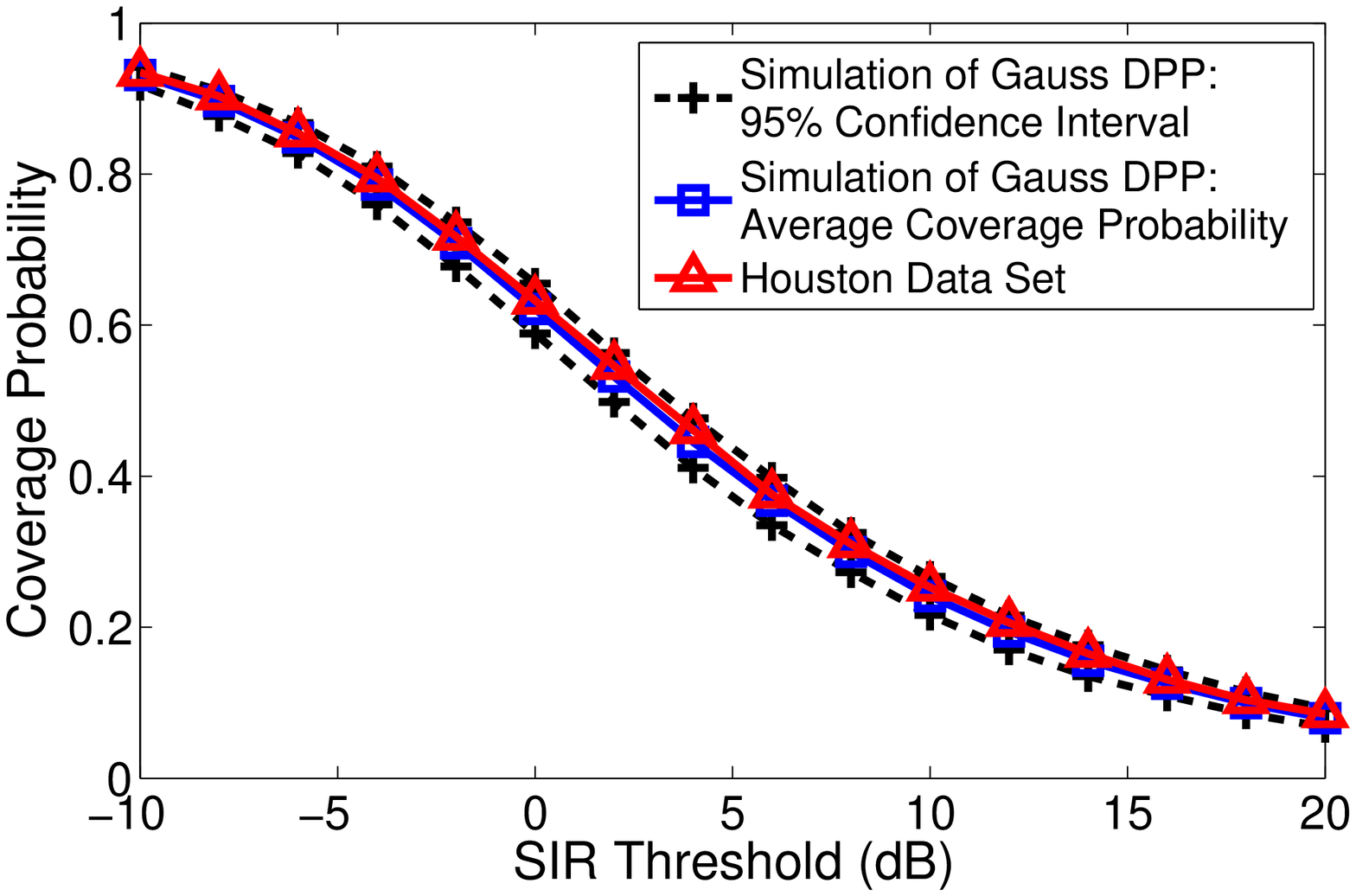}
                \caption{Houston data set}
%                \label{GGfitKest}
        \end{subfigure}
        \hfill
        \begin{subfigure}[b]{0.47\textwidth}
                \includegraphics[height=2in, width=3.2in]{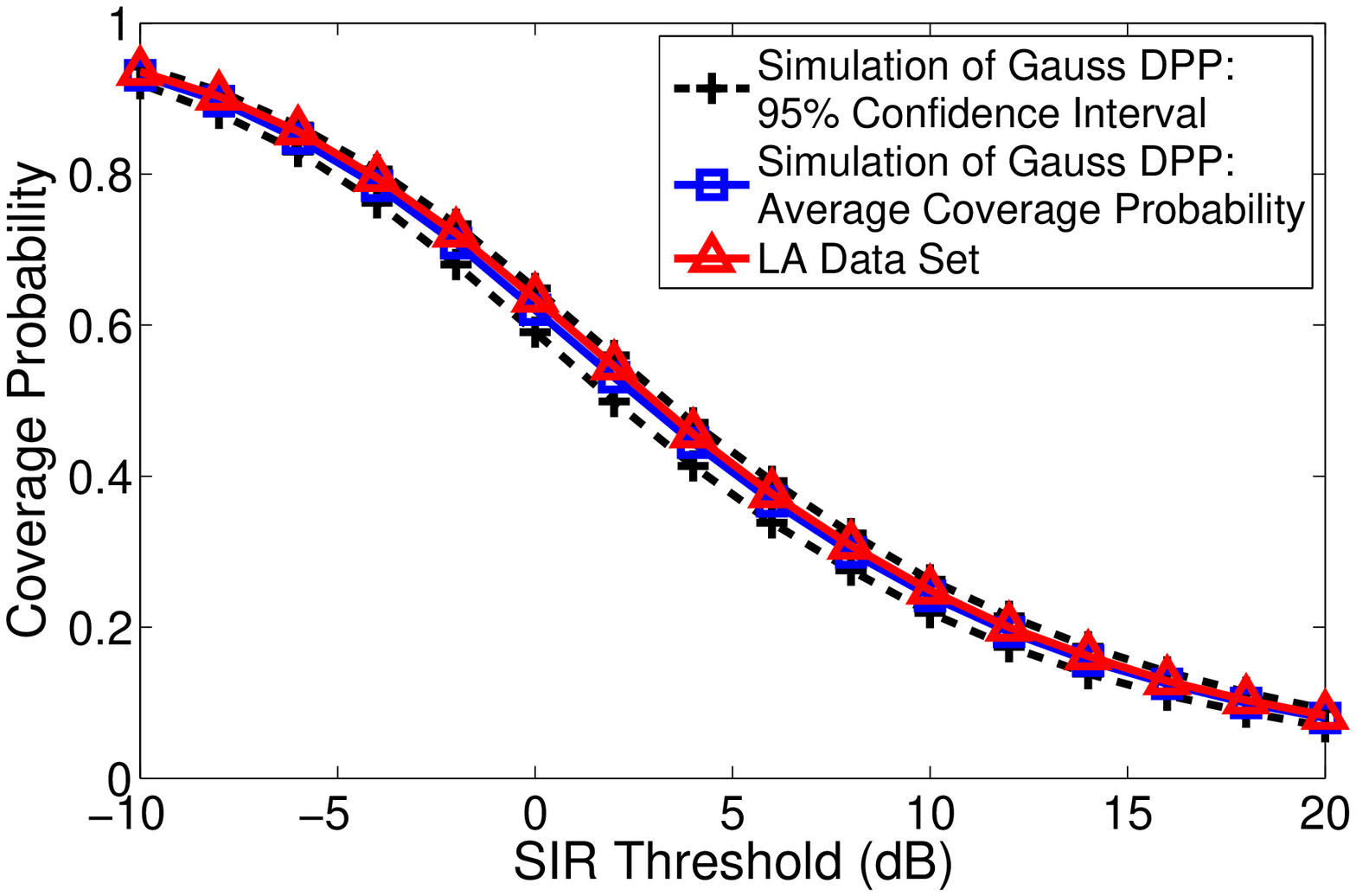}
                \caption{LA data set}
%                \label{coverageGG}
        \end{subfigure}
        \caption{Coverage probability of the fitted Gauss DPP.}\label{coverageGauss}
\end{figure}

To evaluate the goodness-of-fit for these DPP models, we generate 1000 realizations of each DPP model and examine whether the simulated DPPs fit with the behavior of real BS deployments in terms of the summary statistics. Specifically, based on the null hypothesis that real BS deployments can be modeled as realizations of DPPs, we verify whether the K-function of the real data set lies within the envelope of the simulated DPPs. We use similar testing method for the coverage probability; a 95\% confidence interval is used for evaluation. 

\textbf{Goodness-of-fit for Gauss DPP Model:} The testing results for the K function of the fitted Gauss DPP are given in Fig.~\ref{GaussfitHou}, which clearly show that the K functions of the real BS deployments lie within the envelope of the fitted Gauss DPP. The coverage probability for the fitted Gauss DPP is provided in Fig.~\ref{coverageGauss}, from which it can be observed that the coverage probabilities of the Houston and LA data sets lie within the 95\% confidence interval of the simulated Gauss DPPs. In addition, the average coverage probability of the fitted Gauss DPP is slightly lower than that of real data sets, which means that the fitted Gauss DPP corresponds to a slightly smaller repulsiveness than the real deployments. %By comparing Fig.~\ref{intGauss} and Fig.~\ref{coverageGauss1}, a larger variance for interferee nce distribution than SINR distribution can be observed.
% thereby showing SINR distribution is a more precise metric than interference distribution. 

Therefore, in terms of the above summary statistics, the Gauss DPP model can be used as a reasonable point process model for real BS deployments. In addition, due to the concise definition of its kernel, the shot noise analysis of the Gauss DPP is possible, which further motivates the use of Gauss DPPs to model real-world macro BS deployments.

\begin{comment}
\begin{figure}[h]
        \centering
                \includegraphics[height=2in, width=3.2in]{HouKestCauchyDPP2.eps}
                \includegraphics[height=2in, width=3.2in]{LAKestCauchyDPP2.eps}
        \caption{K function for the Cauchy DPP fitted to the Houston (left) and LA (right) data set.}\label{CauchyfitHou}
\end{figure}

\begin{figure}[h]
        \centering
                \includegraphics[height=2in, width=3.2in]{HouCauchyDPPCOP2.eps}
                \includegraphics[height=2in, width=3.2in]{LACauchyDPPCOP2.eps}
        \caption{Coverage probability of the Cauchy DPP fitted to the Houston (left) and LA (right) data set.}\label{coverageCauchy}
\end{figure}
\end{comment}
\begin{figure}
        \begin{subfigure}[b]{0.47\textwidth}
                \includegraphics[height=2in, width=3.2in]{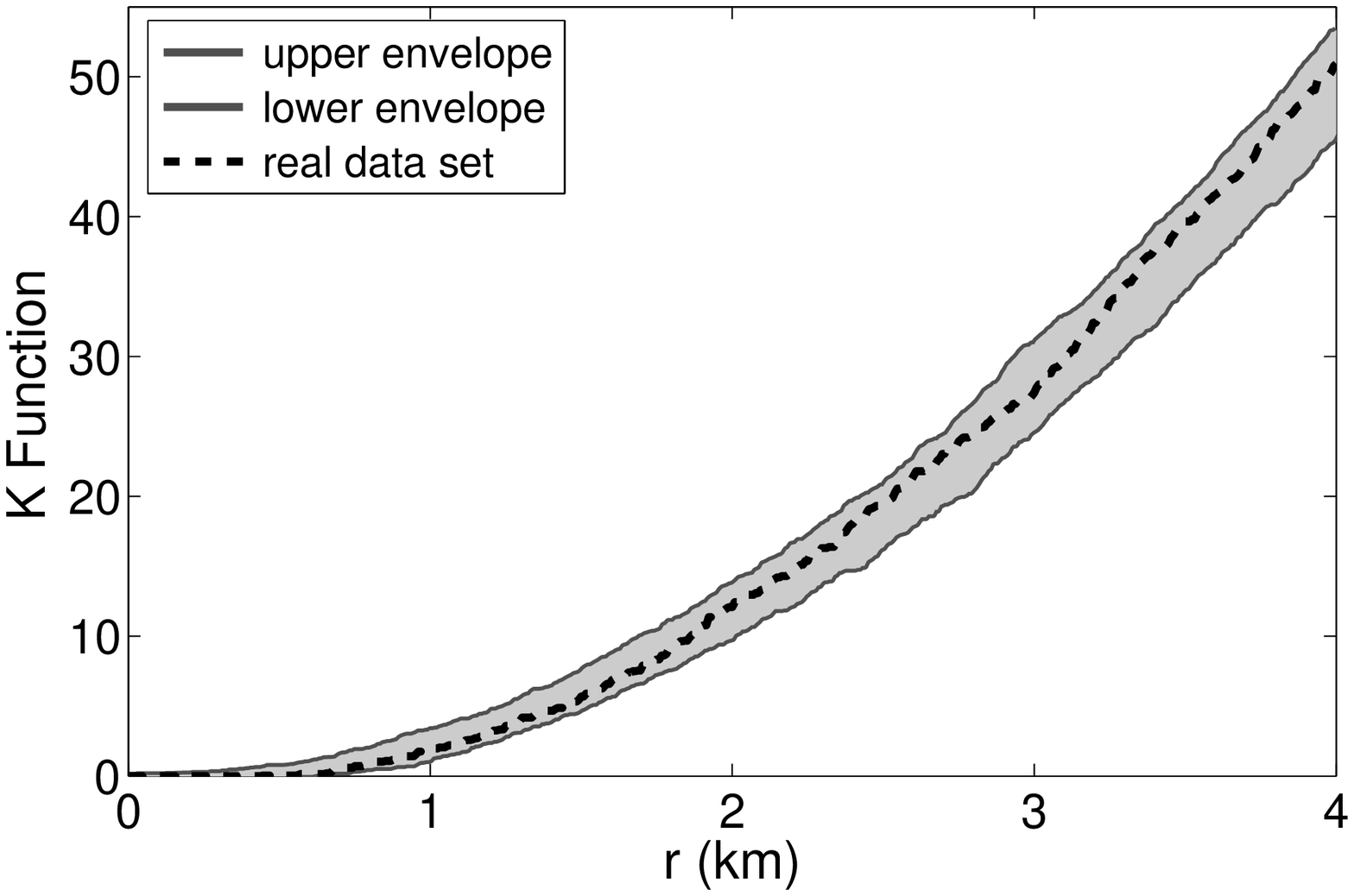}
               \caption{K function}
 %               \label{CauchyfitHou}
        \end{subfigure}
        \hfill
        \begin{subfigure}[b]{0.46\textwidth}
                \includegraphics[height=2in, width=3.2in]{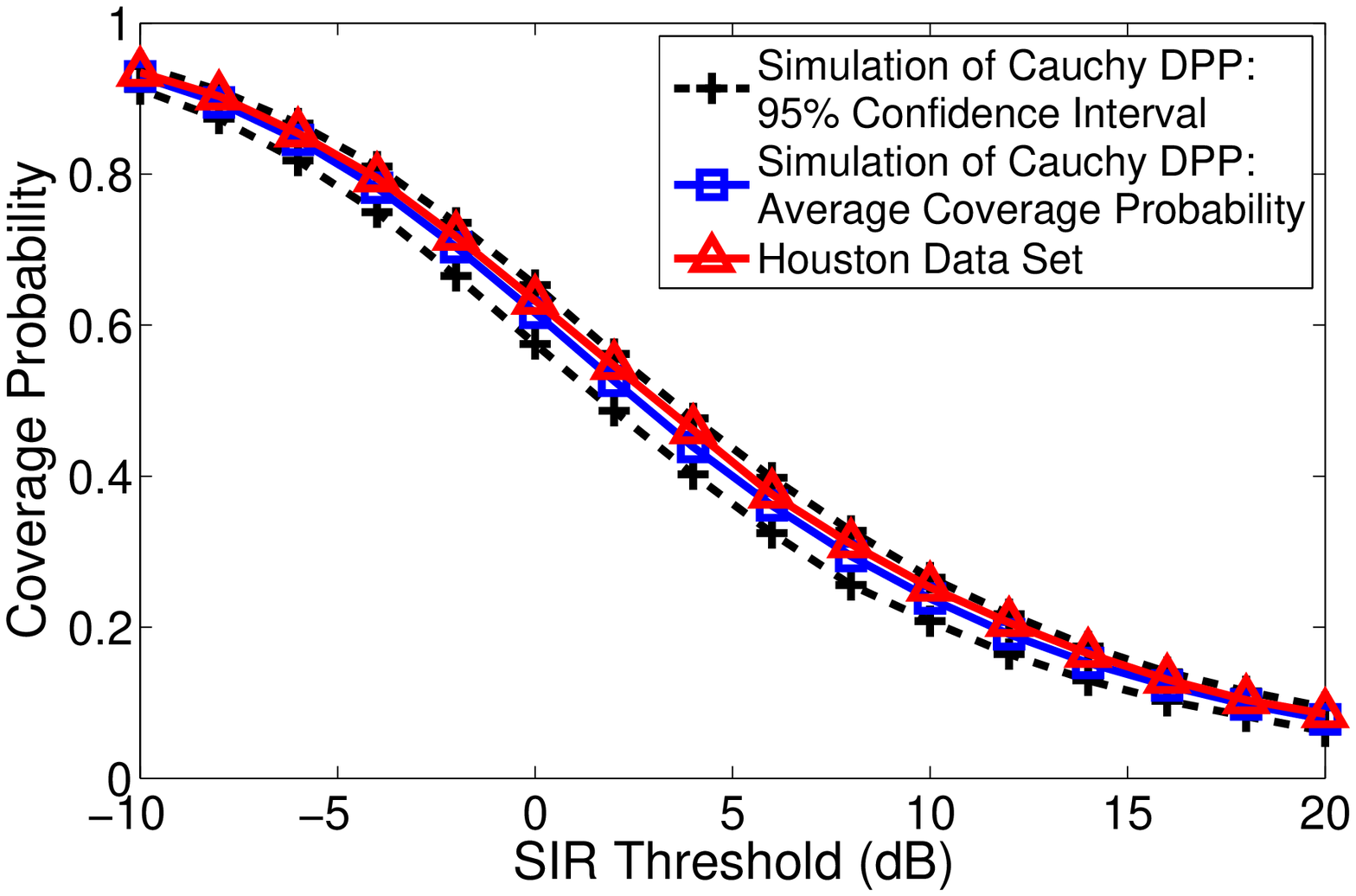}
                \caption{Coverage probability}
 %               \label{coverageCauchy}
        \end{subfigure}
        \caption{Goodness-of-fit for the Cauchy DPP fitted to the Houston data set.}\label{cauchyfit}
\end{figure}

\begin{figure}
        \begin{subfigure}[b]{0.47\textwidth}
                \includegraphics[height=2in, width=3.2in]{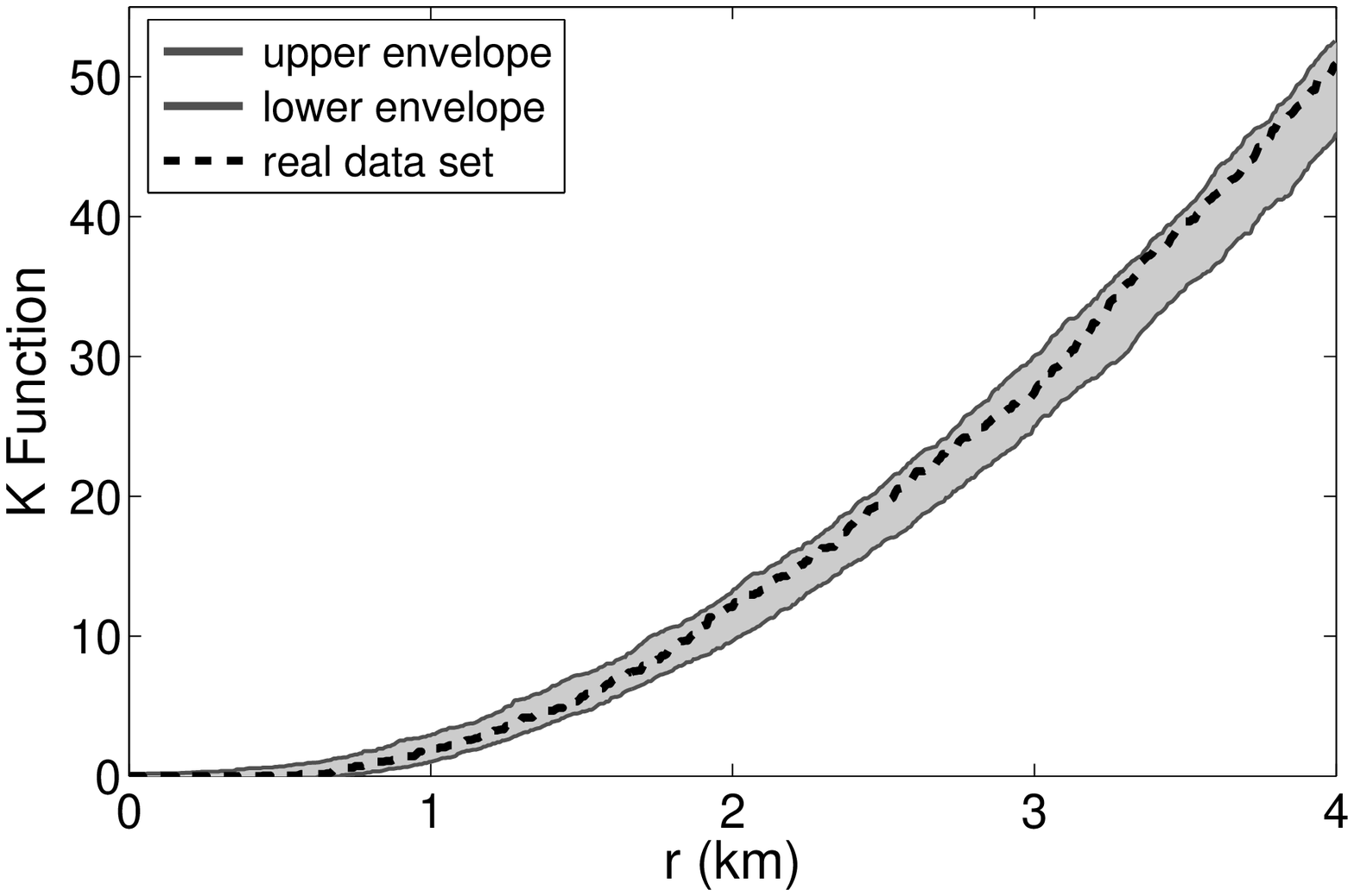}
                \caption{K function}
%                \label{GGfitKest}
        \end{subfigure}
        \hfill
        \begin{subfigure}[b]{0.46\textwidth}
                \includegraphics[height=2in, width=3.2in]{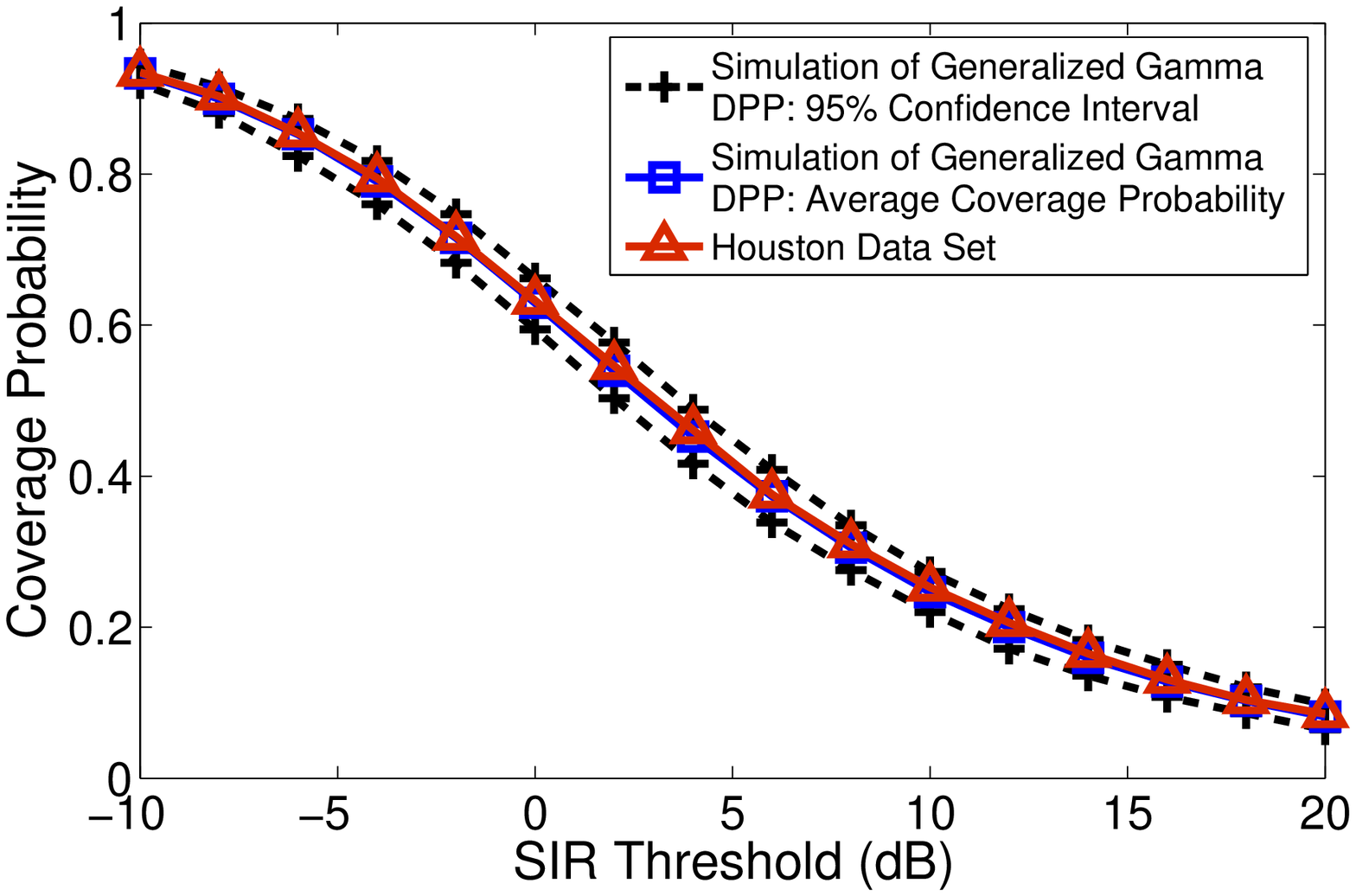}
                \caption{Coverage probability}
%                \label{coverageGG}
        \end{subfigure}
        \caption{Goodness-of-fit for the Generalized Gamma DPP fitted to the Houston data set.}\label{GGfit}
\end{figure}

\textbf{Goodness-of-fit for the Cauchy DPP Model:} Based on the same method as for the Gauss DPP model, we tested the goodness-of-fit for the Cauchy DPP model. The fitting results for the Houston data set are shown in Fig.~\ref{cauchyfit}, from which it can be concluded that the Cauchy DPP model is also a reasonable point process model for real BS deployments. Similar fitting results are also observed for the LA data set, and thus we omit the details.
%The fitting results for the K function and coverage probability are shown in Fig.~\ref{CauchyfitHou} and Fig.~\ref{coverageCauchy}, from which it can be concluded that the Cauchy DPP model is also a reasonable point process model for real BS deployments. 
%However, compared to Gauss DPP model, Fig.~\ref{coverageCauchy} shows Cauchy DPP model corresponds to smaller repulsiveness due to its average coverage probability in Fig.~\ref{coverageCauchy} is lower than that in Fig.~\ref{coverageGauss}.
Compared to the fitted Gauss DPP, the average coverage probability for the fitted Cauchy DPP in Fig.~\ref{cauchyfit} is slightly lower than that in Fig.~\ref{coverageGauss}, which means the fitted Cauchy DPP corresponds to a smaller repulsiveness than the Gauss DPP.

\begin{comment}
\begin{figure}[h]
        \centering
                \includegraphics[height=2in, width=3.2in]{HouKestGGDPP2.eps}
                \includegraphics[height=2in, width=3.2in]{LAKestGGDPP2.eps}
        \caption{K function for the Generalized Gamma DPP fitted to the Houston (left) and LA (right) data set.}\label{GGfitKest}
\end{figure}

\begin{figure}[h]
        \centering
                \includegraphics[height=2in, width=3.2in]{HouGGDPPCOP2.eps}
                \includegraphics[height=2in, width=3.2in]{LAGGDPPCOP2.eps}
        \caption{Coverage probability of the Generalized Gamma DPP fitted to the Houston (left) and LA (right) data set.}\label{coverageGG}
\end{figure}
\end{comment}

\textbf{Goodness-of-fit for the Generalized Gamma DPP Model:} The goodness-of-fit for the Generalized Gamma DPP fitted to the Houston data set is evaluated in Fig.~\ref{GGfit} (the LA data set has similar fitting results). The Generalized Gamma DPP provides the best fit among all these DPP models, especially in terms of coverage probability. In Fig.~\ref{GGfit}, the average coverage probability of the fitted Generalized Gamma DPP almost exactly matches the real BS deployment, while the average coverage probability of the fitted Gauss DPP and the fitted Cauchy DPP all stay below the real data set. This is because the Generalized Gamma DPP corresponds to a higher repulsiveness (which will be proved in Section~\ref{RepComp}), from which a larger coverage probability is expected. 
%The goodness-of-fit for the Generalized Gamma DPP model is evaluated in Fig.~\ref{GGfitKest} and Fig.~\ref{coverageGG}. The Generalized Gamma DPP provides the best fit among all these DPP models, especially in terms of coverage probability. In Fig.~\ref{coverageGG}, the average coverage probability of the Generalized Gamma DPP almost exactly matches the real BS deployment, while the average coverage probability of the Gauss DPP and the Cauchy DPP all stay below the real data set. This is because the Generalized Gamma DPP corresponds to a higher repulsiveness (which will be proved in Section~\ref{RepComp}), from which a larger coverage probability is expected. 

\textbf{Goodness-of-fit for the PPP and the perturbed hexagonal model:} Finally, the goodness-of-fit for the PPP and the perturbed hexagonal grid model are studied. The perturbed hexagonal grid model is obtained by independently perturbing each point of a hexagonal grid in the random direction by a distance $d$~\cite{pairmodel}. This distance is uniformly distributed between 0 and $\eta r$, with $r$ being the radius of the hexagonal cells and $\eta$ is chosen as 0.5 in our simulation. Fig.~\ref{PPPnGrid} depicts the coverage probability of the PPP and of the perturbed hexagonal grid model, which correspond to a lower bound and an upper bound of the actual coverage probability respectively. This is because the PPP exhibits complete spatial randomness while the perturbed grid model maintains good spatial regularity. %Therefore, compared to the DPPs, the PPP and the perturbed hexagonal grid model should be rejected as a reasonable random point process models to study real macro BS deployments. 

%\begin{figure}[h]
%        \centering
%                \includegraphics[height=2in, width=3.2in]{PPPnGridHouston2.eps}
%                \includegraphics[height=2in, width=3.2in]{PPPnGridLA2.eps}
%        \caption{Coverage probability of the PPP and the perturbed grid model fitted to the Houston (left) and LA (right) data set.}\label{PPPnGrid}
%\end{figure}

\begin{figure}
        \begin{subfigure}[b]{0.45\textwidth}
                \includegraphics[height=2in, width=3.2in]{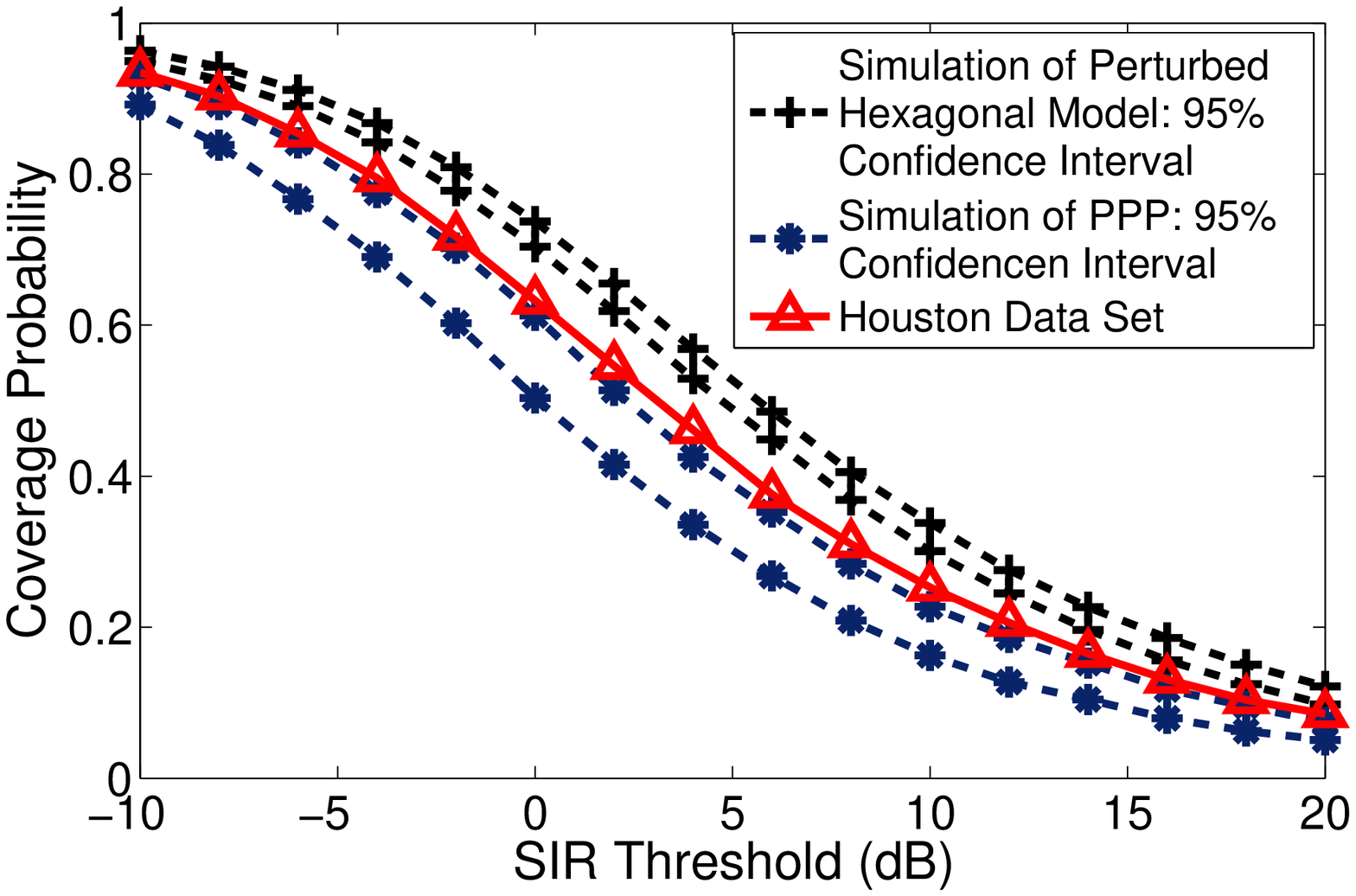}
                \caption{Houton data set}
%                \label{ESFGaussvsPPP}
        \end{subfigure}
        \hfill
        \begin{subfigure}[b]{0.45\textwidth}
                \includegraphics[height=2in, width=3.2in]{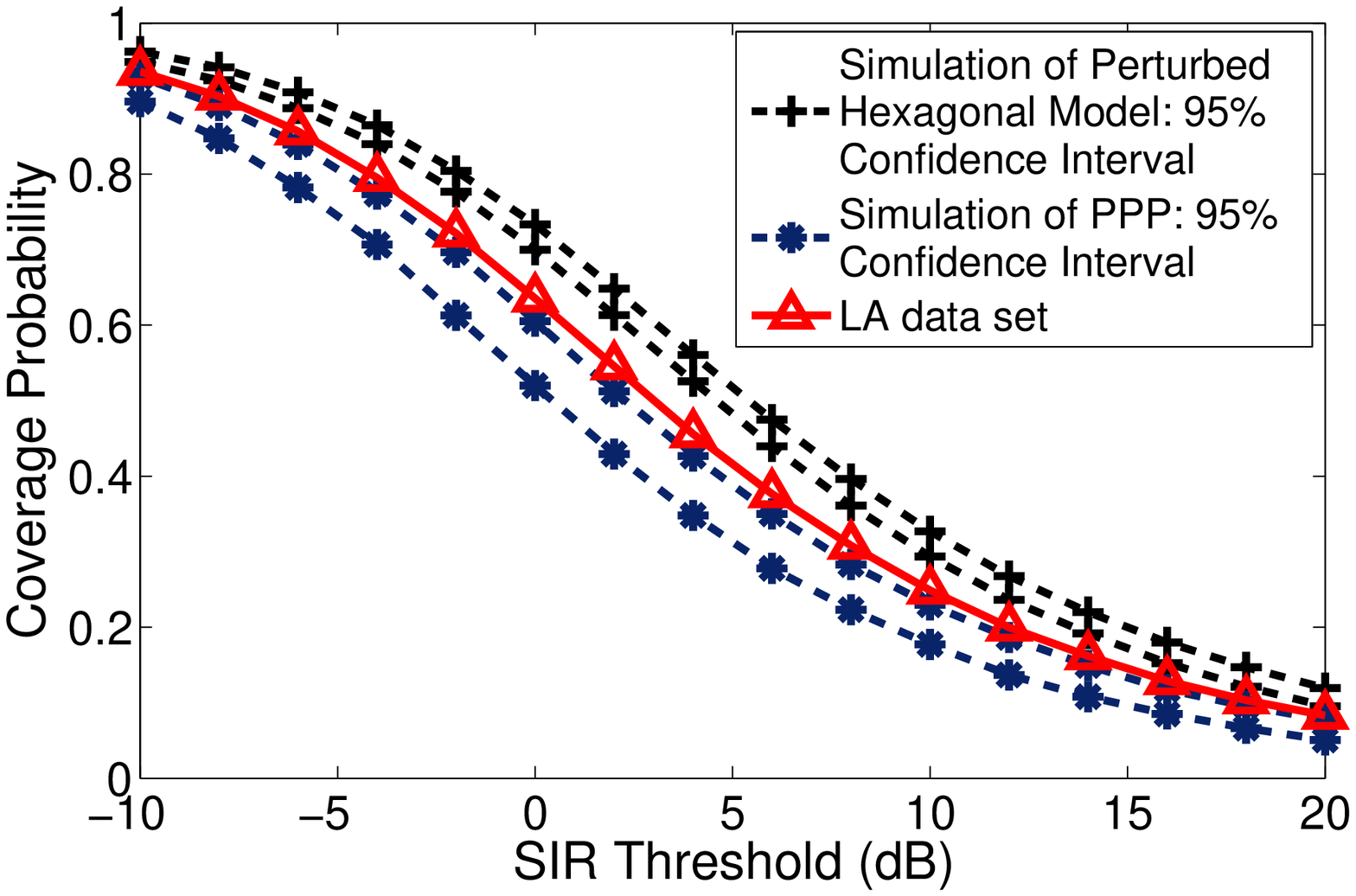}
                \caption{LA data set}
%                \label{NNGaussvsPPP}
        \end{subfigure}
        \caption{Coverage probability of the PPP and the perturbed grid model.}\label{PPPnGrid}
\end{figure}

\subsection{Repulsiveness of Different DPPs}~\label{RepComp}
In order to explain why the Generalized Gamma DPP has larger repulsiveness, we use the metric suggested in~\cite{rubak} to measure the repulsiveness of different DPPs.
Specifically, from Lemma~\ref{PalmDPP}, the intensity measure of a stationary DPP $\Phi$ under its reduced Palm distribution is $\rho_{o}^{(1)}(x)=\rho^{(2)}(0,x)/\rho^{(1)}(x)$, where $\rho^{(2)}$ and $\rho^{(1)}$ are the second and the first order product density of $\Phi$. By calculating the difference of the total expected number of points under the probability distribution $\mathbb{P}$ and the reduced Palm distribution $\mathbb{P}_{o}^{!}$, the repulsiveness of a stationary DPP $\Phi$ with intensity $\lambda$ can be measured using the following metric~\cite{rubak}:
\begin{equation}
	\small
	\mu= \int_{\mathbb{R}^{2}} \left[\lambda - \rho_{o}^{(1)}(x)\right] {\rm d} x =\frac{1}{\lambda} \int_{\mathbb{R}^{2}} | K_{0} (x)|^{2} {\rm d} x =\frac{1}{\lambda} \int_{\mathbb{R}^{2}} | \varphi (x) |^{2} {\rm d} x,
\end{equation}
where $K_0(x)$ and $\varphi(x)$ denote the covariance function and spectral density of $\Phi$ respectively. 

PPP has $\mu = 0$ due to Slivnyak's theorem, while the grid-based model has $\mu=1$ since the point at the origin is excluded under reduced Palm distribution. Generally, larger value of $\mu$ will correspond to a more repulsive point process. This repulsiveness measure for the Gauss, Cauchy and Generalized Gamma model can be calculated as: $\mu_{\text{gauss}}=\lambda \pi \alpha^{2}/2$, $\mu_{\text{cauchy}}=\lambda \pi \alpha^{2}/(2 \nu +1)$, and $\mu_{\text{gengamma}}=\lambda \nu \alpha^{2}/(2^{1+2/\nu} \pi \Gamma(2/\nu))$.
%\begin{equation}\label{repulmetric}
%\small
%\begin{split}
%\mu_{\text{gauss}}&=\frac{\lambda \pi \alpha^{2}}{2},\\
%\mu_{\text{cauchy}}&=\frac{\lambda \pi \alpha^{2}}{2 \nu +1},\\
%\mu_{\text{gengamma}}&=\frac{\lambda \nu \alpha^{2}}{2^{1+2/\nu} \pi \Gamma(2/\nu)}.
%\end{split}
%\end{equation} 
Based on the parameters in Table~\ref{par1}, we can calculate the repulsiveness measure of each DPP model fitted to the Houston data set as $\mu_{\text{gauss}}=0.4999$, $\mu_{\text{cauchy}}=0.4365$ and $\mu_{\text{gengamma}}=0.5905$. Similarly, the repulsiveness measure of each DPP model fitted to the LA data set is given by $\mu_{\text{gauss}}=0.5004$, $\mu_{\text{cauchy}}=0.4351$, $\mu_{\text{gengamma}}=0.5479$. Therefore, it can be concluded that the fitted Generalized Gamma DPP has the largest repulsiveness, followed by the fitted Gauss DPP, while the fitted Cauchy DPP is the least repulsive. Since higher repulsiveness will result in more regularity for the point process, a Generalized Gamma DPP generally corresponds to a larger average coverage probability.

%To summarize the statistical modeling results, we find all these DPP models have reasonable goodness-of-fit to the real BS deployments. In particular, due to its larger repulsiveness, the Generalized Gamma model provides the best fit in terms of coverage probability. However, the Generalized Gamma model is generally less tractable since it is defined based on the spectral density. In contrast, the fitted Cauchy model has the smallest repulsiveness and also less precise results in terms of the summary statistics such as coverage probability. Compared to other DPP models, the Gauss DPP model provides good modeling accuracy to the real BS deployments with better mathematical tractability due to the simple definition of its kernel. %In the next section, we will investigate the mathematical tractability of the cellular networks with DPP configured BSs more thoroughly. 

\section{Performance Comparisons of DPPs and PPPs}\label{SecPerfEva}
Based on the analytical, numerical and statistical results from previous sections, we demonstrate that the DPPs are more accurate than the PPPs to predict key performance metrics in cellular networks for the following reasons.

Firstly, since the DPPs have more regularly spaced point pattern, they will have larger empty space function than the PPPs. Equivalently, this means the distance from the origin to its closest point on the DPPs fitted to real deployments is stochastically less than the PPPs, which can be observed in Fig.~\ref{ESFGaussvsPPP} for the Gauss DPP. Therefore, if each user is associated with its nearest BS, DPPs will lead to a stronger received power at the typical user compared to PPPs in the stochastic dominance sense. 

\begin{comment}
\begin{figure}[h]
        \centering
                \includegraphics[height=2.2in, width=3.2in]{ESFGaussDPPHouPPPvsDPP2.eps}
                \includegraphics[height=2.2in, width=3.2in]{ESFGaussDPPLAPPPvsDPP2.eps}
        \caption{Empty space function of the fitted Gauss DPP and PPP.}\label{ESFGaussvsPPP}
\end{figure}
\end{comment}

%\begin{figure}
%        \begin{subfigure}[b]{0.5\textwidth}
%                \includegraphics[height=2in, width=3.2in]{ESFGaussDPPHouPPPvsDPP2.eps}
%                \caption{Empty space function}
%                \label{ESFGaussvsPPP}
%        \end{subfigure}
%        \begin{subfigure}[b]{0.5\textwidth}
%                \includegraphics[height=2in, width=3.2in]{NNGaussDPPLAPPPvsDPP2.eps}
%                \caption{Nearest neighbor function}
%                \label{NNGaussvsPPP}
%        \end{subfigure}
%        \caption{Comparison of the Gauss DPP and PPP fitted to Houston data set.}\label{GaussvsPPP}
%\end{figure}

\begin{figure}
        \begin{subfigure}[b]{0.47\textwidth}
                \includegraphics[height=2in, width=3.2in]{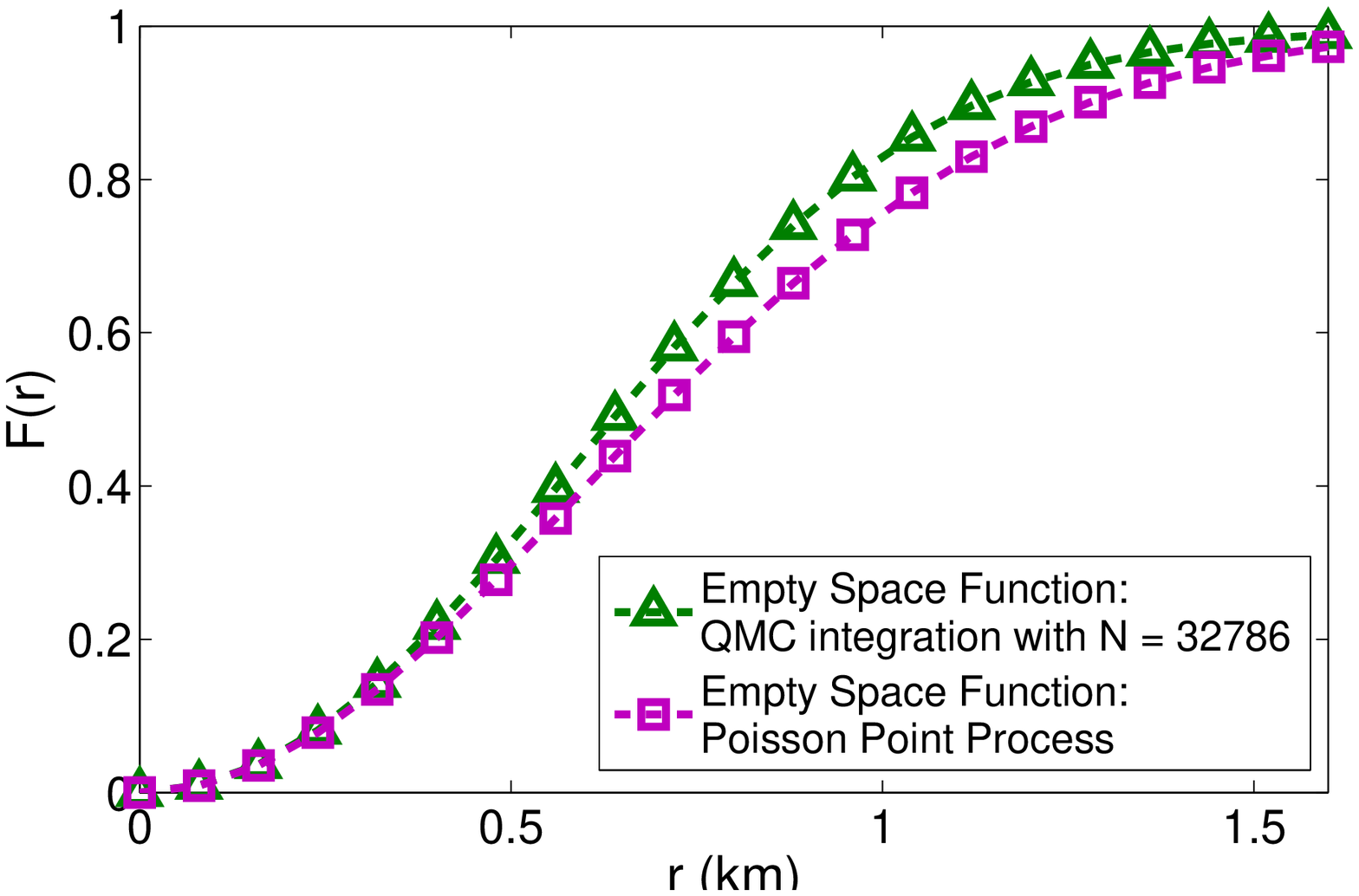}
                \caption{Empty space function}
                \label{ESFGaussvsPPP}
        \end{subfigure}
        \hfill
        \begin{subfigure}[b]{0.47\textwidth}
                \includegraphics[height=2in, width=3.2in]{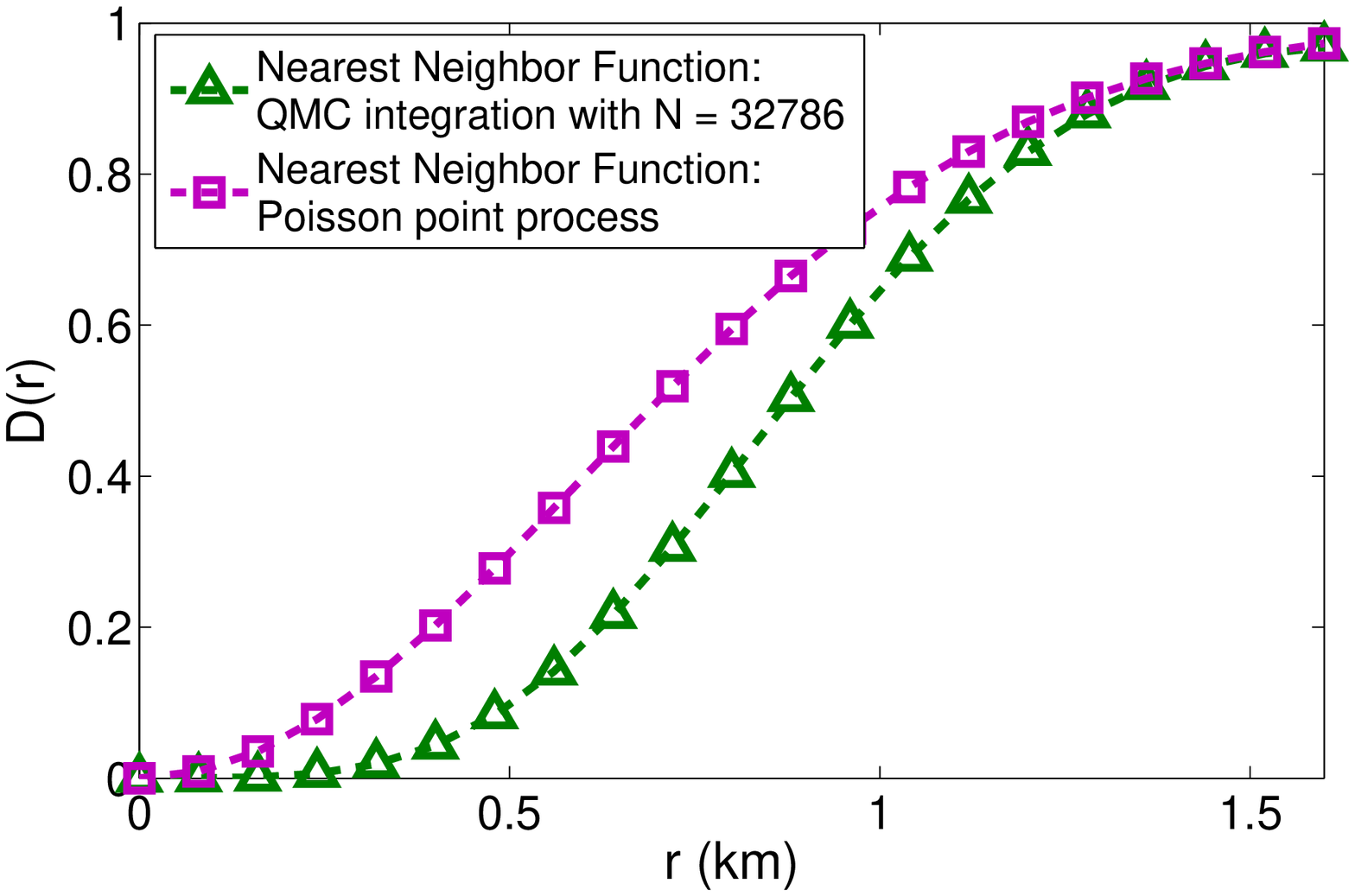}
                \caption{Nearest neighbor function}
                \label{NNGaussvsPPP}
        \end{subfigure}
        \caption{Comparison of the Gauss DPP and PPP fitted to Houston data set.}\label{GaussvsPPP}
\end{figure}

Secondly, the fitted DPPs will have smaller nearest neighbor function than the PPP, which can be observed in Fig.~\ref{NNGaussvsPPP} for the Gauss DPP. In addition, we can also observe from Fig.~\ref{NNGaussvsPPP} that the nearest neighbor function for the Gauss DPP is much smaller than the PPP when $r$ is small. This indicates that the PPP will largely overestimate the nearest neighbor function when $r$ is small, which leads to much closer strong interfering BSs compared to the Gauss DPP. 

\begin{comment}
\begin{figure}[h]
        \centering
                \includegraphics[height=2.2in, width=3.2in]{NNGaussDPPHouPPPvsDPP3.eps}
                \includegraphics[height=2.2in, width=3.2in]{NNGaussDPPLAPPPvsDPP2.eps}
        \caption{Nearest neighbor function of the fitted Gauss DPP and PPP.}\label{NNGaussvsPPP}
\end{figure}
\end{comment}

%\begin{remark}
In addition to the empty space function and the nearest neighbor function, the DPPs are also more accurate in estimating the interference and coverage probability than the PPP. When each user is associated with an arbitrary but fixed BS, an immediate implication of Lemma~\ref{MeanInt} is that the mean interference for a stationary DPP $\Phi$ with intensity $\lambda$ is strictly smaller than that of the PPP with the same intensity. This can be observed by separating~(\ref{meanInt1}) as:
\begin{equation}
\mathbb{E}[I |  x_{0} = (r_0,0)]  = P\lambda \int_{\mathbb{R}^2} \mathit{l} (x) {\rm d}x - \frac{P}{\lambda}\int_{\mathbb{R}^2} |K(x,x_0)|^2 \mathit{l}(x) {\rm d}x,
\end{equation}
where the first term is equal to the mean interference under the PPP distributed BSs by Slivnyak's theorem, while the second term stems from the soft repulsion among BSs in the DPP $\Phi$. 
%\end{remark}

Finally, under the nearest BS association scheme, the coverage probability estimated from the fitted DPPs is validated to be close to the BS deployments in Section~\ref{SubSecStat}. In contrast, the PPP only provides a lower bound to the actual coverage probability.

\section{Conclusion}~\label{SECSumm}
In this paper, the analytical tractability and the modeling accuracy of determinantal point processes for modeling cellular network BS locations are investigated. First, cellular networks with DPP configured BSs are proved to be analytically tractable. Specifically, we have summarized the fact that DPPs have closed form expressions for the product density and reduced Palm distribution, then we have derived the Laplace functional of the DPPs and of independently marked DPPs for functions satisfying certain mild conditions. Based on these computational properties, the empty space function, the nearest neighbor function, and the mean interference were derived analytically and evaluated using the Quasi-Monte Carlo integration method. 
%They are all more accurate than for the PPP. 
In addition, the Laplace transform of the interference and the SIR distribution under the nearest BS association scheme are also derived and numerically evaluated. 

Next, using the K function and the coverage probability, DPPs are shown to be accurate by fitting three stationary DPP models to two real macro BS deployments: the Gauss DPP, Cauchy DPP and Generalized Gamma DPP. 
%By using K-function and coverage probability as the performance metrics, hypothesis testing procedures are adopted to evaluate the goodness-of-fit for these DPP models. 
In particular, the Generalized Gamma DPP is found to provide the best fit in terms of coverage probability due to its higher repulsiveness. However, the Generalized Gamma DPP is generally less tractable since it is defined based on its spectral density. The Gauss DPP also provides a reasonable fit to real BS deployments, but with higher mathematical tractability, due to the simple definition of its kernel. Compared to other DPP models, the fitted Cauchy DPP has the smallest repulsiveness and also less precise results in terms of the summary statistics. Therefore, we conclude that the Gauss DPP provides the best tradeoff between accuracy and tractability.
%and is a more accurate model than PPP or pertubted hexagonal model. %The derived empty space function and nearest neighbor function can accurately capture the repulsiveness among BSs in real cellular networks, which is shown to be more accurate   compared to the PPP. The mean interference under two BS association schemes are derived, which are also more accurate than the PPP since PPP overestimates the interference at the typical user. Finally, Laplace transform of the mean interference and the coverage probability under nearest BS association schemes have been derived. 

Finally, based on a combination of analytical, numerical and statistical results, we demonstrate that DPPs outperform PPPs to model cellular networks in terms of several key performance metrics. 

Future work may include finding different DPP examples that lead to more efficient evaluations of the key performance metrics (i.e., without relying on Quasi-Monte Carlo integration), or extending the SISO single-tier network model analyzed here to MIMO or HetNet models. 

\section*{Acknowledgments}
The work of F. Baccelli and Y. Li was supported by a grant of the Simons Foundation (\#197982 to UT Austin). The authors thank Dr. Ege Rubak (Aalborg University) for sharing the DPP simulation package, and Dr. Paul Keeler (ENS and INRIA, Paris) for his suggestions on the Quasi-Monte Carlo integration method.
\appendices
\section{Proof of Lemma~\ref{LFDPP2}}~\label{LFDPP2Appdx}
For any function $f$ satisfying the conditions in Lemma~\ref{LFDPP2}, define the following function for $k \in \mathbb{N}$: 
\allowdisplaybreaks
\begin{align}
f_k(x)=
\left\{
\begin{array}{ll}
f(x), & \text{if } x \in B(0,k), \\
0, & \text{otherwise}.
\end{array}
\emph{ } \right.
\end{align}
Based on Lemma~\ref{LFDPP}, since each $f_k(x)$ has finite support, we have:
$\mathbb{E}\left[\exp\left(-\int_{\mathbb{R}^2} f_k(x) \Phi({\rm d} x)\right)\right] = \sum_{n=0}^{+\infty}  \frac{(-1)^n}{n!}  \int_{(\mathbb{R}^2)^n}\det \left(K(x_i,x_j)\right)_{1 \leq i,j \leq n} \prod_{i=1}^{n} \left(1-\exp(-f_k(x_i))\right) {\rm d}x_1 ... {\rm d}x_n.$
%\begin{small}
%\begin{align*}
%\mathbb{E}\left[\exp\left(-\int_{\mathbb{R}^2} f_k(x) \Phi({\rm d} x)\right)\right] = \sum_{n=0}^{+\infty}  \frac{(-1)^n}{n!}  \int_{(\mathbb{R}^2)^n}\det \left(K(x_i,x_j)\right)_{1 \leq i,j \leq n} \prod_{i=1}^{n} \left(1-\exp(-f_k(x_i))\right) {\rm d}x_1 ... {\rm d}x_n.
%\end{align*}
%\end{small}

From the monotone convergence theorem, we have:\\
\begin{comment}
1. $\{f_k\}$ converges to $f$ uniformly on $\mathbb{R}^2$. 

Since $\lim\limits_{|x| \rightarrow \infty} f(x) = 0$, so for $\forall \epsilon >0$ , $\exists R >0$, such that $f(x) < \epsilon$ for $\forall |x| > R$. Then for any $k > \lceil R \rceil$, we have:
\begin{align}
f(x) - f_k(x)
\left\{
\begin{array}{ll}
= 0, & \text{for } x \in B(0,k), \\
< \epsilon, & \text{for} x \in B^{c}(0,k).
\end{array}
\emph{ } \right.
\end{align}
Since $f(x) - f_k(x) \geq 0 $ for $\forall x$, thus $|f(x) - f_k(x)| \leq \epsilon $ for any $k > \lceil R \rceil$, which completes the proof.
\end{comment}
1. $\lim\limits_{k \rightarrow \infty} \mathbb{E}\left[ \exp\left(-\int_{\mathbb{R}^2} f_k(x) \Phi({\rm d} x) \right)\right] = \mathbb{E} \left[ \exp\left(-\int_{\mathbb{R}^2} f(x) \Phi({\rm d} x)\right)\right]$.

Let us now show that:\\
2. $\lim\limits_{k \rightarrow \infty} \sum_{n=0}^{+\infty}  \frac{(-1)^n}{n!}  \int_{(\mathbb{R}^2)^n}\det \left(K(x_i,x_j)\right)_{1 \leq i,j \leq n} \prod_{i=1}^{n} \left(1-\exp(-f_k(x_i))\right) {\rm d}x_1 ... {\rm d}x_n = \sum_{n=0}^{+\infty}  \\ \frac{(-1)^n}{n!}  \int_{(\mathbb{R}^2)^n}\det \left(K(x_i,x_j)\right)_{1 \leq i,j \leq n} \prod_{i=1}^{n} \left(1-\exp(-f(x_i))\right) {\rm d}x_1 ... {\rm d}x_n$.

To prove this result, we use the following lemma~\cite[Theorem 7.11]{rudin1976principles}:
\begin{lemma}~\label{limsumlemma}
Suppose $f_n \rightarrow f$ uniformly on a set $E$ in a metric space. Let $x$ be a limit point on $E$ such that $\lim\limits_{t \rightarrow x} f_n(t)$ exists for $\forall n \in \mathbb{N}$, then 
$\lim\limits_{t \rightarrow x}\lim\limits_{n \rightarrow \infty} f_n(t) = \lim\limits_{n \rightarrow \infty} \lim\limits_{t \rightarrow x} f_n(t)$.
\end{lemma}

Let $h_n(k) = \sum_{m=0}^{n} \int_{(\mathbb{R}^2)^m} \frac{(-1)^m}{m!}\det \left(K(x_i,x_j)\right)_{1 \leq i,j \leq m} \prod_{i=1}^{m} \left(1-\exp(-f_k(x_i))\right) {\rm d}x_1 ... {\rm d}x_m$. We prove that $\{h_n\}$ converges uniformly $\forall k \in \mathbb{N}$. This is because:
\begin{align*}
\allowdisplaybreaks
&\biggl|\int_{(\mathbb{R}^2)^m} \frac{(-1)^m}{m!} \det \left(K(x_i,x_j)\right)_{1 \leq i,j \leq m} \prod_{i=1}^{m} \left(1-\exp(-f_k(x_i))\right) {\rm d}x_1 ... {\rm d}x_m \biggl|\\
%\leq & \frac{1}{m!}  \int_{(\mathbb{R}^2)^m}\det \left(K(x_i,x_j)\right)_{1 \leq i,j \leq m} \prod_{i=1}^{m} %\left(1-\exp(-f(x_i))\right) {\rm d}x_1 ... {\rm d}x_m\\
\overset{(a)}{\leq} & \frac{1}{m!}  \left(\int_{\mathbb{R}^2}K(x,x) (1-\exp(-f(x)) ){\rm d}x \right)^m \triangleq M_m,
%& \overset{n \rightarrow \infty}{\rightarrow } \exp\left(\int_{\mathbb{R}^2}K(x,x) (1-\exp(-f(x)) ){\rm d}x\right),
\end{align*}
where (a) follows from Hadamard's inequality, i.e., $\det (\left(K(x_{i},x_{j})\right)_{1\leq i,j\leq n}\leq  \prod_{i=1}^{n} K(x_{i},x_{i})$ if $K$ is positive semi-definite. 
Since $\int_{\mathbb{R}^2}K(x,x) (1-\exp(-f(x)) ){\rm d}x$ is finite by assumption, $\sum_{m=0}^{\infty} M_m $
%= \exp(\int_{\mathbb{R}^2}K(x,x) (1-\exp(-f(x)) ){\rm d}x) $ 
is also finite. Therefore, by Weierstrass M-test~\cite[Theorem 7.10]{rudin1976principles}, $\{h_n\}$ converges uniformly.

Next, we show $\lim\limits_{k \rightarrow \infty} h_n (k)$ exists for $\forall n \in \mathbb{N}$. This is because for $ 0 \leq m \leq n$, we have:
\begin{align}\label{limitexistEq}
\allowdisplaybreaks
&\lim\limits_{k \rightarrow \infty}\int_{(\mathbb{R}^2)^m}  \frac{(-1)^m}{m!}   \det \left(K(x_i,x_j)\right)_{1 \leq i,j \leq m} \prod_{i=1}^{m} \left(1-\exp(-f_k(x_i))\right) {\rm d}x_1 ... {\rm d}x_m \nonumber\\
\overset{(a)}{=}& \int_{(\mathbb{R}^2)^m} \frac{(-1)^m}{m!}   \det \left(K(x_i,x_j)\right)_{1 \leq i,j \leq m} \lim\limits_{k \rightarrow \infty} \prod_{i=1}^{m} \left(1-\exp(-f_k(x_i))\right) {\rm d}x_1 ... {\rm d}x_m \nonumber\\
=& \int_{(\mathbb{R}^2)^m} \frac{(-1)^m}{m!}   \det \left(K(x_i,x_j)\right)_{1 \leq i,j \leq m} \prod_{i=1}^{m} \left(1-\exp(-f(x_i))\right) {\rm d}x_1 ... {\rm d}x_m.
\end{align}
Step (a) follows from the dominated convergence theorem (DCT): given $m$, denote $\textbf{x} \triangleq (x_1,...,x_m)$ and $g_k(\textbf{x}) \triangleq \frac{(-1)^m}{m!} \det \left(K(x_i,x_j)\right)_{1 \leq i,j \leq m} \prod_{i=1}^{m} (1-\exp(-f_k(x_i)))$; then from the definition of $f_k(x)$, $g_k(\textbf{x})$ converges pointwise to $\frac{(-1)^m}{m!} \det \left(K(x_i,x_j)\right)_{1 \leq i,j \leq m} \prod_{i=1}^{m} (1-\exp(-f(x_i)))$. In addition, observe that $|g_k(\textbf{x})| \leq \frac{1}{m!}\prod_{i=1}^{m} K(x_i,x_i)(1-\exp(-f(x_i)))$, we have $\int_{(\mathbb{R}^2)^m} \frac{1}{m!}\prod_{i=1}^{m} K(x_i,x_i)\\(1-\exp(-f(x_i))){\rm d}x_1 ... {\rm d}x_m = \frac{\left(\int_{\mathbb{R}^2}K(x,x)(1-\exp(-f(x))){\rm d}x\right)^m}{m!} < \infty$. Since each term of $h_n(k)$ has a finite limit when $k \rightarrow \infty$, thus $\lim\limits_{k \rightarrow \infty} h_n (k)$ also exists.

Now we can apply Lemma~\ref{limsumlemma} to $h_n(k)$ to derive the desired fact:
\begin{align}
\allowdisplaybreaks
&\lim\limits_{k \rightarrow \infty} \sum_{m=0}^{\infty}\frac{(-1)^m}{m!} \int_{(\mathbb{R}^2)^m} \det \left(K(x_i,x_j)\right)_{1 \leq i,j \leq m} \prod_{i=1}^{m} \left(1-\exp(-f_k(x_i))\right) {\rm d}x_1 ... {\rm d}x_m \nonumber\\
=&\lim\limits_{k \rightarrow \infty}  \lim\limits_{n \rightarrow \infty}\sum_{m=0}^{n}\frac{(-1)^m}{m!} \int_{(\mathbb{R}^2)^m} \det \left(K(x_i,x_j)\right)_{1 \leq i,j \leq m} \prod_{i=1}^{m} \left(1-\exp(-f_k(x_i))\right) {\rm d}x_1 ... {\rm d}x_m \nonumber\\
\overset{(a)}{=} &\lim\limits_{n \rightarrow \infty}  \lim\limits_{k \rightarrow \infty}\sum_{m=0}^{n}\frac{(-1)^m}{m!} \int_{(\mathbb{R}^2)^m} \det \left(K(x_i,x_j)\right)_{1 \leq i,j \leq m} \prod_{i=1}^{m} \left(1-\exp(-f_k(x_i))\right) {\rm d}x_1 ... {\rm d}x_m \nonumber\\
%\overset{(b)}{=} &\lim\limits_{n \rightarrow \infty}\sum_{m=0}^{n}\frac{(-1)^m}{m!} \int_{(\mathbb{R}^2)^m} \det \left(K(x_i,x_j)\right)_{1 \leq i,j \leq m} \prod_{i=1}^{m} \left(1-\exp(-f(x_i))\right) {\rm d}x_1 ... {\rm d}x_m \nonumber\\
\overset{(b)}{=} &\sum_{m=0}^{\infty}\frac{(-1)^m}{m!} \int_{(\mathbb{R}^2)^m} \det \left(K(x_i,x_j)\right)_{1 \leq i,j \leq m} \prod_{i=1}^{m} \left(1-\exp(-f(x_i))\right) {\rm d}x_1 ... {\rm d}x_m,
\end{align}
where (a) is derived using Lemma~\ref{limsumlemma}, and (b) follows from~(\ref{limitexistEq}).

The proof of the lemma follows from these two facts.

\begin{comment}
\section{Proof of Corollary \ref{pgfl}}~\label{PGFLCoroApp}
Based on Lemma~\ref{LFDPP2}, we have:
\allowdisplaybreaks
\begin{align*}
G[v] &= \mathbb{E}\left(\exp\left(\log\prod_{x \in \Phi} v(x)\right)\right)\\
&= \mathbb{E}\left(\exp\left(- \sum_{x \in \Phi} \log \frac{1}{v(x)}\right)\right)\\
&= \sum_{n=0}^{+\infty}  \frac{(-1)^n}{n!}  \int_{(\mathbb{R}^2)^n}\det \left(K(x_i,x_j)\right)_{1 \leq i,j \leq n} \prod_{i=1}^{n} \left(1-\exp(-\log \frac{1}{v(x_i)})\right) {\rm d}x_1 ... {\rm d}x_n\\
&= \sum_{n=0}^{+\infty}  \frac{(-1)^n}{n!}  \int_{(\mathbb{R}^2)^n}\det \left(K(x_i,x_j)\right)_{1 \leq i,j \leq n} \prod_{i=1}^{n} \left(1-v(x_i)\right) {\rm d}x_1 ... {\rm d}x_n,
\end{align*}
\end{comment}

\section{Proof of Lemma \ref{LFMarkDPPCoro}}~\label{LFMarkCoroApp}
This can be proved by the following procedure:
\allowdisplaybreaks
\begin{align*}
\allowdisplaybreaks
&\mathbb{E}\left[\exp(-\sum_{i} f(x_i,p_i) )\right] \\
%= &\mathbb{E}\left[\mathbb{E} \left[\exp(-\sum_{i} f(x_i,p_i)) | \Phi \right]\right] \\
%=& \mathbb{E}\left[\mathbb{E}\left[\prod_{i} \exp(-f(x_i,p_i)) | \Phi \right]\right]\\
\overset{(a)}{=}& \mathbb{E} \left[\prod_{i} \int_{\mathbb{R}^+} \exp(-f(x_i,p))F({\rm d} p) \right] \\
\overset{(b)}{=}& \sum_{n = 0}^{+\infty} \frac{(-1)^n}{n!} \int_{(\mathbb{R}^2)^n} \det \left(K(x_i,x_j)\right)_{1 \leq i,j \leq n} \prod_{i=1}^{n} \left(1-\int_{\mathbb{R}^{+}}\exp(-f(x_i,p_i)) F({\rm d} p_i) \right) {\rm d}x_1 ... {\rm d}x_n,
\end{align*}
where (a) is because all the marks are i.i.d. and independent of DPP $\Phi$, while (b) comes from Corollary~\ref{pgfl}.

\section{Proof of Corollary~\ref{ESFPDFCoro}}\label{ESPPDFCoroProof}
We start the proof with the following two lemmas:
\begin{lemma}\label{SwapIntDiff}
Consider two non-negative functions $g(u,v): \mathbb{R} \times \mathbb{R}^d \rightarrow [0,\infty)$, and $p(u): \mathbb{R} \rightarrow [0,+\infty)$, which satisfy the following conditions: (1) $g(u,v)$ is non-decreasing, right continuous w.r.t. $u$, and $g(u,v) = 0$ for $\forall u \leq 0$; (2) $p(u)$ is bounded, right continuous, and $\lim\limits_{u \rightarrow +\infty} p(u) = 0$; (3) $p(u)$ and $g(u,v)$ do not have common discontinuities for Lebesgue almost all $v$. Let $F(u) = \int_{\mathbb{R}^d} g(u,v) {\rm d}v$, we also assume that $F(u)$ is continuous, non-decreasing and bounded on $\mathbb{R}$. Then the following equation holds:
\begin{align}
\int_{\mathbb{R}} p(u) {\rm d}F(u) = \int_{\mathbb{R}^d \times \mathbb{R}} p(u){\rm d}_{u} g(u,v) {\rm d}v,
\end{align}
where the integrals w.r.t. ${\rm d}F(u) $ and ${\rm d}_{u} g(u,v)$ are in the Stieltjes sense.
\end{lemma}

\begin{proof}
Using Stieltjes integration by parts, we have the following:
\begin{align}
\allowdisplaybreaks
\int_{\mathbb{R}} p(u) {\rm d}F(u) &= \int_{\mathbb{R}} p(u) {\rm d}_u \int_{\mathbb{R}^d} g(u,v) {\rm d}v \nonumber\\
\overset{(a)}{=}& -\int_{\mathbb{R}}\int_{\mathbb{R}^d} g(u,v) {\rm d}v {\rm d}p(u) \nonumber\\
\overset{(b)}{=}& -\int_{\mathbb{R}^d}\int_{\mathbb{R}} g(u,v) {\rm d}p(u) {\rm d}v \nonumber\\
\overset{(c)}{=}& \int_{\mathbb{R}^d}\int_{\mathbb{R}} p(u) {\rm d}_u g(u,v) {\rm d}v,
\end{align}
where (a) and (c) are derived using integration by parts for the Stieltjes integrals, and (b) follows from Fubini's theorem. 
\end{proof}

\begin{lemma}[Rubin~\cite{rudin1976principles}]~\label{RubinLemma}
Suppose $\{f_n\}$ is a sequence of differentiable functions on $[a,b]$ such that $\{f_n(x_0)\}$ converges for some point $x_0$ on $[a,b]$. If $\{f_{n}^{'}\}$ converges uniformly on $[a,b]$ to $f'$, then $\{f_n\}$ converges uniformly on $[a,b]$ to a function $f$, and $f^{'}(x) = \lim\limits_{n\rightarrow \infty} f_n^{'}(x)$ for $a \leq x \leq b$.

\begin{comment}
\begin{equation*}
f^{'}(x) = \lim\limits_{n\rightarrow \infty} f_n^{'}(x) \qquad (a \leq x \leq b).
\end{equation*}
\end{comment}

\end{lemma}

We can express the empty space function as $F(r) = \lim\limits_{n \rightarrow \infty} F_n(r)$, where:
\begin{equation*}
F_n(r) = \sum\limits_{k=1}^{n} \frac{(-1)^{k-1}}{k!} \int_{(B(0,r))^k}\det (K(x_i,x_j))_{1 \leq i,j \leq k} {\rm d}x_1...{\rm d}x_k.
\end{equation*}
From Lemma~\ref{ESF}, we know $F_n(r)$ converges pointwise to $F(r)$ for any $r \geq 0$. Let $u(\cdot)$ denote the unit step function and $\delta(\cdot)$ denote the Dirac measure. Note that $F_n(r)$ is equal to 0 for $r \leq 0$; then by taking $p(v) = u(v)-u(v-r)$ with $r \in [0,\infty)$, we have:
%\begin{small}
\allowdisplaybreaks
\begin{align}\label{TermDiffEq}
%F(r) =& \int_{0}^{r} {\rm d} F(v) \\
%=&\int_{0}^{r} {\rm d} \left[\sum\limits_{n=1}^{+\infty} \frac{(-1)^{n-1}}{n!}\int_{(\mathbb{R}^2)^n} \det (K(x_i,x_j))_{1 \leq i,j \leq n}\prod\limits_{i=1}^{n} u(v-|x_i|) {\rm d}x_1...{\rm d}x_n \right]\\
F_n(r) =&\int_{\mathbb{R}}p(v) {\rm d}F_n(v) \nonumber\\
\overset{(a)}{=}&\sum\limits_{k=1}^{n} \frac{(-1)^{k-1}}{k!}\int_{(\mathbb{R}^2)^k \times [0,r)}  \det (K(x_i,x_j))_{1 \leq i,j \leq k} {\rm d}\left[\prod\limits_{i=1}^{k} u(v-|x_i|)\right] {\rm d}x_1...{\rm d}x_k \nonumber\\
\overset{(b)}{=} &\sum\limits_{k=1}^{n} \frac{(-1)^{k-1}}{k!} \int_{(\mathbb{R}^2)^k \times [0,r)}  \det (K(x_i,x_j))_{1 \leq i,j \leq k} \sum\limits_{m=1}^{k} \prod\limits_{i=1, i \neq m}^{k} u(v-|x_i|) \delta_{|x_m|}( {\rm d}v) {\rm d}x_1...{\rm d}x_k \nonumber\\
\overset{(c)}{=} & \sum\limits_{k=1}^{n} \frac{(-1)^{k-1}}{k!} \int_{(\mathbb{R}^2)^k \times [0,r)} k \det (K(x_i,x_j))_{1 \leq i,j \leq k} \prod\limits_{i=2}^{k} u(v-|x_i|)\delta_{|x_1|}( {\rm d}v) {\rm d}x_1...{\rm d}x_k \nonumber \\
= &\sum\limits_{k=1}^{n} \frac{(-1)^{k-1}}{(k-1)!} \int_{0}^{+\infty}\int_{0}^{2\pi}\int_{(\mathbb{R}^2)^{k-1}}\int_{0}^{r} \det (K(x_i,x_j))_{1 \leq i,j \leq k} \bigg|_{x_1 = (r_1,\theta)} \nonumber \\
& \times \prod\limits_{i=2}^{k} u(v-|x_i|) r_1 \delta_{r_1}({\rm d}v) {\rm d}x_2...{\rm d}x_k {\rm d} \theta {\rm d}r_1 \nonumber\\
\overset{(d)}{=}& \int_{0}^{r} \sum\limits_{k=1}^{n} \frac{(-1)^{k-1} }{(k-1)!}  2\pi v \int_{(B(0,v))^{k-1}}  \det (K(x_i,x_j))_{1 \leq i,j \leq k} \bigg|_{x_1 = (v,0)} {\rm d}x_2...{\rm d}x_k {\rm d}v
% = &\int_{0}^{r} \sum\limits_{n=0}^{+\infty} \frac{(-1)^{n} }{n!} 2\pi v  \int_{(B(0,v))^{n}}  \det (K(x_i,x_j))_{0 \leq i,j \leq n} \bigg|_{x_0 = (v,0)} {\rm d}x_1...{\rm d}x_n{\rm d}v .
\end{align}
%\end{small}
Step (a) is derived by applying Lemma~\ref{SwapIntDiff} to $F_n(v)$ and $p(v)$. 
%since the integration w.r.t. $x_1, x_2 ,..., x_n$ is over $(\mathbb{R}^2)^n$, which is independent of $v$; and also 
Then (b) follows from the product rule for differentials, and the fact that the Dirac measure is the distributional derivative of the unit step function. Furthermore, (c) is because the determinant $\det(K(x_i,x_j))_{1 \leq i,j \leq n} $ remains the same if we swap the position of $x_1$ and $x_k$, which is equivalent to exchanging the first row and the $k$-th row, and then the first column and the $k$-th column of  $K(x_i,x_j)_{1 \leq i,j \leq n}$. Finally, (d) follows from the the defining property of Dirac measure, and noting that since $\Phi$ is stationary and isotropic, the integration is invariant w.r.t. the angle of $x_1$. Notice that $F_n(r)$ can be expressed as~(\ref{TermDiffEq}), which shows it is differentiable.

Given $r \in [0,\infty)$, we can check $F_n^{'}(v)$ converges uniformly for $v \in [0,r]$ using Hadamard's inequality for positive semi-definite matrices%, i.e.,  $\det (\left(K(x_{i},x_{j})\right)_{1\leq i,j\leq n}) \leq  \prod_{i=1}^{n} K(x_{i},x_{i})$
. Then by applying Lemma~\ref{RubinLemma} to $\{F_n\}$, we have:
\allowdisplaybreaks
\begin{align*}
F(r) &= \int_{0}^{r} \lim\limits_{n \rightarrow \infty} F_{n}^{'}(v) {\rm d}v \\
%&= \int_{0}^{r} \sum\limits_{n=1}^{+\infty} \frac{(-1)^{n-1} }{(n-1)!}  2\pi v \int_{(B(0,v))^{n-1}} \det (K(x_i,x_j))_{1 \leq i,j \leq n} \bigg|_{x_1 = (v,0)} {\rm d}x_2...{\rm d}x_n {\rm d}v\\
&= \int_{0}^{r} \sum\limits_{n=0}^{+\infty} \frac{(-1)^{n} }{n!} 2\pi v  \int_{(B(0,v))^{n}}\det (K(x_i,x_j))_{0 \leq i,j \leq n} \bigg|_{x_0 = (v,0)} {\rm d}x_1...{\rm d}x_n{\rm d}v.
\end{align*}

\section{Proof of Lemma~\ref{MeanIntNearestBSLemma}}~\label{MeanIntNearestBSAppdix}
Denote the empty space function as $F(r)$, then the mean interference is calculated as:
\allowdisplaybreaks
\begin{align*}
\mathbb{E}[I|x^*(0) = x_0] =& -\frac{{\rm d}}{{\rm d}s} \left[\mathbb{E}[\exp(-sI) | |x^*(0) = x_0]\right]  \bigg|_{s=0}\\
\overset{(a)}{=}& -\frac{1}{1-F(r_0)} \sum\limits_{n=0}^{+\infty} \frac{(-1)^n}{n!}\int_{(\mathbb{R}^2)^n} \det(K_{x_0}^{!}(x_i,x_j))_{1\leq i,j \leq n} \\
&\times \frac{{\rm d}}{{\rm d}s}\prod\limits_{i=1}^{n}[1- \frac{\mathbbm{1}_{|x_i| \geq r_0}}{1+sP\mathit{l}(x_i) }] {\rm d}x_1...{\rm d}x_n  \bigg|_{s=0}\\
\overset{(b)}{=}& -\frac{1}{1-F(r_0)} \sum\limits_{n=1}^{+\infty} \frac{(-1)^n}{n!}\int_{(\mathbb{R}^2)^n} \det(K_{x_0}^{!}(x_i,x_j))_{1\leq i,j \leq n} \\
&\times \sum\limits_{k=1}^{n} \prod\limits_{i=1,i\neq k}^{n} [1-\frac{\mathbbm{1}_{|x_i| \geq r_0}}{1+sP\mathit{l}(x_i) }] \frac{P\mathit{l}(x_k)\mathbbm{1}_{|x_k| \geq r_0}}{(1+sP\mathit{l}(x_k))^2} {\rm d}x_1...{\rm d}x_n  \bigg|_{s=0}\\
%= & -\frac{1}{1-F(r_0)} \sum\limits_{n=1}^{+\infty} \frac{(-1)^n}{n!}\int_{(\mathbb{R}^2)^n} \det(K_{x_0}^{!}(x_i,x_j))_{1\leq i,j \leq n}  \\
%& \times \sum\limits_{k=1}^{n} \prod\limits_{i=1,i\neq k}^{n} \mathbbm{1}_{|x_i| < r_0} \mathbbm{1}_{|x_k|\geq r_0} P\mathit{l}(x_k) {\rm d}x_1...{\rm d}x_n\\
\overset{(c)}{=} & \frac{\sum\limits_{n=1}^{+\infty} \frac{(-1)^{n-1}}{n!}\int_{(\mathbb{R}^2)^n}  \det(K_{x_0}^{!}(x_i,x_j))_{1\leq i,j \leq n} \times n \prod\limits_{i = 2}^{n} \mathbbm{1}_{|x_i| < r_0} \mathbbm{1}_{|x_1|\geq r_0}P\mathit{l}(x_1) {\rm d}x_1...{\rm d}x_n}{1-F(r_0)}\\ 
& = \frac{\sum\limits_{n=1}^{+\infty} \frac{(-1)^{n-1}}{(n-1)!}  \int_{(B(0,r_0))^{n-1}} \int_{B^{c}(0,r_0)} \det(K_{x_0}^{!} (x_i,x_j))_{1\leq i,j \leq n}P \mathit{l}(x_1) {\rm d}x_1...{\rm d}x_n}{\sum\limits_{n=0}^{+\infty}\frac{(-1)^n}{n!} \int_{B(0,r_0)^n} \det(K_{x_0}^{!}(x_i,x_j))_{1 \leq i,j \leq n}{\rm d}x_1...{\rm d}x_n}.
\end{align*}
Interchanging the infinite sum and the differentiation in (a) is guaranteed by Lemma~\ref{RubinLemma}.
%applying Lemma~\ref{RubinLemma} to the sequence of functions $\{I_n(s)\}$, where:
%\allowdisplaybreaks
%\begin{align*}
%I_n(s) =& \sum\limits_{k=1}^{n} \frac{(-1)^k}{k!}\int_{(\mathbb{R}^2)^k} \det(K_{x_0}^{!}(x_i,x_j))_{1\leq i,j \leq k} \\
%& \times \sum\limits_{m=1}^{k} \prod\limits_{i=1,i\neq m}^{k} [1-\frac{\mathbbm{1}_{|x_i| \geq r_0}}{1+sP\mathit{l}(x_i) }] \frac{P\mathit{l}(x_k)\mathbbm{1}_{|x_k| \geq r_0}}{(1+sP\mathit{l}(x_k) )^2} {\rm d}x_1...{\rm d}x_k.
%\end{align*}
Then (b) is derived by applying the derivative of product rule. In addition, (c) is true since consider $n$ points $x_1,...,x_n \in \mathbb{R}^2$ such that $|x_k| \geq r_0$ and the rest are  within the open ball $B^o(0,r_0)$, then the determinant $\det(K_{x_0}^{!}(x_i,x_j))_{1 \leq i,j \leq n} $ remains the same if we swap the position of $x_1$ and $x_k$. 
%Therefore, $\int_{(\mathbb{R}^2)^n} \det(K_{x_0}^{!}(x_i,x_j))_{1 \leq i,j \leq n} \prod\limits_{i=1,i\neq k}^{n} \mathbbm{1}_{|x_i|<r_0}\mathbbm{1}_{|x_k| \geq r_0} \mathit{l}(x_k) {\rm d}x_1...{\rm d}x_n$ has the same value irrespective of $k$. 

\bibliographystyle{ieeetr}
\bibliography{reference}

\begin{thebibliography}{10}

\bibitem{li2014fitting}
Y.~Li, F.~Baccelli, H.~Dhillon, and J.~Andrews, ``Fitting determinantal point
  processes to macro base station deployments,'' in {\em IEEE Global
  Communications Conference}, Dec. 2014.

\bibitem{trac}
J.~Andrews, F.~Baccelli, and R.~Ganti, ``A tractable approach to coverage and
  rate in cellular networks,'' {\em IEEE Transactions on Communications},
  vol.~59, pp.~3122--3134, Nov. 2011.

\bibitem{Harpreet}
H.~Dhillon, R.~Ganti, F.~Baccelli, and J.~Andrews, ``Modeling and analysis of
  {K}-tier downlink heterogeneous cellular networks,'' {\em IEEE Journal on
  Selected Areas in Communications}, vol.~30, pp.~550--560, Apr. 2012.

\bibitem{dhillon2013load}
H.~Dhillon, R.~Ganti, and J.~Andrews, ``Load-aware modeling and analysis of
  heterogeneous cellular networks,'' {\em IEEE Transactions on Wireless
  Communications}, vol.~12, pp.~1666--1677, Apr. 2013.

\bibitem{mukherjee2012distribution}
S.~Mukherjee, ``Distribution of downlink {SINR} in heterogeneous cellular
  networks,'' {\em IEEE Journal on Selected Areas in Communications}, vol.~30,
  pp.~575--585, Apr. 2012.

\bibitem{madhusudhanan2011multi}
P.~Madhusudhanan, J.~G. Restrepo, Y.~Liu, T.~X. Brown, and K.~R. Baker,
  ``Multi-tier network performance analysis using a shotgun cellular system,''
  in {\em IEEE Global Communications Conference}, pp.~1--6, Dec. 2011.

\bibitem{elsawy2013stochastic}
H.~ElSawy, E.~Hossain, and M.~Haenggi, ``Stochastic geometry for modeling,
  analysis, and design of multi-tier and cognitive cellular wireless networks:
  A survey,'' {\em IEEE Communications Surveys and Tutorials}, vol.~15, no.~3,
  pp.~996--1019, 2013.

\bibitem{heath2013hetnets}
R.~Heath, M.~Kountouris, and T.~Bai, ``Modeling heterogeneous network
  interference using {P}oisson point processes,'' {\em IEEE Transactions on
  Signal Processing}, vol.~61, pp.~4114--4126, Aug. 2013.

\bibitem{madhusudhanan2013modeling}
P.~Madhusudhanan, X.~Li, Y.~Liu, and T.~Brown, ``Stochastic geometric modeling
  and interference analysis for massive {MIMO} systems,'' in {\em International
  Symposium on Modeling Optimization in Mobile, Ad Hoc Wireless Networks
  (WiOpt)}, pp.~15--22, May 2013.

\bibitem{dhillon2013downlink}
H.~Dhillon, M.~Kountouris, and J.~Andrews, ``Downlink {MIMO} {H}et{N}ets:
  modeling, ordering results and performance analysis,'' {\em IEEE Transactions
  on Wireless Communications}, vol.~12, pp.~5208--5222, Oct. 2013.

\bibitem{kountouris2013antennas}
M.~Kountouris and N.~Pappas, ``Het{N}ets and massive {MIMO}: modeling,
  potential gains, and performance analysis,'' in {\em IEEE-APS Topical
  Conference on Antennas and Propagation in Wireless Communications (APWC)},
  pp.~1319--1322, Sept. 2013.

\bibitem{zeros}
J.~B. Hough, M.~Krishnapur, Y.~Peres, and B.~Vir{\'{a}}g, {\em Zeros of
  {G}aussian analytic functions and determinantal point processes}, vol.~51 of
  {\em University Lecture Series}.
\newblock American Mathematical Society, 2009.

\bibitem{baccelli1997stochastic}
F.~Baccelli, M.~Klein, M.~Lebourges, and S.~Zuyev, ``Stochastic geometry and
  architecture of communication networks,'' {\em Telecommunication Systems},
  vol.~7, no.~1-3, pp.~209--227, 1997.

\bibitem{brown2000shotgun}
T.~Brown, ``Cellular performance bounds via shotgun cellular systems,'' {\em
  IEEE Journal on Selected Areas in Communications}, vol.~18, pp.~2443--2455,
  Nov. 2000.

\bibitem{baccelli2001coverage}
F.~Baccelli and B.~B{\l}aszczyszyn, ``On a coverage process ranging from the
  {B}oolean model to the {P}oisson-{V}oronoi tessellation with applications to
  wireless communications,'' {\em Advances in Applied Probability}, vol.~33,
  no.~2, pp.~293--323, 2001.

\bibitem{blaszczyszyn2014wireless}
B.~Blaszczyszyn, M.~K. Karray, and H.~P. Keeler, ``Wireless networks appear
  {P}oissonian due to strong shadowing,'' {\em arXiv preprint arXiv:1409.4739},
  2014.

\bibitem{pairmodel}
D.~Taylor, H.~Dhillon, T.~Novlan, and J.~Andrews, ``Pairwise interaction
  processes for modeling cellular network topology,'' in {\em IEEE Global
  Communications Conference}, pp.~4524--4529, Dec. 2012.

\bibitem{anjin}
A.~Guo and M.~Haenggi, ``Spatial stochastic models and metrics for the
  structure of base stations in cellular networks,'' {\em IEEE Transactions on
  Wireless Communications}, vol.~12, pp.~5800--5812, Nov. 2013.

\bibitem{riihijarvi2010modeling}
J.~Riihijarvi and P.~Mahonen, ``Modeling spatial structure of wireless
  communication networks,'' in {\em IEEE Conference on Computer Communications
  Workshops}, pp.~1--6, Mar. 2010.

\bibitem{FME}
R.~Ganti, F.~Baccelli, and J.~Andrews, ``Series expansion for interference in
  wireless networks,'' {\em IEEE Transactions on Information Theory}, vol.~58,
  pp.~2194--2205, Apr. 2012.

\bibitem{rubak}
F.~Lavancier, J.~M{\o}ller, and E.~Rubak, ``Statistical aspects of
  determinantal point processes,'' {\em arXiv preprint arXiv:1205.4818}, 2012.

\bibitem{decreusefond2013perfect}
L.~Decreusefond, I.~Flint, and K.~C. Low, ``Perfect simulation of determinantal
  point processes,'' {\em arXiv preprint arXiv:1311.1027}, 2013.

\bibitem{decreusefond14ginibre}
L.~Decreusefond, I.~Flint, and A.~Vergne, ``Efficient simulation of the
  {G}inibre point process,'' {\em arXiv preprint arXiv:1310.0800}, 2014.

\bibitem{shirai2003random}
T.~Shirai and Y.~Takahashi, ``Random point fields associated with certain
  fredholm determinants {I}: fermion, poisson and boson point processes,'' {\em
  Journal of Functional Analysis}, vol.~205, no.~2, pp.~414--463, 2003.

\bibitem{miyoshi2012cellular}
N.~Miyoshi and T.~Shirai, ``A cellular network model with {G}inibre
  configurated base stations,'' {\em Research Rep. on Math. and Comp. Sciences
  (Tokyo Inst. of Tech.)}, Oct. 2012.

\bibitem{nakata2013spatial}
I.~Nakata and N.~Miyoshi, ``Spatial stochastic models for analysis of
  heterogeneous cellular networks with repulsively deployed base stations,''
  {\em Research Rep. on Math. and Comp. Sciences, B-473 (Tokyo Inst. of
  Tech.)}, 2013.

\bibitem{deng14ginibre}
N.~Deng, W.~Zhou, and M.~Haenggi, ``The {G}inibre point process as a model for
  wireless networks with repulsion,'' {\em arXiv preprint arXiv:1401.3677},
  2014.

\bibitem{stochgeom}
M.~Haenggi, {\em Stochastic geometry for wireless networks}.
\newblock Cambridge University Press, 2013.

\bibitem{chiu2013stochastic}
S.~N. Chiu, D.~Stoyan, W.~S. Kendall, and J.~Mecke, {\em Stochastic geometry
  and its applications}.
\newblock John Wiley \& Sons, 2013.

\bibitem{kuo2005lifting}
F.~Y. Kuo and I.~H. Sloan, ``Lifting the curse of dimensionality,'' {\em
  Notices of the AMS}, vol.~52, no.~11, pp.~1320--1328, 2005.

\bibitem{blaszczyszyn2014studying}
B.~Blaszczyszyn and H.~P. Keeler, ``Studying the {SINR} process of the typical
  user in {P}oisson networks by using its factorial moment measures,'' {\em
  arXiv preprint arXiv:1401.4005}, 2014.

\bibitem{integrate}
A.~Jeffrey and D.~Zwillinger, {\em Table of integrals, series, and products}.
\newblock Academic Press, 2007.

\bibitem{ipsen2011determinant}
I.~C. Ipsen and D.~J. Lee, ``Determinant approximations,'' {\em arXiv preprint
  arXiv:1105.0437}, 2011.

\bibitem{R}
A.~Baddeley and R.~Turner, ``Spatstat: an {R} package for analyzing spatial
  point patterns,'' {\em Journal of Statistical Software}, vol.~12, no.~6,
  pp.~1--42, 2005.

\bibitem{rudin1976principles}
W.~Rudin, {\em Principles of mathematical analysis, 3rd edition}.
\newblock McGraw-Hill New York, 1976.

\end{thebibliography}

\end{document}